\documentclass{article}

\usepackage[T1]{fontenc}
\usepackage[utf8]{inputenc}
\usepackage{amsmath}
\usepackage{amsfonts}
\usepackage{amssymb}
\usepackage{amsthm}
\usepackage{amsbsy}
\usepackage{mathrsfs}
\usepackage{bm}
\usepackage{microtype}
\usepackage{physics}
\usepackage{graphicx}
\usepackage{color}
\usepackage{enumitem}
\usepackage[lmargin=0.9in, rmargin=0.9in, tmargin=1in, bmargin=1in]{geometry}
\usepackage{url}
\usepackage[hidelinks]{hyperref}
\hypersetup{
  pdftitle={Integration against POVMs in von Neumann algebras and operational quantum reference frames},
  pdfauthor={Jan Glowacki and Yui Kuramochi}
}
\usepackage{cleveref}

\newtheorem{theorem}{Theorem}[section]
\newtheorem{lemma}[theorem]{Lemma}
\newtheorem{corollary}[theorem]{Corollary}
\newtheorem{proposition}[theorem]{Proposition}
\newtheorem{definition}[theorem]{Definition}
\newtheorem*{definition*}{Definition}

\newtheorem{remark}[theorem]{Remark}

\newcommand{\hi}{\mathcal{H}}
\newcommand{\his}{\mathcal{H}_{\mathcal{S}}}
\newcommand{\hir}{\mathcal{H}_{\mathcal{R}}}
\newcommand{\hisr}{\mathcal{H}_{\mathcal{S}} \otimes \mathcal{H}_{\mathcal{R}}}
\newcommand{\hik}{\mathcal{K}}

\renewcommand{\ip}[2]{\left\langle\,#1\,|\,#2\,\right\rangle}

\newcommand{\M}{\mathcal{M}}
\newcommand{\R}{\mathcal{R}}
\renewcommand{\S}{\mathcal{S}}

\newcommand{\F}{\mathcal{F}}
\newcommand{\N}{\mathcal{N}}
\newcommand{\T}{\mathcal{T}}
\newcommand{\Eff}{\mathcal{E}}

\newcommand{\Cn}{\mathbb{C}}
\newcommand{\Reals}{\mathbb{R}}
\newcommand{\Nn}{\mathbb{N}}
\newcommand{\id}{\boldsymbol{1}}
\newcommand{\algunit}{1}
\newcommand{\zero}{\boldsymbol{0}}

\newcommand{\E}{\mathsf{E}}

\newcommand{\MS}{\M_\S}
\newcommand{\MR}{\M_\R}
\newcommand{\MSR}{\M_\S \bar{\otimes} \M_\R}
\newcommand{\MSp}{(\M_\S)_*}
\newcommand{\MRp}{(\M_\R)_*}
\newcommand{\MSRp}{(\MSR)_*}
\newcommand{\Mp}{\M_*}

\newcommand{\Bb}{B_b(\Sigma,\F,\MS)}
\newcommand{\BbBH}{B_b(\Sigma,\F,B(\his))}
\newcommand{\Bmes}{\mathcal{B}(\Sigma,\F,\MS)}
\newcommand{\Linf}{L^\infty_\E(\Sigma,\MS)}

\begin{document}

\newgeometry{top=2.5cm, left=2.5cm, right=2.5cm, bottom=2.3cm}
\thispagestyle{empty}
\begingroup
\setlength{\parindent}{1em}
\begin{center}
{\fontsize{17}{21}\selectfont\bfseries
Integration against POVMs in von Neumann algebras\\[2pt]
and operational quantum reference frames\par}
\vspace{18pt}
{\fontsize{12}{15}\selectfont
Jan G\l{}owacki\textsuperscript{1,2,3,*}\qquad
Yui Kuramochi\textsuperscript{4}\par}
\vspace{10pt}
{\fontsize{9}{12}\selectfont
\textsuperscript{1}Basic Research Community for Physics, Leipzig, GERMANY\\
\textsuperscript{2}Institute for Quantum Optics and Quantum Information, Vienna, AUSTRIA\\
\textsuperscript{3}Department of Computer Science, University of Oxford, Oxford, UK\\
\textsuperscript{4}Department of Informatics, Kyushu University, Fukuoka, JAPAN\par}
\vspace{6pt}
{\footnotesize\textsuperscript{*}Corresponding author:
\href{mailto:jan.glowacki@oeaw.ac.at}{\texttt{jan.glowacki@oeaw.ac.at}}\par}
\end{center}
\vspace{18pt}

\begin{abstract}
\vspace{5pt}We construct integrals of uniformly bounded ultraweakly measurable functions against positive operator-valued measures (POVMs) in von Neumann algebras, assuming that the integrand algebra has separable predual. The integral lies in the spatial tensor product and is uniquely characterized by its evaluations on product normal states. Integration is unital and completely positive and descends faithfully to the quotient by POVM-null functions. When both preduals are separable, this quotient is a von Neumann algebra and integration is normal, agreeing with the amplification of scalar POVM integration. We establish a dilation factorization, a Leibniz rule and a Fubini theorem, and specify the additional joint-normality hypothesis needed for realization in an ambient algebra.

As an application, under the stated separability and joint-normality assumptions, we extend operational quantum reference frame relativization to covariant POVMs on arbitrary measurable spaces carrying a group action. Its domain is the fixed-point von Neumann algebra of system-valued functions modulo null equality. For homogeneous spaces of second countable locally compact Hausdorff groups with closed stabilizers, this domain is identified with the stabilizer-fixed system algebra. For sharp frames, relativization is a normal unital $*$-homomorphism, faithful when ambient multiplication is injective. For localizable frames on metrizable spaces, restriction recovers continuous integrands at supported configurations. On trivial principal bundles, the domain is identified with system-valued fields on the base. These results provide integration machinery for forthcoming applications to relativistic gauge theories.

\end{abstract}
\vspace{24pt}

\vspace{5pt}

\vspace{24pt}

{\fontsize{9}{12}\selectfont
\noindent\textbf{Acknowledgments and funding}\par\vspace{5pt}
\noindent
This publication was made possible through the support of the ID\#~62312 grant from the John Templeton Foundation, as part of the \href{https://www.templeton.org/grant/the-quantuminformation-structure-ofspacetime-qiss-second-phase}{`The Quantum Information Structure of Spacetime' Project (QISS)}. The opinions expressed in this publication are those of the authors and do not necessarily reflect the views of the John Templeton Foundation. This research was funded in whole or in part by the Austrian Science Fund (FWF) 10.55776/PAT1562525. For open-access purposes, the authors have applied a CC~BY public copyright license to any author-accepted manuscript version arising from this submission.
\par}

\endgroup
\clearpage
\small
\restoregeometry
\thispagestyle{empty}

\noindent{\large\bfseries Contents}
\vspace{4pt}
\begingroup\normalsize
\makeatletter
\@starttoc{toc}
\makeatother
\endgroup

\section{Introduction}
An effect evaluated on a quantum system may depend on a parameter that is itself described by a quantum measurement. If $f(x)$ is a system effect associated with outcome $x$ of a positive operator-valued measure (POVM) $\E$ on another system, the corresponding joint effect is obtained by integrating $f$ against $\E$. This paper constructs such integrals for general $\mathrm{W}^*$-algebras and applies them to operational quantum reference frames on measurable spaces carrying group actions. Such integral operators are needed for the definition of the relativization map which plays the central role in the operational approach to quantum reference frames (QRFs): a frame carries a covariant POVM, and system observables are expressed relative to it by integrating symmetry-dependent operators \cite{carette_operational_2025,loveridge_symmetry_2018}. This work extends relativization maps beyond transitive group actions and into the realm of general $\mathrm{W}^*$-algebras. The central motivation for this extension is the application of this framework in the context of relativistic gauge theories, where reference configurations are points in principal bundles.

\paragraph*{Integration maps.}
To provide the analytic basis for these applications, we first study integration independently of any symmetry action. Let $\MS$ and $\MR$ be $\mathrm{W}^*$-algebras, with $\MSp$ separable, and let $\E:\F\to\Eff(\MR)$ be a normalized POVM on a measurable space $(\Sigma,\F)$. For every uniformly bounded ultraweakly measurable function $f\in\Bb$, we construct a unique integral in the spatial tensor product, characterized by
\[
 (\hat\rho\otimes\hat\omega)\left(\int_\Sigma f\otimes d\E\right)
 =\int_\Sigma \hat\rho(f(x))\,d(\hat\omega \circ\E)(x),
 \qquad \rho\in\S(\MS),\quad\omega\in\S(\MR).
\]
Taking either algebra to be $\Cn$ recovers, respectively, scalar integration against a POVM or weak-star integration of operator-valued functions against a scalar probability measure. Taking both algebras to be $\Cn$ gives ordinary scalar Lebesgue integration. The construction makes no topological assumptions on $\Sigma$. The universal domain $\Bb$ is a $\mathrm{C}^*$-algebra. On this domain, integration is unital, completely positive and functorial under measurable changes of variables and normal channels. For a fixed POVM, integration descends to the quotient $\Linf=\Bb/\N_\E$ by the ideal of $\E$-null functions. The descended map is faithful; it is isometric when $\E$ has the norm-$1$ property (is \emph{localizable}), and an isometric $*$-homomorphism when $\E$ is a PVM. If $\MRp$ is also separable, the quotient is a von Neumann algebra and
\[
 \Linf\cong\MS\bar\otimes L^\infty_\E(\Sigma),
 \qquad
 \int_\Sigma d\E \cong \id_{\MS}\bar\otimes\Phi_\E,
\]
under this identification, where $\Phi_\E$ is scalar POVM integration. The descended map is then normal. We also establish a Stinespring factorization through Naimark dilation and regularity under parameters, differentiation under the integral sign, and a Fubini theorem. For subalgebras $\MS,\MR\subseteq\M$, an ambient-algebra realization requires the additional hypothesis that multiplication extend to a normal $*$-homomorphism $\mu:\MSR\to\MS\vee\MR$. We call such a pair \emph{jointly normal}; split pairs are examples. Under this hypothesis,
\[
 \int_\Sigma f\,d\E:=\mu\left(\int_\Sigma f\otimes d\E\right)
\]
provides the integral operators for the case when the two algebras are jointly normal subalgebras of an ambient one.

\paragraph*{Application to quantum reference frames.}
Covariance of a frame POVM makes integration intertwine the group actions on functions and operators on the tensor product algebra. Restriction to the fixed-point function algebra, consisting of $\E_\R$-equivalence classes of uniformly bounded ultraweakly measurable $\MS$-valued functions that are equivariant almost everywhere for each fixed group element, therefore defines a normal unital CP map into the invariant joint algebra. For sharp frames, this map is a normal $*$-homomorphism, and it is faithful when the multiplication map $\mu$ is injective. Under the stated separability and joint-normality assumptions, this construction requires neither transitivity nor local compactness of the group, and the frame algebra need not be type I, or even a factor. We show that it agrees with previously developed relativization maps in their common setting and define the generalized map by
\[
	\yen^\R := \int_\Sigma d\E_\R: L^\infty_{\E_\R}(\Sigma, \MS)^G \to (\MS \vee \MR)^G.
\]

For a localizable frame on a metrizable configuration space, a localization theorem recovers the value of an integrand that is ultraweakly continuous at a point in the support of the frame observable. We also treat general trivial principal bundles, identifying the relativization domain with the algebra of system-valued fields on the base space. Relative states, frame transformations and the principal bundle gauge formulation will be developed in forthcoming work, using the results presented here.

\paragraph*{Relation to earlier work.}
Bochner and Pettis integration concern functions with values in Banach spaces, integrated against scalar measures \cite{diestel_vector_1977}. Scalar POVM integration of complex-valued functions and its channel interpretation are recalled in Section~\ref{sec:monoidal}, including the comparison with \cite{kuramochi_quantum_2018}.

Integration against vector measures was developed by Bartle, Dunford and Schwartz \cite{bartle_weak_1955}; bilinear integration permits more general pairings of integrand and measure values \cite{bartle_general_1956,diestel_vector_1977}. Dobrakov's operator-valued measure integral \cite{dobrakov_integration_1970} uses the evaluation pairing $X\times L(X,Y)\to Y$. Here the pairing instead sends the integrand and POVM values to their tensor product.

Jefferies and Okada's bilinear integration in Banach space tensor products \cite{jefferies_bilinear_2010,jefferies_bilinear_1998,jefferies_semivariation_2005} is a close precursor. The present construction uses the spatial von Neumann tensor product and positivity of the measure to integrate every uniformly bounded ultraweakly measurable function when the integrand algebra has separable predual. The resulting completely positive maps, null quotients and normality properties connect this integration theory to operator-algebraic quantum theory.

Farenick, Plosker and Smith \cite{farenick_classical_2011} develop an integral of matrix-valued functions against matrix-valued POVMs using a symmetrized Radon--Nikodym formula; Farenick and Kozdron \cite{farenick_conditional_2012} develop its conditional-expectation theory. Their construction returns an element of the same matrix algebra, whereas ours retains the system and frame factors in a spatial tensor product and allows general von Neumann algebras.

Homogeneous-space relativization is studied in \cite{glowacki_quantum_2024}. Fewster, Janssen, Loveridge, Rejzner and Waldron \cite[Sec.~4.4]{fewster_quantum_2024} establish relativization for $G/H$ with compact $H$ in an operator-algebraic QFT setting, already allowing general system von Neumann algebras and using a crossed-product construction. Our relativization map extends these constructions to measurable $G$-spaces and general frame von Neumann algebras; in the homogeneous-space case, compactness of the stabilizer is not required.

A related geometric formulation of quantum reference frames is developed by Vanzella and Butterfield \cite{vanzella_frame_2024}, who describe them by amplitudes on frame bundles. Their squared moduli weight averages of observables that depend on an observer's tetrad; the authors also explore an extension to superposed geometries. In the trivial-bundle setting of Section~\ref{ex:QRFbundle}, for bounded equivariant fields and globally square-integrable amplitudes, their setup is recovered by restricting a relative observable to a pure state of a sharp frame.

A complementary extension of operational QRFs to continuous groupoids is developed by Fedida, Ibort and Mas-Dorca, with applications to the Atiyah groupoids of frame and gauge bundles. Their relativization integrates transported operator fields against covariant POVMs on the groupoid's target fibres, explicitly using the integration theory developed here \cite[Prop.~4.24]{fedida_groupoidal_2026}. Our application retains group actions on configuration spaces and fixed system and frame algebras, while theirs uses groupoid representations on Hilbert bundles; the precise relation between the resulting principal-bundle formulations remains to be established.

Brady develops a complementary formulation of classical and quantum reference frames using Hilbert modules and equivariant \(C^*\)-correspondences, encompassing operational Hilbert-space relativization \cite{brady}. The present work constructs normal relativization maps explicitly by POVM integration on algebras of measurable equivariant fields. A detailed comparison between these constructions remains to be established.

Closely related work covering, among other things, relativization on general homogeneous spaces in the $\mathrm{W}^*$-algebraic setting and constructing concrete examples of such quantum reference frames is being developed by Loveridge and Waldron in their forthcoming paper, \emph{$\mathrm{W}^*$-algebraic Quantum Reference Frames} (in preparation).

\paragraph*{Organization.}
Section~\ref{sec:prelim} gives the preliminaries. Sections~\ref{sec:generalities}--\ref{sec:propint} construct integral operators, identify their domains and establish the integration maps' properties. Sections~\ref{sec:monoidal} and~\ref{sec:CP} give the tensor-product characterization and dilation factorization. Sections~\ref{sec:commuting} and~\ref{sec:parametrized} treat ambient-algebra realizations and parametrized integration. Section~\ref{sec:intsummary} collects the integration results for reference. Section~\ref{sec:QRFapplication} develops the application to quantum reference frames.

\section{Preliminaries}\label{sec:prelim}

\paragraph*{$\mathrm{W}^*$-algebras and preduals.}
A $\mathrm{W}^*$-algebra is a unital $\mathrm{C}^*$-algebra $\M$ that is isometrically isomorphic to the Banach dual of some Banach space. Such a predualizing Banach space, called the \emph{predual} and denoted $\Mp$, is unique up to isometric isomorphism \cite[Cor.~1.13.3]{sakai_Cstar_1971}. Sakai's theorem \cite[Thm.~1.16.7]{sakai_Cstar_1971} asserts that every $\mathrm{W}^*$-algebra is isometrically $*$-isomorphic to a unital $*$-subalgebra of $B(\hi)$, for some Hilbert space $\hi$, that is closed in the ultraweak topology; neither $\hi$ nor this concrete realization is unique. Elements of $\Mp$ are the \emph{normal} (i.e., ultraweakly continuous, see below) linear functionals on $\M$; the normal \emph{states} are
\[
\S(\M) := \{\omega \in (\Mp)_+ \; | \; \hat\omega(\algunit_\M) = 1\},
\]
where $(\Mp)_+$ denotes the positive cone of the predual. We distinguish a predual element $\omega$ from its functional realization $\hat\omega(a):=\langle a,\omega\rangle$. By the Jordan decomposition of normal functionals, the set $\S(\M)$ is $\Cn$-linearly dense in $\Mp$. We denote the self-adjoint part of $\M$ by $\M^{\rm sa}$, the positive cone by $\M_+$, the unit ball by $\M_1 := \{a \in \M : \|a\| \leq 1\}$, and the set of \emph{effects} by
\[
\Eff(\M) := \{e \in \M^{\rm sa} \;| \;\zero \leq e \leq \algunit_\M\}.
\]
For the concrete case $\M = B(\hi)$, the predual is the space of trace-class operators $\T(\hi)$, with the pairing $\omega(A) = \tr[\omega A]$.

We use lowercase Latin letters for abstract algebra elements and uppercase Latin letters for operators on Hilbert spaces. We write $\algunit_\M$ for the algebraic unit and $\id_\M$ for the identity channel on $\M$. For a Hilbert space $\hi$, we write $\Eff(\hi):=\Eff(B(\hi))$.

\paragraph*{Bimodule structure of the predual.}
The predual $\Mp$ carries a natural Banach $\M$-bimodule structure: for $\omega \in \Mp$ and $a \in \M$, the functionals $a \cdot \omega$ and $\omega \cdot a$ defined by
\begin{equation}\label{eq:bimodule}
    (a \cdot \omega)(x) := \omega(xa), \; (\omega \cdot a)(x) := \omega(ax) \; (x \in \M)
\end{equation}
are again normal functionals in $\Mp$, with $\|a \cdot \omega\| \leq \|a\| \cdot \|\omega\|$ and $\|\omega \cdot a\| \leq \|a\| \cdot \|\omega\|$.

\paragraph*{Topologies.}
The \emph{ultraweak} (or $\sigma$-\emph{weak}) topology on $\M$ is $\sigma(\M, \Mp)$, i.e., $a_\alpha \to a$ ultraweakly iff $\omega(a_\alpha) \to \omega(a)$ for all $\omega \in \Mp$. By the Banach--Alaoglu theorem, $\M_1$ is ultraweakly compact. If $\Mp$ is separable, the ultraweak topology on $\M_1$ is metrizable: given a countable norm-dense subset $(\omega_k)_{k \in \Nn}$ of $\Mp$, the formula
\begin{equation}\label{eq:metric}
    d(a,b) := \sum_{k=1}^\infty 2^{-k}\frac{|\omega_k(a-b)|}{1+|\omega_k(a-b)|}
\end{equation}
defines a metric on $\M_1$ inducing the ultraweak topology.

\paragraph*{Faithful representations.}
A $\mathrm{W}^*$-algebra $\M$ with separable predual admits a faithful normal $*$-representation $\pi: \M \hookrightarrow B(\hi)$ on a \emph{separable} Hilbert space.\footnote{Since $\Mp$ is separable, $\S(\M)$ contains a countable norm-dense subset $\{\omega_n\}$. The direct sum of the GNS representations~$\bigoplus_n \pi_{\omega_n}$ is faithful (since the $\omega_n$ are dense, they separate points of $\M$) and normal, and acts on the separable Hilbert space~$\bigoplus_n \hi_{\omega_n}$.} In any concrete realization $\pi(\M) \subseteq B(\hi)$, every normal functional $\omega \in \Mp$ can be represented as $\omega(a) = \tr[T_\omega \pi(a)]$ for some $T_\omega \in \T(\hi)$ (not uniquely in general), with $\|\omega\| = \inf\{\|T\|_1 : \omega(\cdot) = \tr[T \,\pi(\cdot)]\}$; for positive $\omega$, the representative can be chosen positive and $\|\omega\| = \omega(\algunit_\M)$. Faithful normal representations are used throughout the paper as proof tools, but never appear in the hypotheses of the main results.

\paragraph*{Spatial tensor products.}
Given $\mathrm{W}^*$-algebras $\MS$ and $\MR$, choose faithful normal representations $\pi_\S: \MS \hookrightarrow B(\his)$ and $\pi_\R: \MR \hookrightarrow B(\hir)$. The \emph{spatial} (or \emph{normal}) \emph{tensor product} is the $\mathrm{W}^*$-algebra
\[
    \MSR := \overline{\pi_\S(\MS) \otimes_{\rm alg} \pi_\R(\MR)}^{\rm uw} \subset B(\hisr),
\]
i.e., the ultraweak closure of the algebraic tensor product \cite[Sec.~IV.5]{takesaki2001theory}. Up to canonical $*$-isomorphism, this is independent of the choice of faithful normal representations \cite[Cor.~IV.5.3]{takesaki2001theory}; we therefore write $\MSR$ without reference to $\pi_\S, \pi_\R$.

For $\rho \in \MSp$ and $\omega \in \MRp$, the map $(\rho \otimes \omega)(a \otimes b) := \rho(a)\omega(b)$ extends uniquely by ultraweak continuity to a normal functional $\rho \otimes \omega \in \MSRp$. We record the density of product functionals.

\begin{lemma}\label{lem:proddense}
The $\Cn$-linear span of the product normal states $\S(\MSR)_{\rm prod} := \{\rho \otimes \omega : \rho \in \S(\MS),\, \omega \in \S(\MR)\}$ is norm-dense in $\MSRp$.
\end{lemma}
\begin{proof}
Choose faithful normal representations $\pi_\S, \pi_\R$ and identify $\MSR$ with a $\mathrm{W}^*$-subalgebra of $B(\hisr)$. The predual is $\MSRp \cong \T(\hisr)/(\MSR)^\perp$, where $(\MSR)^\perp := \{T \in \T(\hisr) : \tr[TA] = 0 \;\forall A \in \MSR\}$ is the preannihilator.\footnote{Equivalently, $T \sim T'$ iff $\tr[TA]=\tr[T'A]$ for all $A \in \MSR$.} The quotient map $r: \T(\hisr) \twoheadrightarrow \MSRp$ is a norm-continuous surjection. Now consider the rank-one operators $\dyad{\xi' \otimes \eta'}{\xi \otimes \eta}$. They span a dense subspace of $\T(\hisr)$ in trace norm and their images under $r$~are product normal functionals: $\tr[\dyad{\xi' \otimes \eta'}{\xi \otimes \eta} \, A \otimes B] = \braket{\xi}{A\xi'}\braket{\eta}{B\eta'}$. Since, by the Jordan decomposition, every product normal functional is a finite $\Cn$-linear combination of product normal states, the claim follows.
\end{proof}

\paragraph*{Testing against normal states.}
We recall the state-space facts used to characterize integral operators and verify their properties.

\begin{lemma}\label{lem:extfromstates}
Normal states determine elements of a $\mathrm{W}^*$-algebra $\M$ uniquely. For $a=a^*\in\M$,
\[
    \|a\|=\sup_{\omega\in\S(\M)}|\omega(a)|,
    \qquad a\geq0\ \Longleftrightarrow\ \omega(a)\geq0
    \text{ for every }\omega\in\S(\M).
\]
Product normal states determine elements of $\MSR$ uniquely. Consequently, equality of operators, adjoint relations, and equality to the identity can be checked against product normal states.\footnote{Product states do not generate the positive cone of $\MSRp$ in general: the flip operator in $M_2(\Cn) \bar{\otimes} M_2(\Cn)$ has non-negative expectation on all product states but is not positive.}
\end{lemma}

\begin{proof}
The linear span of normal states is the predual, so they separate elements of $\M$. In a faithful normal representation, unit vector functionals restrict to normal states; the norm and positivity assertions therefore follow from the corresponding quadratic-form criteria for self-adjoint operators. Product-state determination follows from Lemma~\ref{lem:proddense}.
\end{proof}

\paragraph*{Normal maps.}
A linear map $\Phi: \M \to \N$ between $\mathrm{W}^*$-algebras is \emph{normal} if it is continuous with respect to the ultraweak topologies or, equivalently, if it possesses a predual map $\Phi_*: \N_* \to \Mp$ satisfying $\omega(\Phi(a)) = (\Phi_*\omega)(a)$ for all $a \in \M$, $\omega \in \N_*$. It is \emph{unital} if $\Phi(\algunit_\M) = \algunit_\N$, \emph{positive} if $\Phi(\M_+) \subseteq \N_+$, $n$-\emph{positive} if $\id_n \otimes \Phi: M_n(\Cn) \otimes \M \to M_n(\Cn) \otimes \N$ is positive, and \emph{completely positive (CP)} if it is $n$-positive for all $n \in \Nn$. Normal unital CP maps are called \emph{channels}. Normal states $\omega \in \S(\M)$ are precisely the channels $\M \to \Cn$.

\paragraph*{Joint normality and multiplication maps.}\label{par:jointnormality}
Let $\M_1,\ldots,\M_n\subseteq\M$ be unital $\mathrm{W}^*$-subalgebras. We call the family \emph{jointly normal}, relative to the spatial tensor product, if the algebraic multiplication map
\[
    \mu_0:\M_1\otimes_{\rm alg}\cdots\otimes_{\rm alg}\M_n\longrightarrow\M,
    \qquad a_1\otimes\cdots\otimes a_n\longmapsto a_1\cdots a_n,
\]
extends to a normal unital $*$-homomorphism
\[
    \mu:\M_1\bar\otimes\cdots\bar\otimes\M_n
    \longrightarrow\M_1\vee\cdots\vee\M_n.
\]
Joint normality necessarily implies mutual commutativity: elements of distinct tensor factors commute, and their images under $\mu$ therefore commute. Conversely, commutativity makes $\mu_0$ a $*$-homomorphism, but does not by itself guarantee its normal extendibility to the spatial tensor product.\footnote{For example, take $\M_1=\M_2=\M=L^\infty([0,1])$. For nested dyadic partitions $\mathcal P_k$, the projections $p_k=\sum_{I\in\mathcal P_k}\chi_I\otimes\chi_I$ decrease to zero in $L^\infty([0,1]^2)$, while $\mu_0(p_k)=1$. Moreover, joint normality need not imply injectivity of $\mu$: multiplication for two identical copies of $\Cn^2$ annihilates $e_1\otimes e_2$. If all the algebras are factors, however, a normal unital $*$-homomorphism from their spatial tensor product is necessarily injective, since its kernel is an ultraweakly closed ideal in a factor.}
When it exists, $\mu$ is unique by ultraweak density of the algebraic tensor product and is surjective because its range is an ultraweakly closed algebra containing every $\M_i$. Whenever integration into an ambient algebra is considered below, joint normality of the relevant family is assumed.\footnote{For an inclusion $N\subseteq P$ with separable preduals and $P$ represented in standard form, this condition for the commuting pair $N,P'$ is the \emph{quasi-split property} of Fidaleo~\cite[Def.~3]{fidaleo_split_1997}.}

A sufficient condition for joint normality of a pair in a concrete representation on $\hi$ is the \emph{split property}: the existence of a type~I factor $\mathcal F\subseteq B(\hi)$ such that
\[
    \MS\subseteq\mathcal F\subseteq\MR'.
\]
The resulting tensor decomposition of $\hi$ gives a normal $*$-isomorphism
\[
    \MS\bar\otimes\MR\cong\MS\vee\MR,
    \qquad A\otimes B\longmapsto AB
\]
\cite{doplicher_standard_1984,summers_independence_1990}. Thus split pairs are jointly normal and their multiplication maps are injective.

\paragraph*{Positive operator-valued measures.}
Given a measurable space $(\Sigma,\F)$ and a $\mathrm{W}^*$-algebra
$\M$, a \emph{normalized positive operator-valued measure (POVM)}
with values in $\M$ is a map $\E:\F\to\Eff(\M)$ such that, for every
normal state $\omega\in\S(\M)$, the set function
\[
    \E_\omega(X):=\omega(\E(X)), \qquad X\in\F,
\]
is a probability measure. Equivalently, $\E(\emptyset)=\zero$,
$\E(\Sigma)=\algunit_\M$, and
\[
    \E\Bigl(\bigcup_n X_n\Bigr)=\sum_n\E(X_n)
\]
for every sequence of pairwise disjoint sets $X_n\in\F$, with
ultraweak convergence of the series. All POVMs considered below
are normalized. A POVM is \emph{sharp}, or a
\emph{projection-valued measure (PVM)}, if every $\E(X)$ is a projection.

A POVM is \emph{localizable} if it has the \emph{norm-$1$ property}:
$\|\E(X)\|=1$ whenever $\E(X)\neq0$
\cite{busch_quantum_2016,carette_operational_2025}.
Equivalently, for every such $X$ there is a sequence of normal states
$\omega_n\in\S(\M)$ with $\E_{\omega_n}(X)\to1$. This follows from
\[
    \|\E(X)\|=\sup_{\omega\in\S(\M)}\omega(\E(X)),
\]
since normal states determine the norm of positive elements.
This is the normal-state formulation of the concentration criterion
in~\cite[Prop.~2.12]{carette_operational_2025}.
In particular, every PVM is localizable. No topology on $\Sigma$
is required.

POVMs are subject to actions of pre- and post-composition with suitable maps: given a measurable map $\alpha: (\Sigma,\F) \to (\Sigma',\F')$, the \emph{push-forward} $\alpha_*\E := \E \circ \alpha^{-1}$ is a POVM on $(\Sigma',\F')$ \cite{busch_quantum_2016}, and given a normal channel $\Phi: \M \to \N$, the map $\Phi \circ \E: X \mapsto \Phi(\E(X))$ is a POVM on $(\Sigma,\F)$ with values in $\N$ and we have
\begin{equation}\label{eq:postcompre}
(\Phi \circ \E)_\omega = \E_{\Phi_*\omega}.
\end{equation}

\section{The operator-valued integrals}\label{sec:generalities}

Throughout Sections~\ref{sec:generalities}--\ref{sec:parametrized}, $\MS$ and $\MR$ are arbitrary $\mathrm{W}^*$-algebras, $(\Sigma,\F)$ is a measurable space, and $\E: \F \to \Eff(\MR)$ is a normalized POVM. Separability of the preduals $\MSp$ and/or $\MRp$ is assumed when needed.

\begin{definition}\label{def:measurable}
    A function $f: \Sigma \to \MS$ is called $(\F$-$)$\emph{ultraweakly measurable} if for each $\rho \in \S(\MS)$ $($equivalently, for each $\rho \in \MSp$$)$ the scalar function
    \[
    	f_\rho: \Sigma \ni x \mapsto \rho(f(x)) \in \Cn
    \]
    is $\F$-measurable.\footnote{In the terminology of Diestel and Uhl \cite[Ch.~II, Sec.~1]{diestel_vector_1977}, a function $f: \Sigma \to \MS$ is \emph{$\Gamma$-measurable} for a subspace $\Gamma \subseteq \MS^*$ if $\Lambda \circ f$ is measurable for every $\Lambda \in \Gamma$. The space $\Bmes$ used throughout this paper is the space of $\MSp$-measurable functions in this sense, rather than the (in general smaller) space of Pettis-measurable functions which in this case would take $\Gamma = \MS^*$.} It is uniformly bounded if $\|f\|_\infty = \sup_{x \in \Sigma} \|f(x)\| < \infty$. We denote by $\Bmes$ the set of $\F$-ultraweakly measurable $\MS$-valued functions on $\Sigma$, and by $\Bb$ the subspace\footnote{We will later show that it is in fact a $\mathrm{C}^*$-algebra, see Prop. \ref{prop:BbisCstar}.} of uniformly bounded ones.
\end{definition}

Notice that for any $f \in \Bb$ the functions $f_\rho$ are $\F$-measurable and bounded (by $\|f\|_\infty$), and hence integrable with respect to any finite measure; in particular $f_\rho$ is $\E_\omega$-integrable for all $\rho \in \S(\MS)$, $\omega \in \S(\MR)$ and any POVM $\E: \F \to \Eff(\MR)$. The integration theory developed here is restricted to uniformly bounded integrands.

\subsection{Preparatory lemmas}

A couple of observations are in place here before we proceed to the main results. We begin by observing that applying vector functionals to uniformly bounded ultraweakly measurable functions and normalized positive operator-valued measures gives rise to scalar bounded measurable integrands, finite complex-valued measures, and thus bounded integrals.

\begin{lemma}[Scalar integrability in a faithful representation]\label{lem:Hilbertlevel}
Choose faithful normal representations $\pi_\S:\MS\hookrightarrow B(\his)$ and $\pi_\R:\MR\hookrightarrow B(\hir)$, and let $f\in\Bb$. For $\xi,\xi'\in\his$ and $\eta,\eta'\in\hir$, set
\[
    \tilde f_{\xi,\xi'}(x):=\braket{\xi}{\pi_\S(f(x))\xi'},
    \qquad \tilde\E_{\eta,\eta'}(X):=\braket{\eta}{\pi_\R(\E(X))\eta'}.
\]
Then $\tilde f_{\xi,\xi'}$ is bounded and measurable, and $\tilde\E_{\eta,\eta'}$ is a finite complex measure. In particular, $\tilde f_{\xi,\xi'}\in L^1(|\tilde\E_{\eta,\eta'}|)$.
More precisely,
\[
    \int_\Sigma |\tilde f_{\xi,\xi'}(x)|\,d|\tilde\E_{\eta,\eta'}|(x)
    \leq \|f\|_\infty\|\xi\|\|\xi'\|\|\eta\|\|\eta'\|.
\]
\end{lemma}

\begin{proof}
Vector functionals are normal, so $\tilde f_{\xi,\xi'}$ is measurable and bounded in modulus by $\|f\|_\infty\|\xi\|\|\xi'\|$. The measure is countably additive by normality. For every finite measurable partition $(X_j)$, positivity and Cauchy--Schwarz give
\[
    \sum_j|\tilde\E_{\eta,\eta'}(X_j)|
    \leq\sum_j\sqrt{\tilde\E_{\eta,\eta}(X_j)
                            \tilde\E_{\eta',\eta'}(X_j)}
    \leq\|\eta\|\|\eta'\|,
\]
since $\E(\Sigma)=\algunit$. Thus its total variation is finite, and integrability follows.
\end{proof}

Next we establish pointwise ultraweak approximation of functions in $\Bb$ by simple functions.

\begin{definition}[Simple functions]
	A function $f: \Sigma \to \MS$ is called \emph{simple} if $f(x) = \sum_{j=1}^n a_j \chi_{X_j}(x)$ for some finite $\F$-measurable partition $(X_j)_{j=1}^n$ of $\Sigma$ and some $a_1,\ldots,a_n \in \MS$.
\end{definition}

Notice that simple functions are automatically $\F$-measurable.

\begin{lemma}[Point-wise approximation by simple functions]\label{lem:simpleapprox}
Assume $\MSp$ separable. Then for each $f \in \Bb$ there exists a sequence $(f_n)_{n \in \Nn}$ of $\MS$-valued simple functions such that $f_n(x) \xrightarrow{\rm uw} f(x)$ for all $x \in \Sigma$ and $\|f_n\|_\infty \leq \|f\|_\infty$ for all $n$.
\end{lemma}

\begin{proof}
It suffices to prove the claim when $\|f\|_\infty = 1$. By separability of $\MSp$, take a sequence $(\omega_k)_{k \in \Nn}$ dense in $\MSp$. For each $n \in \Nn$ and $a \in \MS$ define the ultraweakly open set
\[
    U_n(a) := \big\{b \in \MS : |\omega_k(b-a)| < \tfrac{1}{n} \text{ for } k=1,\ldots,n \big\}.
\]
The unit ball $(\MS)_1$ is ultraweakly compact by the Banach--Alaoglu theorem, and the ultraweak topology on $(\MS)_1$ is metrizable since $\MSp$ is separable (see \eqref{eq:metric}). For each $n$, the sets $\{U_n(a) \cap (\MS)_1\}_{a \in (\MS)_1}$ form an open cover of $(\MS)_1$, and by compactness there exists a finite subcover indexed by $a_{n,1},\ldots,a_{n,M_n} \in (\MS)_1$. Define
\[
    Y_{n,1} := U_n(a_{n,1}), \; Y_{n,j} := U_n(a_{n,j}) \setminus \bigcup_{k=1}^{j-1} U_n(a_{n,k}) \; (j=2,\ldots,M_n)
\]
and $X_{n,j} := f^{-1}(Y_{n,j})$. Since $f$ is ultraweakly measurable, we have
\[
    f^{-1}(U_n(a_{n,j})) = \big\{x \in \Sigma : |\omega_k(f(x)-a_{n,j})| < \tfrac{1}{n} \text{ for } k=1,\ldots,n \big\} \in \F
\]
(as a finite intersection of measurable sets), and hence $X_{n,j} \in \F$ (as a difference of measurable sets). Since $f(\Sigma) \subseteq (\MS)_1$, the family $(X_{n,j})_{j=1}^{M_n}$ is a finite $\F$-measurable partition of $\Sigma$. Define
\[
f_n := \sum_{j=1}^{M_n} a_{n,j}\chi_{X_{n,j}}.
\]
From $a_{n,j} \in (\MS)_1$ we have $\|f_n\|_\infty \leq 1$, and \emph{each $f_n$ takes values in $\MS$}.

We show $f_n(x) \xrightarrow{\rm uw} f(x)$ for each $x \in \Sigma$. Fix $x$, $\omega \in \MSp$, and $\epsilon > 0$. By density of $(\omega_k)$ there is $k_0$ with $\|\omega_{k_0} - \omega\| < \epsilon$. Take $N_0 \geq \max(\epsilon^{-1}, k_0)$ and fix $n \geq N_0$. The unique $j$ with $x \in X_{n,j}$ satisfies $|\omega_{k_0}(f(x)-a_{n,j})| < 1/n \leq \epsilon$ (since $f(x) \in U_n(a_{n,j})$ and $k_0 \leq n$). Hence
\begin{align*}
    |\omega(f(x)-f_n(x))| &\leq |(\omega-\omega_{k_0})(f(x)-f_n(x))| + |\omega_{k_0}(f(x)-f_n(x))| \\
    &\leq 2\|\omega - \omega_{k_0}\| + \epsilon < 3\epsilon,
\end{align*}
where the factor $2$ comes from $\|f(x)-f_n(x)\| \leq \|f(x)\| + \|f_n(x)\| \leq 2$. This establishes $\omega(f_n(x)) \to \omega(f(x))$ for each $\omega \in \MSp$.
\end{proof}

Positive matrix-valued measurable functions can also be point-wise approximated by positive simple functions, entrywise in $L^1$-spaces arising through vector functionals from positive operator-valued measures.

\begin{lemma}[Point-wise approximation by simple positive-matrix-valued functions]\label{lem:posmatapprox} Let $(\Sigma,\F)$ be a measurable space, $n \in \Nn$, and let $F: \Sigma \to M_n(\Cn)_+$ be an $\F$-measurable $($each entry $F_{ij}: \Sigma \to \Cn$ is $\F$-measurable$)$ positive-semidefinite matrix-valued function. Then:
\begin{enumerate}[label=(\roman*)]
	\item There exists a sequence $(F^{(N)})_{N \in \Nn}$ of $M_n(\Cn)_+$-valued simple functions such that $F^{(N)}(x) \to F(x)$ pointwise.
	\item If moreover $F_{ij}$ is $|\E_{\eta_i,\eta_j}|$-integrable for all $i,j$ where $\E: \F \to \Eff(\hi)$ is a POVM on a Hilbert space $\hi$ and $\eta_1,\ldots,\eta_n \in \hi$, then $F^{(N)}_{ij} \to F_{ij}$ in $L^1(|\E_{\eta_i,\eta_j}|)$.
\end{enumerate}
\end{lemma}

\begin{proof}
(i) Since the square root function on $M_n(\Cn)_+$ is continuous, $G := \sqrt{F}: \Sigma \to M_n(\Cn)_+$ is $\F$-measurable. For each $i,j$, choose a sequence of $\F$-measurable simple functions $G^{(N)}_{ij}$ with $|G^{(N)}_{ij}(x)| \leq |G_{ij}(x)|$ and $G^{(N)}_{ij}(x) \to G_{ij}(x)$ for all $x \in \Sigma$. Define
\[
    F^{(N)}(x) := G^{(N)}(x)^\dagger G^{(N)}(x), \; \text{i.e.,} \; F^{(N)}_{ij}(x) = \sum_{k=1}^n \overline{G^{(N)}_{ki}(x)} G^{(N)}_{kj}(x).
\]
Each $F^{(N)}$ is manifestly positive semidefinite, $\F$-measurable, and simple. Since $G^{(N)}(x) \to G(x)$ entrywise, we get $F^{(N)}(x) \to G(x)^\dagger G(x) = F(x)$ pointwise.

(ii) For the $L^1$-convergence, we use the bound
\begin{equation}\label{eq:totalvarbound}
    |\E_{\eta_i,\eta_j}|(B) \leq \int_B \sqrt{d\E_{\eta_i,\eta_i}\,d\E_{\eta_j,\eta_j}} \; \text{for all } B \in \F,
\end{equation}
To see this, put $\lambda=\E_{\eta_i,\eta_i}+\E_{\eta_j,\eta_j}$. The positive $2\times2$ matrix measure with entries $\E_{\eta_r,\eta_s}$, $r,s\in\{i,j\}$, has a positive Radon--Nikodym density with respect to $\lambda$ almost everywhere. Writing $h_{rs}=d\E_{\eta_r,\eta_s}/d\lambda$, we obtain $|h_{ij}|^2\leq h_{ii}h_{jj}$; integrating the square root of this inequality proves \eqref{eq:totalvarbound}, with the right-hand side understood relative to $\lambda$. From $\sum_j |G_{ji}|^2 = F_{ii} \in L^1(\E_{\eta_i,\eta_i})$ we get $G_{ki} \in L^2(\E_{\eta_i,\eta_i})$, and by the Cauchy--Schwarz inequality and \eqref{eq:totalvarbound}, $\overline{G_{ki}}G_{kj} \in L^1(|\E_{\eta_i,\eta_j}|)$. Since $\left|\overline{G^{(N)}_{ki}}G^{(N)}_{kj}\right| \leq \left|\overline{G_{ki}}G_{kj}\right|$, the dominated convergence theorem gives $F^{(N)}_{ij} \to F_{ij}$ in $L^1(|\E_{\eta_i,\eta_j}|)$.
\end{proof}

Finally, recall that a sesquilinear form $b$ on a dense subspace $D$ of a Hilbert space $H$, conjugate-linear in the first variable and linear in the second, satisfying $|b(\psi,\varphi)|\leq C\|\psi\|\|\varphi\|$, extends uniquely to $H\times H$. By the Riesz representation theorem it has the form $b(\psi,\varphi)=\langle\psi,A\varphi\rangle$ for a unique $A\in B(H)$ with $\|A\|\leq C$. If $b$ is positive and $b(\psi,\psi)\leq C\|\psi\|^2$, Cauchy--Schwarz for positive forms gives this bound and $0\leq A\leq C\id$. A pair $(f,\E)$, where $f$ is a positive ultraweakly measurable uniformly bounded function and $\E$ a normalized POVM, gives rise to such a form, which is how the integral operators will be constructed in the next section.

\begin{lemma}[$(f,\E)$-sesquilinear form]\label{lem:semiip}
Let $\hi_\S,\hir$ be Hilbert spaces, $\E: (\Sigma,\F) \to \Eff(\hir)$ a POVM, and $f: \Sigma \to B(\hi_\S)_+$ a positive uniformly bounded $\F$-measurable function. Then the formula
\[
    \ip{\xi \otimes \eta}{\xi' \otimes \eta'}_{(f,\E)} := \int_\Sigma \braket{\xi}{f(x)\xi'} \, d\E_{\eta,\eta'}(x),
\]
where $\E_{\eta,\eta'}(X) := \braket{\eta}{\E(X)\eta'}$, defines a positive semidefinite sesquilinear form on $\hi_\S \otimes_{\rm alg} \hir$ satisfying
\begin{equation}\label{eq:semiipbound}
    \ip{\psi}{\psi}_{(f,\E)} \leq \|f\|_\infty \|\psi\|^2 \; \text{for all } \psi \in \hi_\S \otimes_{\rm alg} \hir,
\end{equation}
and therefore extending to the full tensor product $\hi_\S \otimes \hir$.
\end{lemma}

\begin{proof}
The scalar integrals are well-defined by Lemma~\ref{lem:Hilbertlevel}, applied to the identity representations of $B(\hi_\S)$ and $B(\hir)$. The formula on elementary tensors is conjugate-linear in $\xi,\eta$ and linear in $\xi',\eta'$, so it induces a well-defined sesquilinear form on the algebraic tensor product. Conjugate symmetry follows from the corresponding properties of the integral and $\overline{\E_{\eta,\eta'}(X)} = \E_{\eta',\eta}(X)$. For positive semidefiniteness, let $\psi = \sum_{i=1}^n \xi_i \otimes \eta_i$ and set $f_{ij}(x) := \braket{\xi_i}{f(x)\xi_j}$ and $\mu_{ij} := \E_{\eta_i,\eta_j}$. We need to show
\[
    \ip{\psi}{\psi}_{(f,\E)} = \sum_{i,j=1}^n \int_\Sigma f_{ij}\,d\mu_{ij} \geq 0.
\]
The matrix $F(x) := (f_{ij}(x))_{i,j=1}^n$ is positive semidefinite for each $x$ (since $f \geq 0$), and $F$ is $\F$-measurable. By Lemma~\ref{lem:posmatapprox}, there exists a sequence $(F^{(N)})$ of $M_n(\Cn)_+$-valued simple functions with $F^{(N)}_{ij} \to f_{ij}$ in $L^1(|\mu_{ij}|)$. Writing each simple function as $F^{(N)} = \sum_{m=1}^{M_N} C^{(N,m)} \chi_{Y_m}$ where $C^{(N,m)} \in M_n(\Cn)_+$, we have
\[
    \sum_{i,j=1}^n \int_\Sigma F^{(N)}_{ij}\,d\mu_{ij} = \sum_{m=1}^{M_N} \sum_{i,j=1}^n C^{(N,m)}_{ij}\mu_{ij}(Y_m) = \sum_{m=1}^{M_N} \tr[C^{(N,m)} \cdot M_m^{\mathsf T}],
\]
where $M_m := (\mu_{ij}(Y_m))_{i,j} = (\braket{\eta_i}{\E(Y_m)\eta_j})_{i,j} \in M_n(\Cn)_+$ (positive semidefinite since $\E(Y_m) \geq 0$). The transpose $M_m^{\mathsf T}$ is also positive semidefinite. Since $\tr[PQ] \geq 0$ for $P,Q \geq 0$, each summand is non-negative. Passing to the limit $N \to \infty$ using $L^1$-convergence gives $\ip{\psi}{\psi}_{(f,\E)} \geq 0$.

For the bound \eqref{eq:semiipbound}, notice that $g := \|f\|_\infty\id - f \geq 0$ pointwise, so $\ip{\psi}{\psi}_{(g,\E)} \geq 0$ by the above. Since the form is additive in $f$
\begin{align*}
    \ip{\psi}{\psi}_{(f,\E)} &\leq \ip{\psi}{\psi}_{(\|f\|_\infty\id,\E)} = \|f\|_\infty \sum_{i,j=1}^n \braket{\xi_i}{\xi_j}\braket{\eta_i}{\E(\Sigma)\eta_j} \\
    &= \|f\|_\infty \sum_{i,j=1}^n \braket{\xi_i}{\xi_j}\braket{\eta_i}{\eta_j} = \|f\|_\infty \|\psi\|^2\hspace{-2pt},
\end{align*}
where we used $\E(\Sigma) = \id$.
\end{proof}

\subsection{Integral operators}

We are now equipped to prove the central result of integration theory in $\mathrm{W}^*$-algebras developed in this work.

\begin{theorem}[Existence and uniqueness of integral operators]\label{thm:intop}
    Let $\MS$ and $\MR$ be $\mathrm{W}^*$-algebras and $(\Sigma,\F)$ a measurable space, assume $\MSp$ separable. Then for any POVM $\E:\F \to \Eff(\MR)$ and uniformly bounded $\F$-measurable function $f: \Sigma \to \MS$ there exists a unique element
    \[
        \int_\Sigma f \otimes d\E \in \MSR,
    \]
    satisfying, for all $\rho \in \S(\MS)$ and $\omega \in \S(\MR)$, the characterizing formula:
    \begin{equation}\label{eq:characformula}
        (\rho \otimes \omega)\Big(\int_\Sigma f \otimes d\E\Big) = \int_\Sigma f_\rho \, d\E_\omega.
    \end{equation}
        In particular, for simple functions $f = \sum_j a_j \chi_{X_j}$ we have
    \begin{equation}\label{eq:intsimple}
        \int_\Sigma f \otimes d\E = \sum_j a_j \otimes \E(X_j).
    \end{equation}
    Moreover, this assignment is functorial in the following sense: for any measurable $\alpha: (\Sigma,\F) \to (\Sigma',\F')$, any uniformly bounded $\F'$-measurable function $f: \Sigma' \to \MS$, and normal channels $\Psi: \MS \to \N_\S$ and $\Phi: \MR \to \N_\R$ with $(\N_\S)_*$ separable, we have\footnote{Extending this result by including nontrivial $\Psi: \MS \to \N_\S$ is due to Samuel Fedida.}
    \begin{equation}\label{eq:pushfwd}
        \int_\Sigma (\Psi \circ f \circ \alpha) \otimes d(\Phi \circ \E) = (\Psi \otimes \Phi)\Big(\int_{\Sigma'} f \otimes d(\alpha_*\E)\Big).
    \end{equation}
\end{theorem}

\begin{proof}
Choose faithful normal representations $\pi_\S: \MS \hookrightarrow B(\his)$ (with $\his$ separable, possible since $\MSp$ is separable) and $\pi_\R: \MR \hookrightarrow B(\hir)$. Write $\tilde{f} := \pi_\S \circ f: \Sigma \to B(\his)$ and $\tilde{\E} := \pi_\R \circ \E: \F \to \Eff(\hir)$. We identify $\MSR$ with the ultraweak closure of $\pi_\S(\MS) \otimes_{\rm alg} \pi_\R(\MR)$ in $B(\hisr)$.

\emph{Step 1: Existence in $B(\hisr)$.} For $f\geq0$, Lemma~\ref{lem:semiip} and the representation theorem for bounded sesquilinear forms give a unique positive operator $A_f\in B(\hisr)$ with $0\leq A_f\leq\|f\|_\infty\id$ and
\begin{equation}\label{eq:form-matrix-coefficients}
    \braket{\xi\otimes\eta}{A_f(\xi'\otimes\eta')}
    =\int_\Sigma\braket{\xi}{\tilde f(x)\xi'}\,
       d\tilde\E_{\eta,\eta'}(x).
\end{equation}
For self-adjoint $f$, put $C=\|f\|_\infty$ and define $A_f:=A_{f+C\algunit}-C\id$. Since $0\leq f+C\algunit\leq2C\algunit$, the positive construction gives $-C\id\leq A_f\leq C\id$, and therefore $\|A_f\|\leq C$. Normalization of $\E$ shows that this operator still satisfies~\eqref{eq:form-matrix-coefficients}. For arbitrary $f$, the functions $\mathrm{Re}\,f$ and $\mathrm{Im}\,f$ are bounded and ultraweakly measurable; define
\[
    A_f:=A_{\mathrm{Re}\,f}+iA_{\mathrm{Im}\,f},
    \qquad \|A_f\|\leq2\|f\|_\infty.
\]
Equation~\eqref{eq:form-matrix-coefficients} holds in every case. Since elementary tensors span a dense subspace of $\hisr$, it determines $A_f$ uniquely. The constructions therefore agree on their overlapping domains, and $f\mapsto A_f$ is linear. Write $\int\tilde f\otimes d\tilde\E:=A_f$.

We also need the trace-class version of~\eqref{eq:form-matrix-coefficients}:
\begin{equation}\label{eq:form-trace-identity}
    \tr[(\tilde\rho\otimes\tilde\omega)A_f]
    =\int_\Sigma\tr[\tilde\rho\,\tilde f(x)]\,
       d\tilde\E_{\tilde\omega}(x),
    \qquad \tilde\rho\in\T(\his),\quad\tilde\omega\in\T(\hir).
\end{equation}
For positive $\tilde\rho=\sum_k\lambda_k\dyad{\xi_k}{\xi_k}$ and $\tilde\omega=\sum_l\nu_l\dyad{\eta_l}{\eta_l}$ with unit vectors, both sides are the sum of the corresponding diagonal matrix-coefficient identities. The scalar integrals are absolutely summable, with
\[
    \sum_{k,l}\lambda_k\nu_l
       \int_\Sigma\left|\braket{\xi_k}{\tilde f(x)\xi_k}\right|\,
              d\tilde\E_{\eta_l,\eta_l}(x)
    \leq\|f\|_\infty\tr[\tilde\rho]\tr[\tilde\omega].
\]
Thus the interchange of sums and integrals is justified. Decomposing arbitrary trace-class operators into linear combinations of positive ones proves~\eqref{eq:form-trace-identity} in general.

\emph{Step 2: Membership in $\MSR$ (via approximation).} For $\MS$-valued simple functions $f = \sum_j a_j \chi_{X_j}$ ($a_j \in \MS$), uniqueness gives $\int \tilde{f} \otimes d\tilde{\E} = \sum_j \pi_\S(a_j) \otimes \pi_\R(\E(X_j)) \in \pi_\S(\MS) \otimes_{\rm alg} \pi_\R(\MR) \subset \MSR$. For general $f \in \Bb$, Lemma~\ref{lem:simpleapprox} gives $\MS$-valued simple $f_n$ with $\|f_n\|_\infty \leq \|f\|_\infty$ and $f_n(x) \xrightarrow{\rm uw} f(x)$. We claim $\int \tilde{f}_n \otimes d\tilde{\E} \xrightarrow{\rm uw} \int \tilde{f} \otimes d\tilde{\E}$ in $B(\hisr)$: for each $\tau \in \T(\hisr)$ and $\epsilon > 0$, approximate $\tau$ by $\tau_0 \in \T(\his) \otimes_{\rm alg} \T(\hir)$ with $\|\tau - \tau_0\|_1 < \epsilon$, then
\begin{align*}
    \left|\tr[\tau \int (\tilde{f}_n-\tilde{f}) \otimes d\tilde{\E}]\right| &\leq \|\tau - \tau_0\|_1 \cdot \left|\left|\int (\tilde{f}_n-\tilde{f}) \otimes d\tilde{\E}\right|\right| + \left|\tr[\tau_0 \int (\tilde{f}_n-\tilde{f}) \otimes d\tilde{\E}]\right| \\
    &\leq 8\|f\|_\infty \cdot \epsilon + \left|\tr[\tau_0 \int (\tilde{f}_n-\tilde{f}) \otimes d\tilde{\E}]\right|,
\end{align*}
where the second term vanishes as $n \to \infty$ by the dominated convergence theorem (expanding $\tau_0 = \sum_k \tilde{\rho}_k \otimes \tilde{\omega}_k$ as a finite sum). Since each $\int \tilde{f}_n \otimes d\tilde{\E} \in \MSR$ and $\MSR$ is ultraweakly closed in $B(\hisr)$, the limit $\int \tilde{f} \otimes d\tilde{\E}$ lies in $\MSR$.

\emph{Step 3: Defining the abstract integral.} By Step~2, $\int \tilde{f} \otimes d\tilde{\E}$ lies in $\overline{\pi_\S(\MS) \otimes_{\rm alg} \pi_\R(\MR)}^{\rm uw} \subset B(\hisr)$. We define the abstract integral as \[\int_\Sigma f \otimes d\E := (\pi_\S \otimes \pi_\R)^{-1}(\int \tilde{f} \otimes d\tilde{\E}) \in \MSR.\] This element does not depend on the choice of faithful normal representations: the characterizing property~\eqref{eq:characformula} -- verified below -- is formulated entirely in terms of the abstract algebras $\MS$, $\MR$ and their states, and determines $\int_\Sigma f \otimes d\E$ uniquely by Lemma~\ref{lem:extfromstates}. Indeed, for $\rho \in \S(\MS)$ and $\omega \in \S(\MR)$, choose positive trace-class $\tilde{\rho} \in \T(\his)_+$ and $\tilde{\omega} \in \T(\hir)_+$ representing them (i.e., $\rho(\cdot) = \tr[\tilde{\rho}\,\pi_\S(\cdot)]$ and $\omega(\cdot) = \tr[\tilde{\omega}\,\pi_\R(\cdot)]$). Then \[(\rho \otimes \omega)(\int f \otimes d\E) = \tr[(\tilde{\rho} \otimes \tilde{\omega}) \int \tilde{f} \otimes d\tilde{\E}] = \int_\Sigma \tr[\tilde{\rho}\,\tilde{f}(x)] \, d\tilde{\E}_{\tilde{\omega}}(x) = \int_\Sigma f_\rho \, d\E_\omega.\]

\emph{Step 4: Functoriality.} Normality of $\Psi$ ensures that $\Psi\circ f\circ\alpha$ is ultraweakly measurable, and unital positivity ensures it is bounded. Separability of $(\N_\S)_*$ therefore makes its integral against $\Phi\circ\E$ well-defined by the existence result above. For $\rho \in \S(\N_\S)$ and $\omega \in \S(\N_\R)$:
\begin{align*}
    (\rho \otimes \omega)\Big(\int_\Sigma (\Psi \circ f \circ \alpha) \otimes d(\Phi \circ \E)\Big)
    &= \int_\Sigma (\Psi \circ f \circ \alpha)_\rho \, d(\Phi \circ \E)_\omega
    = \int_\Sigma f_{\Psi_*\rho} \circ \alpha \, d\E_{\Phi_*\omega} \\
    &= \int_{\Sigma'} f_{\Psi_*\rho} \, d(\alpha_*\E)_{\Phi_*\omega}
    = (\Psi_*\rho \otimes \Phi_*\omega)\Big(\int_{\Sigma'} f \otimes d(\alpha_*\E)\Big) \\
    &= (\rho \otimes \omega)\Big((\Psi \otimes \Phi)\big(\int_{\Sigma'} f \otimes d(\alpha_*\E)\big)\Big),
\end{align*}
where the second equality uses $(\Psi \circ f \circ \alpha)_\rho = \rho \circ \Psi \circ f \circ \alpha = f_{\Psi_*\rho} \circ \alpha$ and $(\Phi \circ \E)_\omega = \E_{\Phi_*\omega}$, and the third is the Lebesgue change of variables $\int_\Sigma g \circ \alpha \, d\mu = \int_{\Sigma'} g \, d(\alpha_*\mu)$. Uniqueness (Lemma~\ref{lem:extfromstates}) gives \eqref{eq:pushfwd}.
\end{proof}

\section{The domains of integration}\label{sec:Bb}

We now identify natural domains for the integration maps. Two perspectives coexist and are both useful: before a POVM is fixed, one seeks a universal domain common to all integration maps --- this role is played by $\Bb$, the $\mathrm{C}^*$-algebra of uniformly bounded ultraweakly measurable functions, maximal among such common domains; once a POVM $\E$ is fixed, the natural domain is the quotient $\Linf := \Bb/\N_\E$ by the ideal of $\E$-null functions --- defined via an $\E$-essential supremum in analogy with scalar Lebesgue $L^\infty_\mu(\Sigma)$ spaces. Assuming separability of $\MSp$, it is a $\mathrm{C}^*$-algebra and the descended integration map is faithful; if $\MRp$ is also separable, it is a $\mathrm{W}^*$-algebra and integration is normal (Section~\ref{sec:propint}).

\subsection{The universal domain \texorpdfstring{$\Bb$}{Bb}}\label{sec:Bb-universal}

\begin{proposition}\label{prop:BbisCstar}
Assume $\MSp$ separable. Then $\Bb$ is a $\mathrm{C}^*$-algebra under the uniform norm $\|f\|_\infty = \sup_{x \in \Sigma} \|f(x)\|$ and the pointwise operations $f^*(x) := f(x)^*$, $(fg)(x) := f(x)g(x)$.
\end{proposition}

\begin{proof}
\emph{Completeness.} Given a uniformly Cauchy sequence $(f_m) \subset \Bb$ converging uniformly to $f: \Sigma \to \MS$, we have $f_m(x) \to f(x)$ in norm for each $x$. Since $\MS$ is norm-closed (being a $\mathrm{C}^*$-algebra), $f(x) \in \MS$. For each $\omega \in \MSp$, the functions $x \mapsto \omega(f_m(x))$ are measurable and converge pointwise to $x \mapsto \omega(f(x))$, so $f$ is ultraweakly measurable; it is clearly uniformly bounded, so $f \in \Bb$.

\emph{Closure under multiplication.} For $f,g \in \Bb$, the uniform bound $\|fg\|_\infty \leq \|f\|_\infty\|g\|_\infty$ is immediate, and $f(x)g(x) \in \MS$ since $\MS$ is closed under multiplication. It remains to show $fg$ is ultraweakly measurable. Fix $\omega \in \MSp$. By Lemma~\ref{lem:simpleapprox}, there exist $\MS$-valued simple functions $g_n$ with $g_n(x) \xrightarrow{\rm uw} g(x)$ for all $x$ and $\|g_n\|_\infty \leq \|g\|_\infty$. Writing $g_n = \sum_{j=1}^{M_n} b_{n,j}\chi_{X_{n,j}}$ with $b_{n,j} \in \MS$, we have for each $x \in \Sigma$:
\[
    \omega(f(x)g(x)) = \lim_{n \to \infty} \omega(f(x)g_n(x)) = \lim_{n \to \infty}\sum_{j=1}^{M_n} (b_{n,j} \cdot \omega)(f(x))\chi_{X_{n,j}}(x),
\]
where $b_{n,j} \cdot \omega \in \MSp$ is the left-module action \eqref{eq:bimodule}. The convergence $\omega(f(x)g_n(x)) \to \omega(f(x)g(x))$ holds because $g_n(x) \xrightarrow{\rm uw} g(x)$ implies $f(x)g_n(x) \xrightarrow{\rm uw} f(x)g(x)$ (multiplication is separately ultraweakly continuous on bounded sets). Since $x \mapsto (b_{n,j} \cdot \omega)(f(x))$ is measurable by definition of $\Bb$, the right-hand side is a pointwise limit of measurable functions and hence measurable (see e.g.~\cite[Cor.~2.9]{folland_real_1999}).

\emph{Closure under adjoints.} For $\omega\in\MSp$, set $\omega^*(a):=\overline{\omega(a^*)}$, which is again normal. Since $\omega(f(x)^*)=\overline{\omega^*(f(x))}$, the function $f^*$ is ultraweakly measurable, and $\|f^*\|_\infty=\|f\|_\infty$. Thus $f^*\in\Bb$.

The algebraic properties (associativity, distributivity) are inherited pointwise from $\MS$. The $\mathrm{C}^*$-identity $\|f^*f\|_\infty = \|f\|_\infty^2$ follows from $\|f(x)^*f(x)\| = \|f(x)\|^2$ for each $x$.
\end{proof}

The $\mathrm{C}^*$-algebra $\Bb$ is the largest common domain for all integration maps in the following sense:

\begin{proposition}\label{prop:integrability}
Let $f: \Sigma \to \MS$ be ultraweakly measurable but unbounded. Then there
exist a $\mathrm{W}^*$-algebra $\MR$ and a POVM $\E: \F \to \Eff(\MR)$ for which
$\int_\Sigma f \otimes d\E$ does not exist as a bounded operator in $\MSR$ satisfying the characterizing formula.
\end{proposition}

\begin{proof}
Suppose $f: \Sigma \to \MS$ is ultraweakly measurable
with $\|f\|_\infty = \infty$. Take $\MR := \ell^\infty(\Sigma)$ and define the
PVM $\E: \F \to \Eff(\ell^\infty(\Sigma))$ by $\E(X) := \chi_X$. For each
$x \in \Sigma$ the evaluation functional $\delta_x$ is a normal state on
$\ell^\infty(\Sigma)$ with $\E_{\delta_x} = \delta_x$, so $\E$ is localizable.
If $\int_\Sigma f \otimes d\E$ existed as a bounded operator satisfying the
characterizing formula \eqref{eq:characformula}, then, by linearity, for every $\rho \in (\MS)_{*,1}$
and $x \in \Sigma$,
\[
    \left|(\rho \otimes \delta_x)\Big(\int_\Sigma f \otimes d\E\Big)\right|
    = \left|\int_\Sigma f_\rho \, d\delta_x\right| = \left|\rho(f(x))\right|,
\]
Taking the supremum over the predual unit ball gives $\|\int_\Sigma f \otimes d\E\| \geq \|f(x)\|$ for each $x$, and hence
$\|\int_\Sigma f \otimes d\E\| \geq \|f\|_\infty = \infty$, a contradiction.
\end{proof}

\subsection{The specific algebraic domain \texorpdfstring{$\Linf$}{L-infty-E}}\label{sec:LinfE}

Once a POVM $\E$ is fixed, the universal domain $\Bb$ can be refined to the quotient space $\Linf$, constructed in complete analogy with the scalar Lebesgue case: one starts from all ultraweakly measurable functions, imposes an $\E$-essential boundedness condition, and quotients by $\E$-null equality. Under the assumption that both $\MSp$ and $\MRp$ are separable, the resulting space is a $\mathrm{W}^*$-algebra carrying a precise predual description.

\begin{definition}\label{def:E-equiv}
A set $N \in \F$ is $\E$-\emph{null} if $\E(N) = 0$ in $\MR$, equivalently $\E_\omega(N) = 0$ for every $\omega \in \S(\MR)$. Two functions $f, g \in \Bb$ are $\E$-\emph{equivalent}, written $f \sim_\E g$, if $\{x \in \Sigma : f(x) \neq g(x)\}$ is contained in an $\E$-null set.
\end{definition}

The collection of $\E$-null sets is a $\sigma$-ideal in $\F$. The following is a simple consequence of the above definition.

\begin{lemma}\label{lem:singlenull}
Assume $\MSp$ separable. We have $f \sim_\E g$ if and only if $f_\rho = g_\rho$ holds $\E_\omega$-a.e.\ for every $\rho \in \S(\MS)$ and every $\omega \in \S(\MR)$.
\end{lemma}

\begin{proof}
The forward direction is immediate. For the converse, let $\{\rho_n\}_{n \in \Nn}$ be a countable norm-dense subset of $\S(\MS)$ (existing by separability of $\MSp$). For each $n$ the set $N_n := \{x : \rho_n(f(x) - g(x)) \neq 0\}$ is $\E$-null by assumption. Then $N := \bigcup_n N_n$ is $\E$-null, and for $x \notin N$ we have $\rho_n(f(x) - g(x)) = 0$ for all $n$, hence $\rho(f(x)-g(x)) = 0$ for every $\rho \in \MSp$ by norm-density and linearity, hence $f(x) = g(x)$.
\end{proof}

We now give the definition of $\Linf$ space, straightforwardly generalizing the one of Lebesgue theory.

\begin{definition}
	Assume $\MSp$ separable. Then for $f \in \Bmes$, the scalar function $x \mapsto \|f(x)\|$ is $\F$-measurable: $\|f(x)\| = \sup_n |\rho_n(f(x))|$ for any countable norm-dense $\{\rho_n\} \subset (\MS)_{*,1}$, a pointwise supremum of measurable functions. The \emph{$\E$-essential supremum norm} is given by
\begin{equation}\label{eq:essup}
    \|f\|_\E := \inf\{C \geq 0 : \E(\{x \in \Sigma : \|f(x)\| > C\}) = 0\} \; \in [0, \infty],
\end{equation}
which equals $\sup_{\omega \in \S(\MR)} \operatorname{ess\text{-}sup}_{\E_\omega}\|f(\cdot)\|$. We define
\begin{equation}\label{eq:Linf-def}
    \Linf := \{f \in \Bmes : \|f\|_\E < \infty\} \big/\!\sim_\E.
\end{equation}
\end{definition}

\begin{lemma}[Norm recovery for localizable POVMs]\label{lem:localizable-norm}
Let $\E:\F\to\Eff(\MR)$ be a localizable POVM. For every bounded measurable $h:\Sigma\to\Cn$,
\begin{equation}\label{eq:scalar-norm-recovery}
    \|h\|_\E
    =\sup_{\omega\in\S(\MR)}\left|\int_\Sigma h\,d\E_\omega\right|.
\end{equation}
\end{lemma}

\begin{proof}
The upper bound follows because every $\E_\omega$ vanishes on $\E$-null sets. If $\|h\|_\E=0$, this also proves equality. Otherwise, fix $0\leq c<\|h\|_\E$ and $\epsilon>0$. The set $A:=\{x:|h(x)|>c\}$ is not $\E$-null. Partition a bounded Borel set containing $h(A)$ into finitely many Borel pieces of diameter less than $\epsilon$. At least one of their preimages in $A$, denoted by $X\subseteq A$, is not $\E$-null. Choose $z\in h(X)$; then $|z|>c$ and $|h(x)-z|<\epsilon$ on $X$. By localizability, there are normal states $\omega_n$ with $\E_{\omega_n}(X)\to1$, and
\[
    \left|\int_\Sigma h\,d\E_{\omega_n}-z\right|
    \leq \epsilon\E_{\omega_n}(X)
       +(\|h\|_\infty+|z|)(1-\E_{\omega_n}(X)).
\]
Consequently the supremum in~\eqref{eq:scalar-norm-recovery} is at least $|z|-\epsilon>c-\epsilon$. Letting $\epsilon\downarrow0$ and then $c\uparrow\|h\|_\E$ gives the lower bound.
\end{proof}

As it turns out, the classes in $\Linf$ always admit bounded representatives.

\begin{lemma}[Bounded representatives]\label{lem:canonical}
Assume $\MSp$ separable and consider a~closed two-sided $*$-ideal of $\Bb$ given by $\N_\E := \{f \in \Bb : f(x) = 0\ \E\text{-a.e.}\}$. We have an isometric $*$-isomorphism
\[
	\Linf \cong \Bb / \N_\E.
\]
In particular, $\Linf$ is a unital $\mathrm{C}^*$-algebra under the quotient norm $\|[f]\|_\E = \inf_{g \sim_\E f}\|g\|_\infty$.
\end{lemma}

\begin{proof}
Every equivalence class $[f] \in \Linf$ has a bounded representative: setting $N := \{x : \|f(x)\| > \|f\|_\E\}$ (which lies in $\F$ by measurability of $\|f(\cdot)\|$ and is $\E$-null by \eqref{eq:essup}) and $f' := f \cdot \chi_{\Sigma \setminus N}$ gives $f' \in \Bb$ with $\|f'\|_\infty = \|f\|_\E$ and $[f'] = [f]$ in $\Linf$. Conversely, $\N_\E$ is closed in $\Bb$ under the uniform norm (if $f_n \in \N_\E$ and $f_n \to f$ uniformly, then $\|f(x)\| \leq \|f_n(x)\| + \|f_n-f\|_\infty$ gives $\|f(x)\| = 0$ off a countable union of $\E$-null sets, hence $\E$-a.e.) and is a two-sided $*$-ideal (if $f = 0$ $\E$-a.e.\ and $g \in \Bb$, then $fg = gf = 0$ and $f^* = 0$ hold on the same co-null set). The quotient norm on $\Bb / \N_\E$ therefore coincides with the $\E$-essential supremum, and the bounded-representative construction furnishes an isometric $*$-isomorphism $\Bb/\N_\E \to \Linf$. Since a quotient of a $\mathrm{C}^*$-algebra with respect to a closed two-sided $*$-ideal is a $\mathrm{C}^*$-algebra, $\Linf$ inherits this structure from $\Bb$.
\end{proof}

When $\MRp$ is separable, $\Linf$ acquires a $\mathrm{W}^*$-structure and a concrete predual description. Fix, once and for all, a norm-dense sequence $\{\omega_n\}_{n \in \Nn} \subset \S(\MR)$ (possible by separability of $\MRp$) and define the \emph{dominating measure}
\begin{equation}\label{eq:muE}
    \mu_\E := \sum_{n=1}^\infty 2^{-n} \E_{\omega_n},
\end{equation}
a probability measure on $(\Sigma,\F)$. It is the measure corresponding to the faithful state $\omega := \sum_{n=1}^\infty 2^{-n} \omega_n$.

\begin{theorem}[$\mathrm{W}^*$-structure of $\Linf$]\label{thm:Linfstructure}
Assume $\MSp$ and $\MRp$ separable. Then:
\begin{enumerate}[label=(\roman*)]
    \item The measure $\mu_\E$ dominates $\E$ in the sense that for every $N \in \F$, $\E(N) = 0 \Leftrightarrow \mu_\E(N) = 0$.
    \item There is an isometric $*$-isomorphism
    \[
        \Linf \;\cong\; \MS \, \bar\otimes \, L^\infty_\E(\Sigma),
    \]
    where $L^\infty_\E(\Sigma):=L^\infty_{\mu_\E}(\Sigma)$, under which $\Linf$ is a $\mathrm{W}^*$-algebra.
    \item The predual is given by
    \[
        \Linf_* \;\cong\; \MSp \,\hat\otimes_\pi\, L^1_\E(\Sigma),
    \]
    where $\hat\otimes_\pi$ denotes the projective (completed) Banach-space tensor product \cite[Ch.~VIII]{diestel_vector_1977}, with the duality pairing
    \[
        \big\langle [f] , \, \rho \otimes h \big\rangle = \int_\Sigma \rho(f(x)) \, h(x) \, d\mu_\E(x)
    \]
    for $\rho \in \MSp$, $h \in L^1_\E(\Sigma)$.
\end{enumerate}
\end{theorem}

\begin{proof}
(i) For $N \in \F$: $\mu_\E(N) = 0 \Leftrightarrow \E_{\omega_n}(N) = 0$ for every $n$; since $\omega \mapsto \E_\omega(N) = \omega(\E(N))$ is norm-continuous on $\MRp$ and $\{\omega_n\}$ is dense in $\S(\MR)$, and normal states linearly span $\MRp$, this is equivalent to $\E_\omega(N) = 0$ for every $\omega \in \MRp$, i.e.\ $\E(N) = 0$. Any two such measures $\mu_\E, \mu_\E'$ share this property, hence are mutually absolutely continuous; we write $L^\infty_\E(\Sigma)=L^\infty_{\mu_\E}(\Sigma)$.

(ii) Write $X=\MSp$. Let $L^\infty_{w^*}(\mu_\E;X^*)$ denote the space of essentially norm-bounded weak-star measurable functions $f:\Sigma\to X^*$, modulo equality outside a single measurable $\mu_\E$-null set. Here weak-star measurability means that $x\mapsto\rho(f(x))$ is $\F$-measurable for every $\rho\in X$. Since $X$ is separable, the norm is measurable and scalar a.e.\ equality for every $\rho\in X$ is equivalent to equality off one null set, by a countable norming family. Part (i) therefore identifies $\Linf$ with $L^\infty_{w^*}(\mu_\E;X^*)$.

For any finite measure space and separable Banach space $X$, the weak-star measurable duality theorem gives the isometric identification
\[
    L^\infty_{w^*}(\mu_\E;X^*)=L^1(\mu_\E;X)^*,\qquad
    \langle [f],u\rangle=\int_\Sigma u(x)(f(x))\,d\mu_\E(x),
\]
where $L^1(\mu_\E;X)$ is the Bochner space.\footnote{No countable-generation assumption on $\F$ is needed. Indeed, a bounded functional on $L^1(\mu_\E;X)$, restricted to functions $h\rho$, gives a scalar $L^\infty(\mu_\E)$ density for each $\rho\in X$. Choose these densities on a countable dense $(\mathbb Q+i\mathbb Q)$-linear subspace of $X$. Outside one measurable null set they satisfy linearity and the common bound $|g_\rho(x)|\leq C\|\rho\|$. They thus extend pointwise to elements of $X^*$, giving the required weak-star measurable representative. Simple-function density gives the pairing and uniqueness. This argument uses only scalar Radon--Nikodym duality for a finite measure; completeness of the measure space is not required.}

The spatial tensor-product identification can be seen in a faithful normal representation of $\MS$ on a separable Hilbert space $H$. Each $[f]$ acts by multiplication on $L^2(\mu_\E;H)$, $(M_f\xi)(x)=f(x)\xi(x)$. These operators form an isometric $*$-representation. Its vector coefficients are pairings with Bochner $L^1(\mu_\E;\MSp)$ functions, so the representation is weak-star to ultraweak continuous. The image of its unit ball is consequently ultraweakly compact, and the represented algebra is a von Neumann algebra. Elementary functions $a\chi_A$ are weak-star linearly dense: a Bochner $L^1$ function annihilating all of them has integral zero on every measurable set and hence vanishes a.e. Thus this algebra is precisely the ultraweak closure of the elementary multiplication operators, namely $\MS\bar\otimes L_{\mu_\E}^\infty(\Sigma)$.

(iii) The preceding duality identifies the predual with the Bochner space
\[
    \Linf_* \cong L^1(\mu_\E;\MSp)\cong \MSp\hat\otimes_\pi L^1_{\mu_\E}(\Sigma)
\]
by the standard projective-tensor-product description of Bochner $L^1$ \cite[Ch.~VIII]{diestel_vector_1977}. An absolutely summable tensor series $\sum_k\rho_k\otimes h_k$ pairs with $[f]$ as $\sum_k\int_\Sigma\rho_k(f(x))h_k(x)\,d\mu_\E(x)$, which is well-defined because $\N_\E=\N_{\mu_\E}$.
\end{proof}

\begin{remark}[Classical case]\label{rem:LinfE-classical}
When $\MS = \MR = \Cn$ and $\E = \mu$ is a probability measure on $(\Sigma,\F)$, $\Bmes = \mathcal{B}(\Sigma,\F)$ is the space of complex-valued $\F$-measurable functions, $\Bb = B_b(\Sigma,\F)$, and $\Linf = L^\infty_\mu(\Sigma)$: the construction recovers the Lebesgue $L^\infty_\mu$-space, with its standard predual $L^1_\mu(\Sigma)$.\end{remark}

\section{Properties of the integration maps}\label{sec:propint}

We now investigate the properties of integration with respect to a fixed POVM, as a map from the domains developed in Section~\ref{sec:Bb}. We first treat the universal domain $\Bb$ --- where the integration map is a pointwise-normal unital CP map, recovering the $\E$-essential norm for localizable POVMs and multiplicative for sharp ones (Sec.~\ref{sec:propint-univ}). Next we show that this map descends to the quotient $\Linf$, where it acquires unconditional faithfulness and becomes isometric for localizable $\E$; under separability of $\MRp$, the descended map is a channel between $\mathrm{W}^*$-algebras (Sec.~\ref{sec:propint-desc}).

\subsection{Integration on the universal domain \texorpdfstring{$\Bb$}{Bb}}\label{sec:propint-univ}

\begin{theorem}[Integration map on $\Bb$]\label{thm:integrationmaps2}
Assume $\MSp$ separable. The \emph{integration map}
    \[
        \int_\Sigma d\E: \Bb \ni f \longmapsto \int_\Sigma f \otimes d\E \in \MSR
    \]
is positive, unital and adjoint-preserving.
It is moreover pointwise-normal: if $(f_n) \subset \Bb$ is a bounded sequence with $f_n(x) \xrightarrow{\rm uw} f(x)$ for each $x \in \Sigma$, then
\begin{equation}\label{eq:normality}
        \int_\Sigma f_n \otimes d\E \xrightarrow{\rm uw} \int_\Sigma f \otimes d\E \in \MSR.
    \end{equation}
If $\E$ is localizable, then
\begin{equation}\label{eq:operator-norm-recovery}
    \left\|\int_\Sigma f\otimes d\E\right\|=\|f\|_\E
    \qquad (f\in\Bb),
\end{equation}
so its kernel is exactly $\N_\E$. The map is multiplicative if $\E$ is sharp.
\end{theorem}

\begin{proof}
Since product normal states determine elements of $\MSR$ uniquely (Lemma~\ref{lem:extfromstates}), linearity, adjoint-preservation and unitality can be verified by testing against $\rho \otimes \omega$ with $\rho \in \S(\MS)$, $\omega \in \S(\MR)$. Positivity follows from the bounded-form construction.

\emph{Normality.} Choose faithful normal representations $\pi_\S: \MS \hookrightarrow B(\his)$ (with $\his$ separable) and $\pi_\R: \MR \hookrightarrow B(\hir)$. Let $(f_n) \subset \Bb$ with $\|f_n\|_\infty \leq C$ and $f_n(x) \xrightarrow{\rm uw} f(x)$ for each $x$. For each $\Omega \in \MSRp$ and $\epsilon > 0$, take $\tilde{\Omega} \in \T(\hisr)$ with $\tr[\tilde{\Omega}\,\cdot\,]|_{\MSR} = \Omega$ and $\|\tilde{\Omega}\|_1 \leq \|\Omega\| + \epsilon$. Approximate $\tilde{\Omega}$ by $\tilde{\Omega}_0 \in \T(\his) \otimes_{\rm alg} \T(\hir)$ with $\|\tilde{\Omega} - \tilde{\Omega}_0\|_1 < \epsilon$. Writing $\tilde{f}_n := \pi_\S \circ f_n$, $\tilde{f} := \pi_\S \circ f$, $\tilde{\E} := \pi_\R \circ \E$:
\begin{align*}
    \big|\Omega\big(\int (f_n-f) \otimes d\E\big)\big|
    &\leq \|\tilde{\Omega} - \tilde{\Omega}_0\|_1 \cdot \Big\|\int (\tilde{f}_n-\tilde{f}) \otimes d\tilde{\E}\Big\| + \Big|\tr\Big[\tilde{\Omega}_0 \int (\tilde{f}_n-\tilde{f}) \otimes d\tilde{\E}\Big]\Big| \\
    &\leq 8C \epsilon + \Big|\tr\Big[\tilde{\Omega}_0 \int (\tilde{f}_n-\tilde{f}) \otimes d\tilde{\E}\Big]\Big|.
\end{align*}
Writing $\tilde{\Omega}_0 = \sum_{k=1}^K \tilde{\rho}_k \otimes \tilde{\omega}_k$, the second term equals $\left|\sum_k \int_\Sigma \tr[\tilde{\rho}_k(\tilde{f}_n(x)-\tilde{f}(x))] \, d\tilde{\E}_{\tilde{\omega}_k}(x)\right|$, which vanishes by the dominated convergence theorem. Hence $\limsup_n |\Omega(\int (f_n-f) \otimes d\E)| \leq 8C\epsilon$. Since $\epsilon$ is arbitrary, \eqref{eq:normality} holds.

\emph{Linearity.} For $f,g \in \Bb$ and $\lambda \in \Cn$:
\[
    \int_\Sigma (f + \lambda g)_\rho \, d\E_\omega = \int_\Sigma f_\rho \, d\E_\omega + \lambda \int_\Sigma g_\rho \, d\E_\omega,
\]
so by Lemma~\ref{lem:extfromstates}, $\int_\Sigma (f + \lambda g) \otimes d\E = \int_\Sigma f \otimes d\E + \lambda \int_\Sigma g \otimes d\E$.

\emph{Adjoint-preservation.} For states $\rho \in \S(\MS)$ and $\omega \in \S(\MR)$ (which are self-adjoint as elements of $\MSp$ and $\MRp$ respectively):
\[
    (\rho \otimes \omega)\Big(\big(\int_\Sigma f \otimes d\E\big)^*\Big)
= \overline{(\rho \otimes \omega)\Big(\int_\Sigma f \otimes d\E\Big)}
= \overline{\int_\Sigma f_\rho \, d{\E}_{\omega}}
= \int_\Sigma \overline{f_\rho}\,d{\E}_{\omega}
= \int_\Sigma (f^*)_\rho\,d{\E}_{\omega},
\]
where we used that $\E_\omega$ is a positive (hence real) measure and that $\overline{\rho(f(x))} = \rho(f(x)^*)$ for a state $\rho$. By Lemma~\ref{lem:extfromstates}, $(\int f \otimes d\E)^* = \int f^* \otimes d\E$.

\emph{Positivity.} For $f\geq0$, Step~1 of the proof of Theorem~\ref{thm:intop} constructs the represented integral as the positive operator associated with the form of Lemma~\ref{lem:semiip}. The faithful representation $\pi_\S\otimes\pi_\R$ reflects positivity, so $\int_\Sigma f\otimes d\E\geq0$ in $\MSR$.

\emph{Unitality.} Taking $f = \algunit_{\MS}$: $(\rho \otimes \omega)(\int \algunit \otimes d\E) = \int_\Sigma \rho(\algunit_{\MS}) \, d\E_\omega = 1 = (\rho \otimes \omega)(\algunit_{\MSR})$ for all $\rho \in \S(\MS)$, $\omega \in \S(\MR)$. Since product states determine the identity uniquely (Lemma~\ref{lem:extfromstates}), $\int \algunit \otimes d\E = \algunit_{\MSR}$.

\emph{Norm recovery (localizable $\E$).} Let $\{\varphi_k\}_{k\in\Nn}$ be norm-dense in the unit ball of $\MSp$. Since this family norms $\MS$ and countable unions of $\E$-null sets are $\E$-null,
\[
    \|f\|_\E=\sup_k\|\varphi_k\circ f\|_\E.
\]
Testing the integral against the norm-at-most-one functionals $\varphi_k\otimes\omega$ and applying Lemma~\ref{lem:localizable-norm} gives
\[
    \left\|\int_\Sigma f\otimes d\E\right\|
    \geq\sup_{k,\,\omega\in\S(\MR)}
       \left|\int_\Sigma (\varphi_k\circ f)\,d\E_\omega\right|
    =\sup_k\|\varphi_k\circ f\|_\E=\|f\|_\E.
\]
For the reverse inequality, Lemma~\ref{lem:canonical} gives a bounded representative $f'\sim_\E f$ with $\|f'\|_\infty=\|f\|_\E$. The characterizing formula shows that $f$ and $f'$ have the same integral. Contractivity of the positive unital integration map therefore yields the upper bound. Finally, $\|f\|_\E=0$ is equivalent to $f\in\N_\E$, proving the kernel statement.

\emph{Multiplicativity (sharp $\E$).} We first establish
\begin{equation}\label{eq:mult}
    \Big(\int_\Sigma f \otimes d\E\Big)\Big(\int_\Sigma g \otimes d\E\Big) = \int_\Sigma fg \otimes d\E
\end{equation}
for simple $f \in \Bb$ and arbitrary $g \in \Bb$. By linearity in $f$, it suffices to prove \eqref{eq:mult} when $f = a\chi_X$ for $a \in \MS$ and $X \in \F$. In this case $\int f \otimes d\E = a \otimes \E(X)$. For each $\rho \in \MSp$ and $\omega \in \MRp$:
\begin{align*}
    (\rho \otimes \omega)\Big((a \otimes \E(X))\int_\Sigma g \otimes d\E\Big)
    &= \int_\Sigma g_{\rho \cdot a} \, d\E_{\omega \cdot \E(X)},
\end{align*}
where $\rho \cdot a \in \MSp$ is the right-module action \eqref{eq:bimodule}: $g_{\rho \cdot a}(x) = (\rho \cdot a)(g(x)) = \rho(a\,g(x))$. For a PVM, $\E(X)\E(Y) = \E(X \cap Y)$, so
\[
\E_{\omega \cdot \E(X)}(Y) = (\omega \cdot \E(X))(\E(Y)) = \omega(\E(X)\E(Y)) = \omega(\E(X \cap Y)) = \E_\omega(X \cap Y),
\]
i.e., $\E_{\omega \cdot \E(X)}$ is the restriction of $\E_\omega$ to $X$. This gives
\[
    \int_\Sigma g_{\rho \cdot a} \, d\E_{\omega \cdot \E(X)} = \int_X g_{\rho \cdot a} \, d\E_\omega = \int_\Sigma (fg)_\rho \, d\E_\omega,
\]
establishing \eqref{eq:mult} for simple $f$.

We now extend to arbitrary $f,g \in \Bb$. By Lemma~\ref{lem:simpleapprox}, take $\MS$-valued simple $f_n$ with $\|f_n\|_\infty \leq \|f\|_\infty$ and $f_n(x) \xrightarrow{\rm uw} f(x)$. By the pointwise-normality of the integration map \eqref{eq:normality},
\begin{equation}\label{eq:uwconv}
    \int_\Sigma f_n \otimes d\E \xrightarrow{\rm uw} \int_\Sigma f \otimes d\E \; \in \MSR.
\end{equation}
Combining \eqref{eq:mult} for simple $f_n$ with ultraweak convergence:
\begin{align*}
    (\rho \otimes \omega)\Big(\big(\int f \otimes d\E\big)\big(\int g \otimes d\E\big)\Big)
    &= \lim_{n \to \infty} (\rho \otimes \omega)\Big(\big(\int f_n \otimes d\E\big)\big(\int g \otimes d\E\big)\Big) \\
    &= \lim_{n \to \infty} \int_\Sigma (f_ng)_\rho \, d\E_\omega = \int_\Sigma (fg)_\rho \, d\E_\omega,
\end{align*}
where the first equality uses that multiplication in $\MSR$ is separately ultraweakly continuous (the sequence $\int f_n \otimes d\E$ converges ultraweakly and the factor $\int g \otimes d\E$ is fixed), and the last uses the dominated convergence theorem: $(f_n g)_\rho(x) = \rho(f_n(x)g(x)) \to \rho(f(x)g(x)) = (fg)_\rho(x)$ pointwise (since $f_n(x) \xrightarrow{\rm uw} f(x)$ and multiplication is separately ultraweakly continuous on bounded sets) and $|(f_n g)_\rho(x)| \leq \|f\|_\infty \|g\|_\infty$. This establishes \eqref{eq:mult} for all $f,g \in \Bb$.
\end{proof}

\begin{remark}[Monotone convergence]\label{rem:convergence}
The pointwise-normality of the integration map \eqref{eq:normality} immediately yields the operator-valued analogue of \emph{monotone convergence}: if $(f_n) \subset \Bb$ satisfies $0 \leq f_1(x) \leq f_2(x) \leq \ldots$ with $\sup_n \|f_n\|_\infty < \infty$, then $f(x) := \sup_n f_n(x)$ exists for each $x$ (every bounded increasing sequence in a $\mathrm{W}^*$-algebra has a supremum and converges to it ultraweakly \cite[Lem.~1.7.4]{sakai_Cstar_1971}), $f \in \Bb$, and $\int f_n \otimes d\E \xrightarrow{\rm uw} \int f \otimes d\E$.
\end{remark}

\subsection{Descent to \texorpdfstring{$\Linf$}{L-infty-E}}\label{sec:propint-desc}

When both preduals are separable, integration descends to a faithful unital normal CP map on the quotient $\Linf$.

\begin{theorem}[Descent to $\Linf$]\label{thm:descent}
Assume both $\MSp$ and $\MRp$ separable. The integration map descends to a faithful normal unital CP map between $\mathrm{W}^*$-algebras:\footnote{We write $\int_\Sigma d\E$ also for the descended map, with the domain clear from context.}

\[
\int_\Sigma d\E \;:\; \Linf \;\longrightarrow\; \MSR, \qquad [f] \longmapsto \int_\Sigma f \otimes d\E;
\]
it is isometric if $\E$ is localizable and an isometric $*$-homomorphism if $\E$ is sharp.\end{theorem}

\begin{proof}
The integration map $\int_\Sigma d\E: \Bb \to \MSR$ satisfies $\N_\E \subseteq \ker(\int_\Sigma d\E)$: if $f = 0$ $\E$-a.e., then $\int_\Sigma f_\rho \, d\E_\omega = 0$ for all $\rho, \omega$, so $\int_\Sigma f \otimes d\E = 0$ (product states are $\Cn$-linearly dense in $\MSRp$). Hence the map factors through the quotient $*$-homomorphism $q_\E: \Bb \twoheadrightarrow \Linf$:
\[
\Bb \;\xrightarrow{\;q_\E\;}\; \Linf \;\xrightarrow{\;\int_\Sigma d\E\;}\; \MSR.
\]

By Lemma~\ref{lem:canonical}, any $[f] \in \Linf$ has a representative $f' \in \Bb$ with $\|f'\|_\infty = \|[f]\|_\E$, so $\|\int_\Sigma f' \otimes d\E\| \leq \|f'\|_\infty = \|[f]\|_\E$, assuring contractivity. Unitality is immediate. If $\E$ is localizable, equation~\eqref{eq:operator-norm-recovery} gives $\|\int_\Sigma f\otimes d\E\|=\|[f]\|_\E$, so the descended map is isometric and hence injective.

\emph{Complete positivity.} For any $[f] \in M_n(\Linf)_+$, the functional calculus in $M_n(\Linf)$ provides $[f]^{1/2} \in M_n(\Linf)_+$. Lift $[f]^{1/2}$ to any $l \in M_n(\Bb)$ with $q_\E^{(n)}(l) = [f]^{1/2}$ (surjectivity of $q_\E^{(n)}$), and set $g := l^*l \in M_n(\Bb)_+$. Then $q_\E^{(n)}(g) = [f]^{1/2*}\,[f]^{1/2} = [f]$, so $g$ is a positive representative of $[f]$. Complete positivity of the universal-domain map (Corollary~\ref{cor:CP}) yields $\int_\Sigma g \otimes d\E \in M_n(\MSR)_+$, and well-definedness of the descent on equivalence classes gives $\int_\Sigma d\E^{(n)}([f]) = \int_\Sigma g \otimes d\E \in M_n(\MSR)_+$.

\emph{Faithfulness.} Suppose $[f] \geq 0$ in $\Linf$ and $\int_\Sigma f \otimes d\E = 0$, and choose a positive representative $f$ as above. For every $\rho \in \S(\MS)$ and $\omega \in \S(\MR)$, we have $\int_\Sigma f_\rho\,d\E_\omega=0$ with $f_\rho\geq0$. Thus $\E_\omega(\{f_\rho\ne0\})=0$ for every $\omega$, so $\{f_\rho\ne0\}$ is $\E$-null. Lemma~\ref{lem:singlenull} gives $f=0$ $\E$-a.e., hence $[f]=0$. This argument uses no separability assumption on $\MRp$.

\emph{Isometry and multiplicativity for sharp $\E$.} For a PVM, the map $\int_\Sigma d\E: \Bb \to \MSR$ is a $*$-homomorphism. Since $q_\E$ is a surjective $*$-homomorphism, the descended map $\int_\Sigma d\E: \Linf \to \MSR$ is also a $*$-homomorphism. Faithfulness then gives injectivity: if $\int_\Sigma f \otimes d\E = 0$, then $\int_\Sigma f^*f \otimes d\E = (\int_\Sigma f \otimes d\E)^*(\int_\Sigma f \otimes d\E) = 0$ with $[f^*f] \geq 0$, so $[f^*f] = 0$, hence $[f] = 0$. An injective $*$-homomorphism between $\mathrm{C}^*$-algebras is isometric.

\emph{Normality.} Put $\mathcal A=\Linf$ and denote the descended map by $I_\E$. First let $0\leq a_n\uparrow a$ in $\mathcal A$ with $\|a_n\|\leq C$. Choose positive representatives $f_n\in\Bb$. The order inequalities and the essential bound hold outside measurable $\E$-null sets; removing their countable union and setting all representatives to zero there gives
\[
    0\leq f_n(x)\leq f_{n+1}(x)\leq C\algunit \qquad(x\in\Sigma).
\]
Then $g(x)=\sup_n f_n(x)$ is ultraweakly measurable and bounded. Its class is the supremum of the $[f_n]$: any upper bound dominates every $f_n$ off a common null set and hence dominates $g$ there. Thus $[g]=a$, and pointwise-normality (Theorem~\ref{thm:integrationmaps2}) gives $I_\E(a_n)\uparrow I_\E(a)$.

To pass from sequences to arbitrary nets, choose a faithful normal state $\rho_0$ on $\MS$, which exists since $\MSp$ is separable. By Theorem~\ref{thm:Linfstructure},
\[
    \varphi([f])=\int_\Sigma\rho_0(f(x))\,d\mu_\E(x)
\]
is a faithful normal state on $\mathcal A$: it belongs to the displayed predual, and positivity and faithfulness follow from those of $\rho_0$ and the equality of the null ideals. In particular, $\mathcal A$ is $\sigma$-finite; its predual need not be separable. For a bounded increasing net $0\leq a_\alpha\uparrow a$, normality of $\varphi$ and directedness allow increasing indices $\alpha_n$ with $\varphi(a-a_{\alpha_n})<1/n$. If $b=\sup_n a_{\alpha_n}$, then $b\leq a$ and $\varphi(a-b)=0$, so faithfulness gives $b=a$. Sequential preservation of suprema and positivity now imply
\[
    I_\E(a)=\sup_n I_\E(a_{\alpha_n})
    \leq\sup_\alpha I_\E(a_\alpha)\leq I_\E(a).
\]
Hence $I_\E$ preserves bounded increasing suprema and is normal.
\end{proof}

\begin{corollary}[Functoriality of the descended map]\label{cor:descent-functorial}
Assume that $\MSp$, $\MRp$, $(\N_\S)_*$ and $(\N_\R)_*$ are separable. The descended integration map of Theorem~\ref{thm:descent} remains functorial: for any measurable $\alpha: (\Sigma,\F) \to (\Sigma',\F')$, normal channels $\Psi: \MS \to \N_\S$ and $\Phi: \MR \to \N_\R$, and any $f \in L^\infty_{\alpha_*\E}(\Sigma',\MS)$,
\begin{equation}\label{eq:pushfwd-descent}
    \int_\Sigma (\Psi \circ f \circ \alpha) \otimes d(\Phi \circ \E)
    \;=\;
    (\Psi \otimes \Phi)\Big(\int_{\Sigma'} f \otimes d(\alpha_*\E)\Big).
\end{equation}
\end{corollary}

\begin{proof}
By Lemma~\ref{lem:canonical}, every class admits a bounded representative. The pullback $f \mapsto f \circ \alpha$ along $\alpha$ pulls $(\alpha_*\E)$-null sets in $\Sigma'$ back to $\E$-null sets in $\Sigma$ (since $\alpha_*\E(N) = \E(\alpha^{-1}(N))$), and any $\E$-null set is also $(\Phi \circ \E)$-null (since $\Phi$ sends zero to zero). Since $\Psi$ is normal and contractive, post-composition with $\Psi$ preserves ultraweak measurability and boundedness, as well as equality outside these null sets. Thus $f\mapsto\Psi\circ f\circ\alpha$ defines a well-defined map
\[
    L^\infty_{\alpha_*\E}(\Sigma',\MS)
    \longrightarrow
    L^\infty_{\Phi \circ \E}(\Sigma,\N_\S).
\]
The identity \eqref{eq:pushfwd-descent} now follows by applying \eqref{eq:pushfwd} to a bounded representative and using the definition of the descended integration maps.
\end{proof}

\begin{corollary}[Predual of the descended integration map]\label{cor:predual}
Assume $\MSp$ and $\MRp$ separable. Let $\mu_\E$ be a dominating probability measure on $(\Sigma,\F)$ and identify $\Linf \cong \MS \, \bar\otimes \, L^\infty_\E(\Sigma)$ with predual $\Linf_* \cong \MSp \,\hat\otimes_\pi\, L^1_\E(\Sigma)$ (Thm.~\ref{thm:Linfstructure}). For each $\omega \in \MRp$, let $h_\omega := d\E_\omega/d\mu_\E \in L^1_\E(\Sigma)$ denote the Radon--Nikodym density of $\E_\omega = \omega \circ \E$ with respect to $\mu_\E$. Then the predual of the descended integration map $\int_\Sigma d\E: \Linf \to \MSR$ is the unique bounded linear map
\[
    \Big(\int_\Sigma d\E\Big)_*: \MSRp \;\longrightarrow\; \MSp \,\hat\otimes_\pi\, L^1_\E(\Sigma)
\]
satisfying
\begin{equation}\label{eq:predual}
    \Big(\int_\Sigma d\E\Big)_*\!(\rho \otimes \omega) \;=\; \rho \otimes h_\omega, \qquad \rho \in \S(\MS), \;\omega \in \S(\MR).
\end{equation}
\end{corollary}

\begin{proof}
Normality of the integration map (Thm.~\ref{thm:descent}) gives a unique bounded predual $(\int_\Sigma d\E)_*: (\MSR)_* \to \Linf_*$ characterized by $\sigma\big(\int_\Sigma f \otimes d\E\big) = (\int_\Sigma d\E)_*(\sigma)([f])$ for all $\sigma \in (\MSR)_*$ and $[f] \in \Linf$. By Lemma~\ref{lem:proddense}, the linear span of product normal states is norm-dense in $\MSRp$, so these values determine the predual map uniquely. For $\rho \in \S(\MS)$, $\omega \in \S(\MR)$, and $[f] \in \Linf$, the characterizing formula and the Radon--Nikodym identity $d\E_\omega = h_\omega \, d\mu_\E$ give
\begin{align*}
    (\rho \otimes \omega)\Big(\int_\Sigma f \otimes d\E\Big)
    &= \int_\Sigma \rho(f(x)) \,d\E_\omega(x) \\
    &= \int_\Sigma \rho(f(x)) \,h_\omega(x)\,d\mu_\E(x)
     = \langle [f],\rho \otimes h_\omega\rangle.
\end{align*}
where the last pairing is that of $\MSp \,\hat\otimes_\pi\, L^1_\E(\Sigma)$ with $\MS \, \bar\otimes \, L^\infty_\E(\Sigma) \cong \Linf$. This proves \eqref{eq:predual}; by linearity, the same formula holds for all $\rho\in\MSp$ and $\omega\in\MRp$.
\end{proof}

\section{Integration and monoidal structure}\label{sec:monoidal}

Theorem \ref{thm:descent} asserts that, under separability of preduals, the descended integration map is a morphism in the category of $\mathrm{W}^*$-algebras and normal CP subunital maps ($\mathrm{W}^*$-morphisms). In this section we  show that normalized POVMs $\E: (\Sigma,\F) \to \Eff(\MR)$ are in bijection with faithful normal unital CP maps $\Phi_\E: L^\infty_\E(\Sigma) \to \MR$ and identify the descended integration maps $L^\infty_\E(\Sigma,\MS) \to \MSR$ with the amplification $\id_{\MS} \otimes \Phi_\E$.

\subsection{POVMs as \texorpdfstring{$\mathrm{W}^*$}{W*}-morphisms}

For $\MR$ with separable predual and $\MS = \Cn$ the space $L^\infty_\E(\Sigma, \MS)$ specializes to the classical Lebesgue space $L^\infty_\E(\Sigma) $ (Rem.~\ref{rem:LinfE-classical}). It is a commutative $\mathrm{W}^*$-algebra under the $\E$-essential supremum norm, naturally associated with the POVM~$\E$. The following proposition identifies normalized POVMs on $(\Sigma,\F)$ whose null ideals contain a fixed $\sigma$-ideal $\N$ with normal unital CP maps out of the $\sigma$-finite commutative von Neumann algebra $L^\infty_\N(\Sigma)$. Faithfulness corresponds precisely to $\N=\N_\E$.\footnote{The statement and proof of Prop. \ref{prop:povm-bijection} were developed during the Stromboli Summer Intensive 2025 workshop organized by Andrea Di Biagio and Guilherme Franzmann and endorsed by the Basic Research Community for Physics (BRCP).}

\begin{proposition}[POVM--morphism bijection]\label{prop:povm-bijection}
Assume $\MRp$ separable. Fix a measurable space $(\Sigma,\F)$ and a $\sigma$-ideal $\N \subseteq \F$ such that $L^\infty_{\N}(\Sigma)$ is a $\sigma$-finite von Neumann algebra.\footnote{This holds, in particular, whenever $\N$ is the null ideal of a scalar probability measure.} There is a bijection between
\begin{itemize}
    \item normalized POVMs $\E:(\Sigma,\F)\to\Eff(\MR)$ such that $\N\subseteq\N_\E$, and
    \item normal unital positive linear maps $\Phi:L^\infty_{\N}(\Sigma)\to\MR$,
\end{itemize}
where $L^\infty_{\N}(\Sigma)\cong B_b(\Sigma)/\!\sim_\N$ denotes the space of $\N$-equivalence classes of essentially bounded scalar functions. The map $\Phi_\E$ is automatically completely positive; it is faithful if and only if $\N=\N_\E$. The correspondence is characterized by
\[
    \Phi_\E(\chi_X)=\E(X), \qquad X\in\F.
\]
\end{proposition}

\begin{proof}
Take a normalized POVM $\E$ with $\N\subseteq\N_\E$.
For every normal state $\omega\in\S(\MR)$, the probability measure
$\E_\omega:=\omega\circ\E$ vanishes on $\N$.
On simple functions $f=\sum_j\lambda_j\chi_{X_j}$, with the
$X_j$ disjoint, define
\[
    \Phi_\E(f):=\sum_j\lambda_j\E(X_j).
\]
Since $\E$ vanishes on $\N$, this assignment is well defined on
$\N$-equivalence classes. It is linear, positive and unital on
simple functions, and satisfies
\[
    \omega(\Phi_\E(f))=\int_\Sigma f\,d\E_\omega.
\]
For a real-valued simple function $g$, the operator $\Phi_\E(g)$
is self-adjoint, so
\[
    \|\Phi_\E(g)\|
    =\sup_{\omega\in\S(\MR)}
       \left|\int_\Sigma g\,d\E_\omega\right|
    \leq \|g\|_\N,
\]
where $\|\cdot\|_\N$ denotes the quotient norm.
Decomposing an arbitrary simple function into its real and
imaginary parts therefore gives
\[
    \|\Phi_\E(f)\|
    \leq \|\mathrm{Re}\,f\|_\N+\|\mathrm{Im}\,f\|_\N
    \leq 2\|f\|_\N.
\]
Since simple functions are norm dense in $L^\infty_\N(\Sigma)$,
the assignment extends uniquely to a bounded linear map
$\Phi_\E:L^\infty_\N(\Sigma)\to\MR$.
Every positive element admits a nonnegative bounded representative
and can be approximated in norm by nonnegative simple functions.
Thus the extension is positive; it is also unital, and hence
contractive. Complete positivity follows from commutativity
of the domain.

\emph{Normality.}
For each normal state $\omega\in\S(\MR)$, it suffices to prove
that the positive functional $\omega\circ\Phi_\E$ is normal.
By \cite[Cor.~III.3.11]{takesaki2001theory}, this is equivalent to
\begin{equation}\label{eq:normalityproj}
    \omega\!\left(\Phi_\E\!\left(\sum_jp_j\right)\right)
    =\sum_j\omega(\Phi_\E(p_j))
\end{equation}
for every mutually orthogonal family of projections in
$L^\infty_\N(\Sigma)$.
Since this algebra is $\sigma$-finite, only countably many
members of such a family can be nonzero.

After discarding the zero projections, choose representatives
$p_j=\chi_{U_j}$ with the $U_j$ disjoint, removing
$\N$-null overlaps if necessary. Then $\sum_jp_j=\chi_{\bigcup_jU_j}$ in $L^\infty_\N(\Sigma)$,
and $\sigma$-additivity of $\E$ gives
\[
    \omega\!\left(\Phi_\E\!\left(\sum_jp_j\right)\right)
    =\omega\!\left(\E\!\left(\bigcup_jU_j\right)\right)
    =\sum_j\omega(\E(U_j))
    =\sum_j\omega(\Phi_\E(p_j)).
\]
Thus $\omega\circ\Phi_\E$ is normal for every normal state
$\omega$, and consequently $\Phi_\E$ is normal.

\emph{Conversely.}
Take a normal unital positive linear map
$\Phi:L^\infty_\N(\Sigma)\to\MR$ and define
\[
    \E_\Phi(X):=\Phi(\chi_X), \qquad X\in\F.
\]
Positivity and unitality give
$0\leq\E_\Phi(X)\leq\algunit_{\MR}$,
$\E_\Phi(\emptyset)=0$ and $\E_\Phi(\Sigma)=\algunit_{\MR}$.
If $(X_n)$ is a disjoint sequence, then
\[
    \sum_{n=1}^{m}\chi_{X_n}
    \uparrow \chi_{\bigcup_nX_n}
\]
in $L^\infty_\N(\Sigma)$. Normality of $\Phi$ therefore gives
\[
    \sum_{n=1}^{m}\E_\Phi(X_n)
    \xrightarrow{\rm uw}
    \E_\Phi\!\left(\bigcup_nX_n\right).
\]
Hence $\E_\Phi$ is a normalized POVM.
For $X\in\N$, the element $\chi_X$ is zero in
$L^\infty_\N(\Sigma)$, so $\E_\Phi(X)=0$.
Thus $\N\subseteq\N_{\E_\Phi}$.

\emph{Inverses.}
For every $X\in\F$,
\[
    \E_{\Phi_\E}(X)=\Phi_\E(\chi_X)=\E(X),
    \qquad
    \Phi_{\E_\Phi}(\chi_X)=\E_\Phi(X)=\Phi(\chi_X).
\]
The second equality extends by linearity and norm continuity
from characteristic functions to all of $L^\infty_\N(\Sigma)$.

\emph{Faithfulness.}
Suppose first that $\N=\N_\E$.
Let $f\geq0$ in $L^\infty_\N(\Sigma)$ satisfy $\Phi_\E(f)=0$,
and choose a nonnegative bounded representative.
For every $n\geq1$, positivity gives
\[
    0\leq \frac1n\,\E(\{f\geq1/n\})
    \leq \Phi_\E(f)=0.
\]
Thus each set $\{f\geq1/n\}$ belongs to $\N$.
Their countable union is $\{f>0\}$, so $f=0$ in
$L^\infty_\N(\Sigma)$, proving faithfulness.

If instead $\N\subsetneq\N_\E$, choose
$X\in\N_\E\setminus\N$.
Then $\chi_X$ is a nonzero positive element of
$L^\infty_\N(\Sigma)$, while
$\Phi_\E(\chi_X)=\E(X)=0$.
Hence $\Phi_\E$ is not faithful.
\end{proof}

\begin{corollary}
In the setting of Proposition~\ref{prop:povm-bijection}, assume
additionally that $\N=\N_\E$. Under the canonical identification
$\Cn\bar\otimes\MR\cong\MR$, the corresponding map is
\[
    \Phi_\E
    =
    \int_\Sigma d\E:
    L^\infty_\N(\Sigma)=L^\infty_\E(\Sigma)
    \longrightarrow \MR.
\]
Thus $\E$ is sharp if and only if $\Phi_\E$ is multiplicative,
and localizable if and only if $\Phi_\E$ is isometric
(hence injective).
\end{corollary}

\begin{proof}
The maps agree on characteristic functions, since
\[
    \Phi_\E(\chi_X)=\E(X)=\int_\Sigma\chi_X\,d\E.
\]
Linearity, norm continuity and density of simple functions
give the identification on all of $L^\infty_\E(\Sigma)$.\footnote{We could have defined $\Phi_\E$ via Theorem~\ref{thm:descent} applied with $\MS=\Cn$. We established an independent construction for the purpose of clarifying the relation between $\mathrm{W}^*$-algebraic integration theory and monoidal structure on the category of $\mathrm{W}^*$-algebras and morphisms. The presented construction can be seen as an extension of integration theory of scalar functions with POVMs developed in~\cite{busch_quantum_2016}.}

The forward implications follow from Theorem~\ref{thm:descent}.
Conversely, if $\Phi_\E$ is multiplicative, then
\[
    \E(X)^2
    =\Phi_\E(\chi_X)^2
    =\Phi_\E(\chi_X^2)
    =\E(X)
\]
for every $X\in\F$. Since $\E(X)$ is positive, it is a projection,
so $\E$ is sharp.

If $\Phi_\E$ is isometric and $\E(X)\neq0$, then
$X\notin\N_\E=\N$, and therefore
\[
    \|\E(X)\|
    =\|\Phi_\E(\chi_X)\|
    =\|\chi_X\|_\E
    =1.
\]
Thus $\E$ has the norm-$1$ property and is localizable.
\end{proof}

\begin{remark}[Related work]\label{POVM morphism int}
The identification of POVMs with quantum-classical channels appears in various restricted settings in the literature, perhaps the closest result being due to Roumen \cite{roumen_categorical_2014} who established a bijection between continuous POVMs on compact Hausdorff spaces and normal positive unital maps from commutative to type I von Neumann algebras. For POVMs taking values in $B(\hi)$, the associated scalar integration channels are discussed by Kuramochi \cite[Sec.~V.1]{kuramochi_quantum_2018} in the context of channel compatibility. Here the receiving algebra is an arbitrary von Neumann algebra with separable predual. Proposition~\ref{prop:povm-bijection} provides the bijection in the $\mathrm{W}^*$-algebraic setting of this paper, including POVMs on non-compact (even non-topological) spaces.
\end{remark}

\begin{remark}[Operational interpretation]\label{rem:povm-morphism-physics}
Under the bijection, the morphism $\Phi_\E: L^\infty_{\N}(\Sigma) \to \MR$ is the (Heisenberg-picture) quantum-classical channel implementing the measurement $\E$: a bounded measurable classical observable $f \in L^\infty_\N(\Sigma)$ is mapped to its quantum counterpart $\int f\,d\E \in \MR$. The predual map $(\Phi_\E)_*: \MRp \to L^1_{\N}(\Sigma)$ sends each normal state $\omega$ to the density of its outcome probability measure $\E_\omega=\omega\circ\E$ with respect to a probability measure (cf.\ Corollary~\ref{cor:predual}).
\end{remark}

\subsection{Integration as monoidal lifting}

The spatial $\mathrm{W}^*$-tensor product $\bar\otimes$ is bifunctorial on normal CP subunital maps ($\mathrm{W}^*$-morphisms): for any pair $\Psi: \MS \to \N_\S$ and $\Phi: \MR \to \N_\R$ of such maps, there exists a unique $\mathrm{W}^*$-morphism
\[
    \Psi \,\bar\otimes\, \Phi: \M_\S \,\bar\otimes\, \M_\R  \longrightarrow \N_\S \,\bar\otimes\, \N_\R
\]
satisfying $(\Psi \,\bar\otimes\, \Phi)(a \otimes b) = \Psi(a) \otimes \Phi(b)$ on elementary tensors, and this assignment endows the category of $\mathrm{W}^*$-algebras and morphisms with a symmetric monoidal structure with tensor unit $\Cn$ \cite[Sec.~4.2]{westerbaan_thesis_2018}. Applied to $\Psi = \id_{\MS}$ and $\Phi = \Phi_\E$, this yields a normal unital CP map
\[
    \id_{\MS} \,\bar\otimes\, \Phi_\E\;:\; \MS \,\bar\otimes\, L^\infty_\E(\Sigma) \longrightarrow \MS \,\bar\otimes\, \MR.
\]
Using the identification $\Linf \cong \MS \,\bar\otimes\, L^\infty_\E(\Sigma) $ of Theorem~\ref{thm:Linfstructure}(ii) (up to symmetry of $\bar\otimes$), this is a normal unital CP map $\Linf \to \MSR$. As shown below, this is exactly the descended integration map.

\begin{theorem}[Tensor-product characterization of integration]\label{thm:int-categorical}
Assume $\MSp$ and $\MRp$ separable. Under the identification $\Linf \cong \MS \,\bar\otimes\, L^\infty_\E(\Sigma) $ of Theorem~\ref{thm:Linfstructure}(ii) and the bijection of Proposition~\ref{prop:povm-bijection} we have
\[
    \int_\Sigma d\E \;\cong\; \id_{\MS} \,\bar\otimes\, \Phi_\E\;:\; \MS \,\bar\otimes\, L^\infty_\E(\Sigma) \;\longrightarrow\; \MS \,\bar\otimes\, \MR.
\]
\end{theorem}

\begin{proof}
Both maps are normal unital CP maps between the same $\mathrm{W}^*$-algebras. For $a\in\MS$ and $X\in\F$, Theorem~\ref{thm:Linfstructure}(ii) identifies $a\otimes[\chi_X]$ with $[a\chi_X]$. The elementary integral formula of Theorem~\ref{thm:intop} and the definition of $\Phi_\E$ give
\[
    \int_\Sigma a\chi_X\otimes d\E
    =a\otimes\E(X)
    =(\id_{\MS}\,\bar\otimes\,\Phi_\E)(a\otimes[\chi_X]).
\]
By linearity, the maps agree on finite sums of such tensor indicators. Scalar simple functions are norm dense in $L^\infty_\E(\Sigma)$, so this span is ultraweakly dense in $\MS\bar\otimes L^\infty_\E(\Sigma)$. Normality therefore gives equality everywhere.
\end{proof}

\begin{remark}
	Abstract amplification of a normal completely positive map by an identity map does not require separability of preduals. Interpreting this amplification as integration on $\Linf$, however, uses the function-space identification of Theorem~\ref{thm:Linfstructure}(ii), established here under separability of both $\MSp$ and $\MRp$. The direct construction in Theorem~\ref{thm:intop} requires separability only of $\MSp$ and characterizes the resulting operators by their product-state evaluations.
\end{remark}

\begin{corollary}
Under the bijection of Proposition~\ref{prop:povm-bijection}, the existence of the descended operator-valued integration map as a normal unital CP map is thus a direct consequence of bifunctoriality of $\bar\otimes$ in the $\mathbf{W}^*$-category \cite{westerbaan_thesis_2018} applied to the morphism $\Phi_\E: L^\infty_\E(\Sigma) \to \MR$. Under the predual separability hypotheses of Corollary~\ref{cor:descent-functorial}, functoriality under normal channels also follows readily from this characterization. Let $q:L^\infty_\E(\Sigma)\to L^\infty_{\Phi\circ\E}(\Sigma)$ be the canonical normal quotient map, so that $\Phi\circ\Phi_\E=\Phi_{\Phi\circ\E}\circ q$. Then
\[
(\Psi \bar\otimes \Phi) \int_\Sigma f \otimes d\E = (\Psi \bar\otimes \Phi_{\Phi\circ\E})(\id_{\MS}\bar\otimes q)[f]_\E = \int_\Sigma \Psi \circ f \otimes d(\Phi \circ \E).
\]
Thus, in the light of the categorical results of \cite{westerbaan_thesis_2018}, existence, uniqueness and functorial properties of the descended integration map follow from our Proposition~\ref{prop:povm-bijection}.
\end{corollary}

\section[Stinespring factorization]{Stinespring factorization}\label{sec:CP}

In this section, we establish a Stinespring factorization \cite{stinespring_positive_1955} of the integration map through the Naimark dilation of the POVM, from which complete positivity follows as a corollary. Since the Naimark dilation theorem (in its measurable-space formulation recalled in \cite[Sec.~1]{stinespring_positive_1955}) applies to POVMs on Hilbert spaces, we work with the concrete POVM $\tilde{\E} := \pi_\R \circ \E$ obtained from a faithful normal representation $\pi_\R: \MR \hookrightarrow B(\hir)$; we continue to use the tilde notation to distinguish such represented objects from the abstract ones. Recall that a \emph{minimal Naimark dilation} of a POVM $\tilde{\E}: \F \to \Eff(\hi)$ (on a Hilbert space) is a triple $(\hik, V, \hat{\E})$, where $\hik$ is a (not necessarily separable) Hilbert space, $\hat{\E}: \F \to B(\hik)$ is a PVM, and $V: \hi \to \hik$ is a linear isometry such that $\tilde{\E}(X) = V^*\hat{\E}(X)V$ for all $X \in \F$, and $\overline{\rm span}\{\hat{\E}(X)V\hi : X \in \F\} = \hik$. The following lemma records the compatibility of minimal Naimark dilation with operators commuting with the POVM.

\begin{lemma}\label{lem:naimarkrep}
Let $\tilde{\E}: \F \to \Eff(\hi)$ be a POVM with minimal Naimark dilation $(\hik,V,\hat{\E})$. Then there exists a unique normal representation $\pi: \tilde{\E}(\F)' \to B(\hik)$ such that for each $A \in \tilde{\E}(\F)'$, $X \in \F$ and $\eta \in \hi$:
\begin{equation}\label{eq:pidef}
    \pi(A)\hat{\E}(X)V\eta = \hat{\E}(X)VA\eta.
\end{equation}
Moreover, $\pi(A)V = VA$ and $\pi(A)\hat{\E}(X) = \hat{\E}(X)\pi(A)$ for all $X \in \F$.
\end{lemma}

\begin{proof}
We show that the right-hand side of \eqref{eq:pidef} determines a bounded operator. Set $\xi := \sum_{j=1}^n \hat{\E}(X_j)V\eta_j$ and $\xi_A := \sum_{j=1}^n \hat{\E}(X_j)VA\eta_j$. We compute
\[
    \|\xi_A\|^2 = \sum_{j,k=1}^n \braket{VA\eta_j}{\hat{\E}(X_j)\hat{\E}(X_k)VA\eta_k} = \sum_{j,k=1}^n \braket{A\eta_j}{\tilde{\E}(X_j \cap X_k)A\eta_k}.
\]
The matrix $C := (\tilde{\E}(X_j \cap X_k))_{j,k=1}^n$ (with entries in $\tilde{\E}(\F)'' \subseteq B(\hi)$) is positive semidefinite. Its square root $D = (D_{jk}) = \sqrt{C}$ has entries in $\tilde{\E}(\F)''$, hence commuting with $A \in \tilde{\E}(\F)'$. Thus
\begin{align*}
    \|\xi_A\|^2 &= \sum_{j,k,\ell=1}^n \braket{\eta_j}{D_{\ell j}^*A^*AD_{\ell k}\eta_k} = \sum_{\ell=1}^n \big\|A\sum_{j=1}^n D_{\ell j}\eta_j\big\|^2 \\
    &\leq \|A\|^2 \sum_{\ell=1}^n \big\|\sum_{j=1}^n D_{\ell j}\eta_j\big\|^2 = \|A\|^2 \sum_{j,k=1}^n \braket{\eta_j}{\tilde{\E}(X_j \cap X_k)\eta_k} = \|A\|^2\|\xi\|^2.
\end{align*}
This shows $\|\xi_A\| \leq \|A\| \cdot \|\xi\|$, so $\pi(A)$ is well-defined and bounded on ${\rm span}\{\hat{\E}(X)V\eta\}$, extending by continuity to $\hik$ by minimality. Uniqueness follows from minimality.

Setting $X = \Sigma$ in \eqref{eq:pidef} gives $\pi(A)V\eta = VA\eta$, i.e., $\pi(A)V = VA$. For commutativity with $\hat{\E}$:
\[
    \pi(A)\hat{\E}(X)\hat{\E}(Y)V\eta = \pi(A)\hat{\E}(X \cap Y)V\eta = \hat{\E}(X \cap Y)VA\eta = \hat{\E}(X)\pi(A)\hat{\E}(Y)V\eta,
\]
and minimality gives $\pi(A)\hat{\E}(X) = \hat{\E}(X)\pi(A)$. That $\pi$ is a $*$-representation: for $A,B \in \tilde{\E}(\F)'$:
\begin{align*}
    \braket{\hat{\E}(X)V\eta}{\pi(A^*B)\hat{\E}(Y)V\zeta} &= \braket{\eta}{\tilde{\E}(X \cap Y)A^*B\zeta} = \braket{A\eta}{\tilde{\E}(X \cap Y)B\zeta}\\
    &= \braket{\hat{\E}(X)VA\eta}{\hat{\E}(Y)VB\zeta} = \braket{\pi(A)\hat{\E}(X)V\eta}{\pi(B)\hat{\E}(Y)V\zeta},
\end{align*}
so $\pi(A^*B) = \pi(A)^*\pi(B)$. To prove normality, let $0\leq A_\alpha\uparrow A$ in $\tilde{\E}(\F)'$. Then $A_\alpha\to A$ strongly, and~\eqref{eq:pidef} gives
\[
    \pi(A_\alpha)\hat{\E}(X)V\eta
    =\hat{\E}(X)VA_\alpha\eta
    \longrightarrow\hat{\E}(X)VA\eta
    =\pi(A)\hat{\E}(X)V\eta
\]
in norm. By linearity, convergence holds on the span of these vectors, which is dense by minimality. Since $\pi$ is a $*$-representation, $0\leq\pi(A_\alpha)\leq\pi(A)$, so the operators are uniformly bounded and convergence extends to all of $\hik$. Consequently $\pi(A_\alpha)\uparrow\pi(A)$, proving normality.
\end{proof}

\begin{theorem}[Stinespring factorization]\label{thm:stinespring}
Assume $\MSp$ separable and choose faithful normal representations $\pi_\S: \MS \hookrightarrow B(\his)$ $($with $\his$ separable$)$ and $\pi_\R: \MR \hookrightarrow B(\hir)$. Let $\tilde{\E} := \pi_\R \circ \E: \F \to \Eff(\hir)$ and let $(\hik,V,\hat{\E})$ be a minimal Naimark dilation of $\tilde{\E}$. Then for every $f \in \Bb$:
\begin{equation}\label{eq:stinespring}
    (\pi_\S \otimes \pi_\R)\Big(\int_\Sigma f \otimes d\E\Big) = (\id_{\his} \otimes V^*)\Big(\int_\Sigma (\pi_\S \circ f) \otimes d\hat{\E}\Big)(\id_{\his} \otimes V).
\end{equation}
\end{theorem}

\begin{proof}
Since $\hat{\E}$ is a PVM, the integration map $\int_\Sigma d\hat{\E}: \BbBH \to B(\his \otimes \hik)$ is a pointwise-normal unital $*$-homomorphism by Theorem~\ref{thm:integrationmaps2}.\footnote{The multiplicativity proof uses separability of $\his$ $($via Lemma~\ref{lem:simpleapprox} applied to $B(\his)$$)$, but does not require separability of $\hik$.} Since $\pi_\S(\MS) \subseteq B(\his)$ and any $f \in \Bb$ gives $\pi_\S \circ f \in \BbBH$, the right-hand side of \eqref{eq:stinespring} is well-defined.

To verify \eqref{eq:stinespring}, test both sides against $\tilde{\rho} \otimes \tilde{\omega}$ with $\tilde{\rho} \in \T(\his)$, $\tilde{\omega} \in \T(\hir)$:
\begin{align*}
    &\tr\Big[\tilde{\rho} \otimes \tilde{\omega} \, (\id_{\his} \otimes V^*)\big(\int (\pi_\S \circ f) \otimes d\hat{\E}\big)(\id_{\his} \otimes V)\Big] \\
    &\; = \tr\Big[\tilde{\rho} \otimes V\tilde{\omega} V^* \int (\pi_\S \circ f) \otimes d\hat{\E}\Big]
    = \int_\Sigma \tr[\tilde{\rho}\,\pi_\S(f(x))] \, d\hat{\E}_{V\tilde{\omega} V^*}(x).
\end{align*}
Now $\hat{\E}_{V\tilde{\omega} V^*}(X) = \tr[V\tilde{\omega} V^*\hat{\E}(X)] = \tr[\tilde{\omega}\,V^*\hat{\E}(X)V] = \tr[\tilde{\omega}\,\tilde{\E}(X)] = \tilde{\E}_{\tilde{\omega}}(X)$, so the right-hand side equals $\int_\Sigma f_\rho \, d\E_\omega$ where $\rho(\cdot) = \tr[\tilde{\rho}\,\pi_\S(\cdot)] \in \MSp$ and $\omega(\cdot) = \tr[\tilde{\omega}\,\pi_\R(\cdot)] \in \MRp$. Theorem \ref{thm:intop} concludes.
\end{proof}

\begin{remark}\label{rem:alternative}
The factorization \eqref{eq:stinespring} provides an independent consistency check: the integral constructed from bounded sesquilinear forms in Theorem~\ref{thm:intop} necessarily agrees with the one obtained by dilating, integrating against a PVM, and compressing. It also suggests an alternative, equivalent construction of the integral. Given a POVM $\tilde{\E}$ on a Hilbert space $\hir$ with Naimark dilation $(\hik, V, \hat{\E})$ and a uniformly bounded ultraweakly measurable $\tilde{f}: \Sigma \to B(\his)$ one may \emph{define} $\int_\Sigma \tilde{f} \otimes d\tilde{\E} := (\id_{\his}\otimes V^*)(\int_\Sigma \tilde{f} \otimes d\hat{\E})(\id_{\his}\otimes V)$. However, the PVM integral on the right-hand side seems to require similar foundational work to generate integral operators on the spatial tensor product---the dilation route reorganizes these arguments but does not seem to bypass them.
\end{remark}

\begin{corollary}\label{cor:CP}
Assume $\MSp$ separable. The integration map $\int_\Sigma d\E: \Bb \to \MSR$ is completely positive.
\end{corollary}

\begin{proof}
The factorization \eqref{eq:stinespring} writes $(\pi_\S \otimes \pi_\R) \circ \int d\E$ as the composition
\[
    f \longmapsto \pi_\S \circ f \longmapsto \int_\Sigma (\pi_\S \circ f) \otimes d\hat{\E} \longmapsto (\id_{\his} \otimes V^*)\Big(\int_\Sigma (\pi_\S \circ f) \otimes d\hat{\E}\Big)(\id_{\his} \otimes V).
\]
The first map, $f \mapsto \pi_\S \circ f$, is a $*$-homomorphism from $\Bb$ to $\BbBH$, since $\pi_\S$ is a $*$-homomorphism. The second is a $*$-homomorphism (Thm.~\ref{thm:integrationmaps2}), and the third is CP (conjugation by an isometry). Since $*$-homomorphisms between $\mathrm{C}^*$-algebras are CP and CP maps compose, $(\pi_\S \otimes \pi_\R) \circ \int d\E$ is CP. Further, since $\pi_\S \otimes \pi_\R: \MSR \to B(\hisr)$ is a faithful $*$-homomorphism, an element of $M_n(\MSR)$ is positive if and only if its image under $\id_n \otimes (\pi_\S \otimes \pi_\R)$ is positive. Hence $\int d\E: \Bb \to \MSR$ is CP.
\end{proof}

\section{Commuting subalgebras}\label{sec:commuting}

When the two algebras are not tensor factors but commuting subalgebras of a global algebra, the integral can be realized in the latter under the joint normality hypothesis. This is the setting relevant to algebraic quantum field theory, where local algebras are subalgebras of a global algebra rather than tensor factors.

Let $\M$ be a $\mathrm{W}^*$-algebra, and let $\MS, \MR \subset \M$ be $\mathrm{W}^*$-subalgebras with $[a,b] = 0$ for all $a \in \MS$, $b \in \MR$. Throughout this section we assume that the pair is jointly normal in the sense of Section~\ref{sec:prelim}, so that algebraic multiplication extends to a normal unital $*$-homomorphism
\[
    \mu: \MSR \longrightarrow \M, \; a \otimes b \longmapsto ab,
\]
from the spatial tensor product into $\M$. See Section~\ref{sec:prelim} for the definition and Remark~\ref{rem:AQFT} for its application to AQFT.

\begin{theorem}\label{thm:embeddedprops}
Consider a pair of commuting $\mathrm{W}^*$-algebras $\MS,\MR \subseteq \M$ such that $\mu$ $($see above$)$ exists and $\MSp$ is separable. The \emph{embedded integration map}
\[
    \int_\Sigma d\E: \Bb \longrightarrow \MS \vee \MR, \; f \longmapsto \int_\Sigma f \, d\E := \mu\Big(\int_\Sigma f \otimes d\E\Big)
\]
is a pointwise-normal unital CP map and a $*$-homomorphism if $\E$ is sharp. If $\E$ is localizable and $\mu$ is injective, then $\|\int_\Sigma f\,d\E\|=\|f\|_\E$ and its kernel is $\N_\E$; in particular, it induces an isometry on $\Linf$. For simple functions $f = \sum_j a_j \chi_{X_j}$ with $a_j \in \MS$,
\begin{equation}\label{eq:embeddedchar}
    \int_\Sigma f \, d\E = \sum_j a_j \E(X_j) \in \MS \vee \MR;
\end{equation}
together with pointwise-normality, this determines $\int_\Sigma f \, d\E$ uniquely for all $f \in \Bb$.
\end{theorem}

\begin{proof}
Since $\mu$ is a normal unital $*$-homomorphism, composition with the integration map preserves all stated properties. The simple-function formula \eqref{eq:embeddedchar} follows from $\mu(\sum_j a_j \otimes \E(X_j)) = \sum_j a_j\E(X_j)$. The general case follows by ultraweak continuity: for simple $f_n \to f$ as in Lemma~\ref{lem:simpleapprox}, $\int f_n \, d\E \xrightarrow{\rm uw} \int f \, d\E$ in $\MS \vee \MR$ since $\mu$ is normal and $\int f_n \otimes d\E \xrightarrow{\rm uw} \int f \otimes d\E$ in $\MSR$.
\end{proof}

\begin{corollary}[Descent of the embedded integration map]\label{cor:embeddeddescent}
Assume $\MSp$ and $\MRp$ separable. The embedded integration map of Theorem~\ref{thm:embeddedprops} factors through the quotient $q_\E: \Bb \twoheadrightarrow \Linf$ to a normal unital CP map
\[
    \int_\Sigma d\E: \Linf \;\longrightarrow\; \MS \vee \MR, \qquad [f] \longmapsto \mu\Big(\int_\Sigma f \otimes d\E\Big).
\]
It is a $*$-homomorphism if $\E$ is sharp. It is faithful if $\mu$ is injective, isometric if $\E$ is localizable and $\mu$ is injective, and an isometric $*$-homomorphism if $\E$ is sharp and $\mu$ is injective.
\end{corollary}

\begin{proof}
The factorization through $q_\E$ is inherited from $\int_\Sigma d\E: \Bb \to \MSR$ (Thm.~\ref{thm:descent}) by composition with $\mu$. The descended map is the composition
\[
    \Linf \xrightarrow{\;\int_\Sigma d\E\;} \MSR \xrightarrow{\;\mu\;} \MS \vee \MR,
\]
inheriting normality, unitality, complete positivity, and contractivity from both factors. If $\mu$ is injective, it is faithful; composing it with the faithful map $\int_\Sigma d\E$ on $\Linf$ (Theorem~\ref{thm:descent}) therefore gives a faithful map. The structural strengthenings (isometric $*$-homomorphism under sharpness, isometry under localizability) follow by composing the corresponding properties of $\int_\Sigma d\E$ on $\Linf$ with $\mu$, using that $\mu$ injective implies $\mu$ isometric (a $*$-homomorphism between $\mathrm{C}^*$-algebras is isometric iff injective).
\end{proof}

\begin{remark}\label{rem:AQFT}
When $\mu$ is not injective, the embedded integral carries strictly less information than the tensor-product integral --- distinct elements of $\MSR$ can map to the same element of $\M$.

The commuting-subalgebra setting is relevant to algebraic quantum field theory~\cite{haag_local_1996}, where Einstein causality ensures that the local algebras of spacelike separated regions commute. Applying the ambient-algebra integration construction additionally requires joint normality of these algebras.\footnote{Hyperfiniteness or injectivity of the individual local algebras alone does not establish joint normality: the diffuse abelian example in Section~\ref{sec:prelim} already shows that injectivity is insufficient.} A sufficient condition is the split property for the pair, which supplies a normal multiplication isomorphism and therefore preserves the information contained in the tensor-product integral~\cite{doplicher_standard_1984,summers_independence_1990}. Suitable nuclearity assumptions imply the split property in broad classes of models; applicability must be checked for the regions and representation under consideration~\cite{fewster_split_2016}. In representations on a separable Hilbert space, the local algebras also have separable preduals. 
\end{remark}

\section{Parametrized integration}\label{sec:parametrized}

We now study how the integral depends on a parameter. Throughout this section, $\MS, \MR \subset \M$ form a jointly normal pair of $\mathrm{W}^*$-subalgebras with $\MSp$ separable.

\subsection{Regularity under parameters}

The following proposition establishes basic structural compatibility of parametrized integrals.

\begin{proposition}[Regularity of parametrized integrals]\label{prop:regularity}
Let $U$ be a set, $Q: U \times \Sigma \to \MS$ a function such that $Q(p,\cdot) \in \Bb$ for each $p \in U$ with $\|Q\|_\infty := \sup_{(p,x)} \|Q(p,x)\| < \infty$, and define $G: U \to \MS \vee \MR$ by $G(p) := \int_\Sigma Q(p,\cdot) \, d\E$. Then $\|G\|_\infty \leq \|Q\|_\infty$, and:
\begin{enumerate}[label=(\roman*)]
    \item If $(U,\mathcal{G})$ is a measurable space and $Q$ is jointly $(\mathcal{G} \otimes \F)$-ultraweakly measurable, then $G$ is ultraweakly $\mathcal{G}$-measurable.
  	\item Assume $\MRp$ separable. If $U$ is a topological space and the map $U \to \Linf$, $p \mapsto [Q(p,\cdot)]$, is ultraweakly continuous, then $G$ is ultraweakly continuous.
\end{enumerate}
\end{proposition}

\begin{proof}
We have $G = \mu \circ \tilde{G}$ where $\tilde{G}(p) := \int_\Sigma Q(p,\cdot) \otimes d\E \in \MSR$. Since $\mu: \MSR \to \M$ is a normal contraction, we have $\|G(p)\| \leq \|\tilde{G}(p)\|$. The bound $\|G\|_\infty \leq \|Q\|_\infty$ then follows from positivity and unitality of the integration map (Thm.~\ref{thm:integrationmaps2}).
(i) Since $G = \mu \circ \tilde{G}$ and $\mu$ is normal, we have  $\Omega(G(p)) = (\Omega \circ \mu)(\tilde{G}(p))$ with $\Omega \circ \mu \in \MSRp$. It therefore suffices to show that $p \mapsto \tilde{\Omega}(\tilde{G}(p))$ is $\mathcal{G}$-measurable for every $\tilde{\Omega} \in \MSRp$. For product states we have $(\rho \otimes \omega)(\tilde{G}(p)) = \int_\Sigma \rho(Q(p,x))\,d\E_\omega(x)$. Since $(p,x) \mapsto \rho(Q(p,x))$ is $(\mathcal{G} \otimes \F)$-measurable and bounded, the $\mathcal{G}$-measurability of $p \mapsto \int_\Sigma \rho(Q(p,x))\,d\E_\omega(x)$ follows from standard integration theory.\footnote{If $h: U \times \Sigma \to \Cn$ is bounded and $(\mathcal{G} \otimes \F)$-measurable and $\nu$ is a finite measure on $(\Sigma,\F)$, then $p \mapsto \int_\Sigma h(p,x)\,d\nu(x)$ is $\mathcal{G}$-measurable: for $h=\chi_{A\times B}$ the integral is $\chi_A\nu(B)$. The sets $Z\in\mathcal G\otimes\F$ for which $p\mapsto\nu(Z_p)$ is measurable form a Dynkin system containing the measurable rectangles, hence contain $\mathcal G\otimes\F$ by the $\pi$--$\lambda$ theorem. Simple-function approximation and dominated convergence then give the assertion for bounded measurable $h$.} The extension to all $\tilde{\Omega} \in \MSRp$ follows by density of product states (Lemma~\ref{lem:proddense}).
(ii) By Theorem~\ref{thm:descent}, the descended integration map $\int_\Sigma d\E: \Linf \to \MSR$ is ultraweakly continuous. Since $\tilde{G}(p) = \int_\Sigma [Q(p,\cdot)] \otimes d\E$ is the composition of $p \mapsto [Q(p,\cdot)]$ with the descended integration map, it is ultraweakly continuous as a function $U \to \MSR$. Composition with the normal $*$-homomorphism $\mu: \MSR \to \M$ then gives $G = \mu \circ \tilde{G}$ ultraweakly continuous as a function $U \to \M$.
\end{proof}

\begin{remark}\label{rem:seqreg}
For first-countable $U$ $($e.g.\ metrizable$)$, the hypothesis of item~(ii) can be checked sequentially: assuming $\MRp$ separable, if $p_n \to p$ in $U$ implies $Q(p_n,x) \xrightarrow{\rm uw} Q(p,x)$ for each $x$, and if $\|Q\|_\infty < \infty$, then $p \mapsto [Q(p,\cdot)] \in \Linf$ is ultraweakly continuous. Indeed, for $h \in L^1_{\mu_\E}(\Sigma)$ and $\rho \in \MSp$, $|h(x)\rho(Q(p_n,x))| \leq |h(x)|\,\|Q\|_\infty\,\|\rho\|$ is $\mu_\E$-integrable and the integrand converges pointwise, so the dominated convergence theorem gives $\langle [Q(p_n,\cdot)], \rho \otimes h\rangle \to \langle [Q(p,\cdot)], \rho \otimes h\rangle$ (Thm.~\ref{thm:Linfstructure}); sequential ultraweak continuity follows by density, and is equivalent to ultraweak continuity for first-countable $U$. For general topological $U$, the dominated convergence theorem applies only to sequences, and the $\Linf$-level continuity of item~(ii) is the appropriate hypothesis.
\end{remark}

\subsection{Leibniz rule}

We now turn to discuss how differentiability interacts with integration in the $\mathrm{W}^*$-algebraic setting.

\begin{definition}\label{def:uwdiff}
    Let $I \subset \Reals$ be an open interval and $f: I \to \MS$. The function $f$ is called \emph{ultraweakly differentiable at $t \in I$} if there exists an element $\dot{f}(t) \in \MS$, called the \emph{ultraweak derivative of $f$ at $t$}, such that for every $\rho \in \S(\MS)$ (equivalently, for every $\rho \in \MSp$) the scalar function
    \[
        f_\rho: I \ni s \mapsto \rho(f(s)) \in \Cn
    \]
    is differentiable at $s = t$ with derivative $\rho(\dot{f}(t))$. The function $f$ is ultraweakly differentiable on $I$ if it is ultraweakly differentiable at every $t \in I$.
\end{definition}

\begin{lemma}\label{lem:uwdiff}
Let $I \subset \Reals$ be an open interval and $f: I \to \MS$. The following are equivalent:
\begin{enumerate}[label=(\roman*)]
    \item $f$ is ultraweakly differentiable at $t \in I$ in the sense of Definition~\ref{def:uwdiff}, with ultraweak derivative $\dot f(t) \in \MS$;
    \item $f$ is weak-$*$ differentiable at $t$ when $\MS$ is viewed as the Banach dual $(\MSp)^*$, with weak-$*$ derivative $\dot f(t) \in \MS$; that is, $h^{-1}(f(t+h) - f(t)) \to \dot f(t)$ in the weak-$*$ topology $\sigma(\MS, \MSp)$ as $h \to 0$.
\end{enumerate}
The weak-$*$ topology on $\MS$ coincides with the ultraweak topology, so the convergence in (ii) is ultraweak. When the ultraweak derivative exists, it is unique.
\end{lemma}

\begin{proof}
$(i) \Rightarrow (ii)$ Since the weak-$*$ and ultraweak topologies on $\MS$ coincide by definition---the ultraweak topology on a $\mathrm{W}^*$-algebra is $\sigma(\MS, \MSp)$---this is immediate.

$(ii) \Rightarrow (i)$ By definition of the weak-$*$ topology $\sigma(\MS, \MSp)$, a net $a_h$ converges to $a$ in $\sigma(\MS, \MSp)$ if and only if $\rho(a_h) \to \rho(a)$ for every $\rho \in \MSp$. Applied to $a_h = h^{-1}(f(t+h) - f(t))$ and $a = \dot f(t)$, this gives the equivalence of (ii) with the condition that $\rho(h^{-1}(f(t+h) - f(t))) \to \rho(\dot f(t))$ for every $\rho \in \MSp$, i.e., $h^{-1}(\rho(f(t+h)) - \rho(f(t))) \to \rho(\dot f(t))$. The latter is precisely scalar differentiability of $s \mapsto \rho(f(s))$ at $s = t$ with derivative $\rho(\dot f(t))$ for every $\rho \in \MSp$ --- equivalently, for every $\rho \in \S(\MS)$ by Jordan decomposition and linearity. This is condition (i).

For uniqueness: if $\dot f_1(t)$ and $\dot f_2(t)$ both satisfy (i), then $\rho(\dot f_1(t) - \dot f_2(t)) = 0$ for every $\rho \in \MSp$, and since $\MSp$ separates points of $\MS$, $\dot f_1(t) = \dot f_2(t)$.
\end{proof}

\begin{proposition}[Quantum Leibniz rule]\label{prop:Leibniz}
Let $I \subset \Reals$ be an open interval and $Q: I \times \Sigma \to \MS$ such that:
\begin{enumerate}[label=(\alph*)]
    \item $Q(t,\cdot) \in \Bb$ for each $t \in I$, with $\sup_{t \in K} \|Q(t,\cdot)\|_\infty < \infty$ for every compact $K \subset I$;
    \item for each $x \in \Sigma$, $t \mapsto Q(t,x)$ is ultraweakly differentiable on $I$ with ultraweak derivative $\dot{Q}(t,x) \in \MS$;
    \item $\dot{Q}(t,\cdot) \in \Bb$ for each $t \in I$, with $\sup_{t \in K} \|\dot{Q}(t,\cdot)\|_\infty < \infty$ for every compact $K \subset I$.
\end{enumerate}
Then $t \mapsto \int_\Sigma Q(t,x)\,d\E(x)$ is ultraweakly differentiable on $I$ in $\M$, and integration commutes with the derivative:
\[
    \frac{d}{dt}\int_\Sigma Q(t,x)\,d\E(x) = \int_\Sigma \dot Q(t,x)\,d\E(x).
\]
\end{proposition}

\begin{proof}
Define $\tilde G(t) := \int_\Sigma Q(t,\cdot) \otimes d\E \in \MSR$ and $G(t) := \mu(\tilde G(t)) = \int_\Sigma Q(t,\cdot)\,d\E \in \M$. Fix $t_0 \in I$ and a compact neighborhood $K \subset I$ of $t_0$, and define the difference quotients
\[
    D_h(x) := h^{-1}\big(Q(t_0+h, x) - Q(t_0, x)\big), \qquad h \in \Reals\setminus\{0\} \text{ with } t_0 + h \in I.
\]
Each $D_h \in \Bb$ as a linear combination of elements of $\Bb$ (hypothesis~(a)), and by linearity of the integration map on $\Bb$,
\begin{equation}\label{eq:diffquotient}
    h^{-1}(\tilde G(t_0+h) - \tilde G(t_0)) = \int_\Sigma D_h \otimes d\E \in \MSR.
\end{equation}

\emph{Pointwise ultraweak convergence of $D_h$.} For each fixed $x \in \Sigma$, hypothesis~(b) and Definition~\ref{def:uwdiff} applied to $t \mapsto Q(t, x): I \to \MS$ give
\[
    \rho(D_h(x)) = h^{-1}\big(\rho(Q(t_0+h, x)) - \rho(Q(t_0, x))\big) \xrightarrow{h \to 0} \rho(\dot Q(t_0, x))
\]
for every $\rho \in \MSp$. By Lemma~\ref{lem:uwdiff}, this is precisely $D_h(x) \xrightarrow{\rm uw} \dot Q(t_0, x)$ in $\MS$ as $h \to 0$.

\emph{Uniform bound on $\|D_h(\cdot)\|_\infty$.} Fix $\rho \in \S(\MS)$ and $x \in \Sigma$. By hypothesis~(b) and Definition~\ref{def:uwdiff} applied to $s \mapsto Q(s, x)$, the scalar function $s \mapsto \rho(Q(s, x))$ is differentiable on $I$ with derivative $\rho(\dot Q(s, x))$ at each $s \in I$. The real and imaginary parts $s \mapsto \mathrm{Re}(\rho(Q(s,x)))$ and $s \mapsto \mathrm{Im}(\rho(Q(s,x)))$ are therefore real-differentiable on $I$ with derivatives $\mathrm{Re}(\rho(\dot Q(s,x)))$ and $\mathrm{Im}(\rho(\dot Q(s,x)))$. By the mean value theorem applied on the closed interval between $t_0$ and $t_0 + h$ (contained in $K$ for $|h|$ small enough),
\[
    \big|\mathrm{Re}(\rho(D_h(x)))\big| \leq \sup_{s \in K} \big|\mathrm{Re}(\rho(\dot Q(s,x)))\big|,
\]
and similarly for the imaginary part. The bound $|\rho(\dot Q(s,x))| \leq \|\rho\|\,\|\dot Q(s,x)\|_{\MS} \leq \|\dot Q(s,\cdot)\|_\infty$ holds for $\rho \in \S(\MS)$ (so $\|\rho\| = 1$), giving
\[
    \big|\mathrm{Re}(\rho(D_h(x)))\big| \leq C_K, \qquad \big|\mathrm{Im}(\rho(D_h(x)))\big| \leq C_K, \qquad C_K := \sup_{s \in K} \|\dot Q(s,\cdot)\|_\infty,
\]
which is finite by hypothesis~(c). Decomposing $D_h(x) = \mathrm{Re}(D_h(x)) + i\,\mathrm{Im}(D_h(x))$ into self-adjoint parts in $\MS$,
\[
    \|\mathrm{Re}(D_h(x))\|_{\MS} = \sup_{\rho \in \S(\MS)} \big|\rho(\mathrm{Re}(D_h(x)))\big| = \sup_{\rho \in \S(\MS)} \big|\mathrm{Re}(\rho(D_h(x)))\big| \leq C_K
\]
by the self-adjoint norm formula in Lemma~\ref{lem:extfromstates} (using $\rho(\mathrm{Re}(a)) = \mathrm{Re}(\rho(a))$ for self-adjoint functionals on self-adjoint elements), and likewise $\|\mathrm{Im}(D_h(x))\|_{\MS} \leq C_K$. The triangle inequality gives $\|D_h(x)\|_{\MS} \leq 2 C_K$ for every $x \in \Sigma$ and all sufficiently small nonzero $h$, hence $\|D_h(\cdot)\|_\infty \leq 2 C_K$.

\emph{Ultraweak differentiability of $G$.} Combining the uniform bound $\|D_h(\cdot)\|_\infty \leq 2 C_K$ with the pointwise convergence $D_h(x) \xrightarrow{\rm uw} \dot Q(t_0, x)$ for each $x$, pointwise-normality of the integration map on $\Bb$ (Thm.~\ref{thm:integrationmaps2}), applied along each sequence of nonzero $h_n\to0$, gives
\[
    \int_\Sigma D_h \otimes d\E \;\xrightarrow{\rm uw}\; \int_\Sigma \dot Q(t_0, \cdot) \otimes d\E \quad \text{in } \MSR \text{ as } h \to 0,
\]
where the limit is well-defined in $\MSR$ by Theorem~\ref{thm:intop} applied to $\dot Q(t_0, \cdot) \in \Bb$ (hypothesis~(c)). Testing against each normal functional and using the sequential criterion for scalar limits as $h\to0$ justifies the stated limit. Normality of $\mu: \MSR \to \M$ propagates this convergence to $\M$:
\[
    \mu\Big(\int_\Sigma D_h \otimes d\E\Big) \;\xrightarrow{\rm uw}\; \mu\Big(\int_\Sigma \dot Q(t_0,\cdot) \otimes d\E\Big) = \int_\Sigma \dot Q(t_0,\cdot)\,d\E \quad \text{in } \M.
\]
Together with~\eqref{eq:diffquotient} and the definition of $G$,
\[
    h^{-1}\big(G(t_0+h) - G(t_0)\big) \;\xrightarrow{\rm uw}\; \int_\Sigma \dot Q(t_0, \cdot)\,d\E \quad \text{in } \M.
\]
By Lemma~\ref{lem:uwdiff} applied to $G: I \to \M$, $G$ is ultraweakly differentiable at $t_0$ with ultraweak derivative
\[
    \dot G(t_0) = \int_\Sigma \dot Q(t_0,\cdot)\,d\E \in \M.
\]
Since $t_0 \in I$ was arbitrary, $G$ is ultraweakly differentiable on $I$ with the displayed derivative, as claimed.
\end{proof}

\subsection{Operator-valued Fubini theorem}\label{sec:Fubini}

We now establish an operator-valued Fubini theorem in which both measures are POVMs valued in independent $\mathrm{W}^*$-algebras. We first construct the product POVM, then state and prove the Fubini identity.

\begin{lemma}[Product POVM]\label{lem:productPOVM}
Let $\M_1$ and $\M_2$ be $\mathrm{W}^*$-algebras, $\E_1: \F_1 \to \Eff(\M_1)$ a POVM on $(\Sigma_1,\F_1)$ and $\E_2: \F_2 \to \Eff(\M_2)$ a POVM on $(\Sigma_2,\F_2)$. There exists a unique POVM $\E_1 \otimes \E_2: \F_1 \otimes \F_2 \to \Eff(\M_1 \bar\otimes \M_2)$ satisfying $(\E_1 \otimes \E_2)(X_1 \times X_2) = \E_1(X_1) \otimes \E_2(X_2)$ for all $X_1 \in \F_1$, $X_2 \in \F_2$. It is characterized by the \emph{product-measure property}: for all $\omega_1 \in \S(\M_1)$, $\omega_2 \in \S(\M_2)$ and $Z \in \F_1 \otimes \F_2$,
\begin{equation}\label{eq:prodmeasure}
    (\E_1 \otimes \E_2)_{\omega_1 \otimes \omega_2}(Z) = ((\E_1)_{\omega_1} \otimes (\E_2)_{\omega_2})(Z),
\end{equation}
where $(\E_1)_{\omega_1} \otimes (\E_2)_{\omega_2}$ is the product of the probability measures on $\Sigma_1 \times \Sigma_2$. If $\E_1$ and $\E_2$ are both sharp $($PVMs$)$, then $\E_1 \otimes \E_2$ is sharp. If $\E_1$ and $\E_2$ are both localizable and at least one of $(\M_1)_*$ and $(\M_2)_*$ is separable, then $\E_1\otimes\E_2$ is localizable.
\end{lemma}

\begin{proof}
Choose faithful normal representations $\pi_1: \M_1 \hookrightarrow B(\hi_1)$ and $\pi_2: \M_2 \hookrightarrow B(\hi_2)$ and set $\tilde{\E}_j := \pi_j \circ \E_j$. Let $(\hik_1, V_1, \hat{\E}_1)$ and $(\hik_2, V_2, \hat{\E}_2)$ be Naimark dilations (Sec.~\ref{sec:CP}). The PVMs $\hat{\E}_1(X_1) \otimes \id_{\hik_2}$ and $\id_{\hik_1} \otimes \hat{\E}_2(X_2)$ on $\hik_1 \otimes \hik_2$ commute (acting on different tensor factors), so by the spectral theorem for commuting projection-valued measures they determine a unique PVM $\hat{\E}_1 \otimes \hat{\E}_2: \F_1 \otimes \F_2 \to B(\hik_1 \otimes \hik_2)$ satisfying $(\hat{\E}_1 \otimes \hat{\E}_2)(X_1 \times X_2) = \hat{\E}_1(X_1) \otimes \hat{\E}_2(X_2)$. Define
\[
    (\tilde{\E}_1 \boxtimes \tilde{\E}_2)(Z) := (V_1 \otimes V_2)^*\,(\hat{\E}_1 \otimes \hat{\E}_2)(Z)\,(V_1 \otimes V_2) \; \text{for } Z \in \F_1 \otimes \F_2.
\]
This is a POVM on $B(\hi_1 \otimes \hi_2)$: it is positive (compression of positive operators), $\sigma$-additive (in the ultraweak topology, inherited from the PVM), and normalized since $V_1, V_2$ are isometries.

\emph{Values in $\M_1 \bar{\otimes} \M_2$.} Let $\mathcal{C} := \{Z \in \F_1 \otimes \F_2 : (\tilde{\E}_1 \boxtimes \tilde{\E}_2)(Z) \in \pi_1(\M_1) \bar{\otimes} \pi_2(\M_2)\}$. On rectangles, $(\tilde{\E}_1 \boxtimes \tilde{\E}_2)(X_1 \times X_2) = \tilde{\E}_1(X_1) \otimes \tilde{\E}_2(X_2) \in \pi_1(\M_1) \otimes_{\rm alg} \pi_2(\M_2) \subset \pi_1(\M_1) \bar{\otimes} \pi_2(\M_2)$, so $\mathcal{C}$ contains all rectangles. Moreover, $\mathcal{C}$ is closed under complements (since $(\tilde{\E}_1 \boxtimes \tilde{\E}_2)(Z^c) = \id - (\tilde{\E}_1 \boxtimes \tilde{\E}_2)(Z)$ and $\id \in \pi_1(\M_1) \bar{\otimes} \pi_2(\M_2)$) and under countable disjoint unions (ultraweak $\sigma$-additivity and ultraweak closedness of $\pi_1(\M_1) \bar{\otimes} \pi_2(\M_2)$). By the Dynkin $\pi$-$\lambda$ theorem, $\mathcal{C} = \F_1 \otimes \F_2$.

We may therefore define the abstract product POVM $(\E_1 \otimes \E_2)(Z) := (\pi_1 \otimes \pi_2)^{-1}((\tilde{\E}_1 \boxtimes \tilde{\E}_2)(Z)) \in \M_1 \bar{\otimes} \M_2$, which is a POVM (inheriting positivity, $\sigma$-additivity and normalization from the Hilbert-space level). On rectangles, $(\E_1 \otimes \E_2)(X_1 \times X_2) = \E_1(X_1) \otimes \E_2(X_2)$ holds by construction. For the product-measure property: for $\omega_j \in \S(\M_j)$ with positive trace-class representatives $\tilde{\omega}_j$,
\begin{align*}
    (\E_1 \otimes \E_2)_{\omega_1 \otimes \omega_2}(Z) &= \tr[(\tilde{\omega}_1 \otimes \tilde{\omega}_2)\,(\tilde{\E}_1 \boxtimes \tilde{\E}_2)(Z)]
    = \tr[(V_1\tilde{\omega}_1V_1^* \otimes V_2\tilde{\omega}_2V_2^*)\,(\hat{\E}_1 \otimes \hat{\E}_2)(Z)] \\
    &= (\hat{\E}_1 \otimes \hat{\E}_2)_{V_1\tilde{\omega}_1V_1^* \otimes V_2\tilde{\omega}_2V_2^*}(Z)
    = ((\hat{\E}_1)_{V_1\tilde{\omega}_1V_1^*} \otimes (\hat{\E}_2)_{V_2\tilde{\omega}_2V_2^*})(Z)
    = ((\E_1)_{\omega_1} \otimes (\E_2)_{\omega_2})(Z),
\end{align*}
where the fourth equality uses that the scalar measure induced by a product state on the joint PVM of commuting PVMs factors as the product of the individual scalar measures, and the last uses $(\hat{\E}_j)_{V_j\tilde{\omega}_jV_j^*}(X) = \tr[\tilde{\omega}_j\,V_j^*\hat{\E}_j(X)V_j] = (\E_j)_{\omega_j}(X)$. Uniqueness follows from Lemma~\ref{lem:extfromstates}: the product-measure property determines $(\E_1 \otimes \E_2)(Z)$ on all product states $\omega_1 \otimes \omega_2$, hence uniquely as an element of $\M_1 \bar{\otimes} \M_2$; in particular, the construction is independent of the choice of representations.

Sharpness: if $\E_1$ and $\E_2$ are PVMs, the Naimark dilations are trivial ($V_j = \id$), so $\E_1 \otimes \E_2 = \hat{\E}_1 \otimes \hat{\E}_2$ is a PVM.

Localizability: by symmetry, suppose $(\M_2)_*$ is separable. For $Z\in\F_1\otimes\F_2$, set $Z_{x_1}:=\{x_2:(x_1,x_2)\in Z\}$ and $f_Z(x_1):=\E_2(Z_{x_1})$. Scalar Fubini shows that $f_Z$ is ultraweakly measurable and, using~\eqref{eq:prodmeasure}, that
\[
    \int_{\Sigma_1}f_Z\otimes d\E_1
       =\tau\big((\E_1\otimes\E_2)(Z)\big),
\]
where $\tau:\M_1\bar\otimes\M_2\to\M_2\bar\otimes\M_1$ is the tensor flip. By~\eqref{eq:operator-norm-recovery},
\[
    \|(\E_1\otimes\E_2)(Z)\|=\|f_Z\|_{\E_1}.
\]
Localizability of $\E_2$ gives $\|f_Z(x_1)\|\in\{0,1\}$, so this essential supremum is either zero or one. Every nonzero product effect therefore has norm one.
\end{proof}

\begin{theorem}[Quantum Fubini Theorem]\label{thm:Fubini}
Let $\MS$, $\M_1$ and $\M_2$ be $\mathrm{W}^*$-subalgebras of a $\mathrm{W}^*$-algebra $\M$ with separable preduals, and assume that the triple $(\MS,\M_1,\M_2)$ is jointly normal.\footnote{Separability of $(\M_j)_*$ is used for the iterated integrals, where the inner integral produces an element of $\MS \vee \M_j$ and the outer integration requires separability of the predual of the domain algebra. This is guaranteed: since $\mu: \MS \bar\otimes \M_j \to \M$ is a normal surjection onto $\MS \vee \M_j$, its predual map embeds $(\MS \vee \M_j)_*$ isometrically into $(\MS \bar\otimes \M_j)_*$, which is separable (it is a quotient of $\T(\his \otimes \hi_j)$ for separable $\his, \hi_j$).} Let $\E_j: \F_j \to \Eff(\M_j)$ be a POVM on $(\Sigma_j,\F_j)$ for $j = 1,2$, and let $Q: \Sigma_1 \times \Sigma_2 \to \MS$ be uniformly bounded and ultraweakly jointly $(\F_1 \otimes \F_2)$-measurable. Define:
\[
    G_1: \Sigma_1 \to \MS \vee \M_2, \; G_1(x_1) := \int_{\Sigma_2} Q(x_1,x_2) \, d\E_2(x_2),
\]
\[
    G_2: \Sigma_2 \to \MS \vee \M_1, \; G_2(x_2) := \int_{\Sigma_1} Q(x_1,x_2) \, d\E_1(x_1).
\]
\begin{enumerate}[label=(\roman*)]
    \item \emph{Well-posedness.} $G_1$ and $G_2$ are uniformly bounded and ultraweakly measurable, with $\|G_j\|_\infty \leq \|Q\|_\infty$.
    \item \emph{Fubini identity.} The iterated integrals agree and equal the product integral:
    \begin{equation}\label{eq:Fubini}
        \int_{\Sigma_1} G_1(x_1) \, d\E_1(x_1) = \int_{\Sigma_1 \times \Sigma_2} Q(x_1,x_2) \, d(\E_1\E_2)(x_1,x_2) = \int_{\Sigma_2} G_2(x_2) \, d\E_2(x_2)
    \end{equation}
    in $\MS \vee \M_1 \vee \M_2 \subset \M$, where $\E_1\E_2 := \mu_{12} \circ (\E_1 \otimes \E_2): \F_1 \otimes \F_2 \to \Eff(\M_1 \vee \M_2)$ is the embedding of the product POVM $($Lemma~\ref{lem:productPOVM}$)$ via the multiplication map $\mu_{12}: \M_1 \bar\otimes \M_2 \to \M$.
\end{enumerate}
\end{theorem}

\begin{proof}
Write $\mathcal A=\MS\bar\otimes\M_1\bar\otimes\M_2$ and let $\mu:\mathcal A\to\MS\vee\M_1\vee\M_2$ be the normal multiplication map supplied by joint normality. Denote its restrictions to the two-factor tensor products by $\mu_{\S 1}$, $\mu_{\S 2}$ and $\mu_{12}$, with codomains the corresponding generated subalgebras.

(i) Define the tensor-valued lifts
\[
 \widetilde G_1(x_1)=\int_{\Sigma_2}Q(x_1,x_2)\otimes d\E_2(x_2)
 \in\MS\bar\otimes\M_2,
\]
\[
 \widetilde G_2(x_2)=\int_{\Sigma_1}Q(x_1,x_2)\otimes d\E_1(x_1)
 \in\MS\bar\otimes\M_1.
\]
The tensor-level argument in Proposition~\ref{prop:regularity} shows that these functions are ultraweakly measurable and bounded by $\|Q\|_\infty$: scalar measurability follows by testing against product normal states, and extends to every normal functional by norm density and the uniform bound. Since $G_1=\mu_{\S 2}\circ\widetilde G_1$ and $G_2=\mu_{\S 1}\circ\widetilde G_2$, normality and contractivity of the partial multiplication maps give (i).

(ii) Let $\tau:(\MS\bar\otimes\M_2)\bar\otimes\M_1\to\mathcal A$ be the canonical isomorphism $a\otimes b_2\otimes b_1\mapsto a\otimes b_1\otimes b_2$, and suppress associativity isomorphisms. Set
\[
 j_1=\tau\left(\int_{\Sigma_1}\widetilde G_1\otimes d\E_1\right),\qquad
 j_2=\int_{\Sigma_2}\widetilde G_2\otimes d\E_2,\qquad
 j=\int_{\Sigma_1\times\Sigma_2}Q\otimes d(\E_1\otimes\E_2).
\]
All three elements belong to $\mathcal A$. The iterated integrals exist because the two-factor tensor products have separable preduals. For normal states $\rho\in\S(\MS)$ and $\omega_j\in\S(\M_j)$, the characterizing formula gives
\begin{align*}
 (\rho\otimes\omega_1\otimes\omega_2)(j_1)
 &=\int_{\Sigma_1}\left(\int_{\Sigma_2}\rho(Q(x_1,x_2))\,d(\E_2)_{\omega_2}(x_2)\right)d(\E_1)_{\omega_1}(x_1),\\
 (\rho\otimes\omega_1\otimes\omega_2)(j_2)
 &=\int_{\Sigma_2}\left(\int_{\Sigma_1}\rho(Q(x_1,x_2))\,d(\E_1)_{\omega_1}(x_1)\right)d(\E_2)_{\omega_2}(x_2).
\end{align*}
By scalar Fubini these are equal to
\[
 \int_{\Sigma_1\times\Sigma_2}\rho(Q(x_1,x_2))\,
 d\big((\E_1)_{\omega_1}\otimes(\E_2)_{\omega_2}\big)(x_1,x_2)
 =(\rho\otimes\omega_1\otimes\omega_2)(j),
\]
where the last equality uses~\eqref{eq:prodmeasure}. Product normal states separate points of $\mathcal A$ (apply Lemma~\ref{lem:proddense} twice), so $j_1=j=j_2$.

It remains to identify their images under $\mu$ with the embedded integrals in~\eqref{eq:Fubini}. Joint normality also supplies normal multiplication maps
\begin{align*}
 \nu_1&:(\MS\vee\M_2)\bar\otimes\M_1\longrightarrow\MS\vee\M_1\vee\M_2,\\
 \nu_2&:(\MS\vee\M_1)\bar\otimes\M_2\longrightarrow\MS\vee\M_1\vee\M_2,\\
 \nu_0&:\MS\bar\otimes(\M_1\vee\M_2)\longrightarrow\MS\vee\M_1\vee\M_2,
\end{align*}
with
\[
 \nu_1\circ(\mu_{\S 2}\bar\otimes\id_{\M_1})=\mu\circ\tau,\quad
 \nu_2\circ(\mu_{\S 1}\bar\otimes\id_{\M_2})=\mu,\quad
 \nu_0\circ(\id_{\MS}\bar\otimes\mu_{12})=\mu.
\]
Indeed, the kernel of each partial multiplication map is an ideal generated by a central projection. The corresponding central projection tensored with the remaining identity is killed by $\mu$ (with the appropriate flip), so $\mu$ factors normally through the partial quotient. On elementary tensors these factorizations act by multiplication, proving the asserted properties of the $\nu_j$.

The generated subalgebras have separable preduals, being normal quotients of two-factor tensor products. Thus functoriality from Theorem~\ref{thm:intop} applies to the partial multiplication maps, which are normal channels. It yields
\begin{align*}
 \int_{\Sigma_1}G_1\,d\E_1
 &=\nu_1\left(\int_{\Sigma_1}(\mu_{\S 2}\circ\widetilde G_1)\otimes d\E_1\right)=\mu(j_1),\\
 \int_{\Sigma_2}G_2\,d\E_2
 &=\nu_2\left(\int_{\Sigma_2}(\mu_{\S 1}\circ\widetilde G_2)\otimes d\E_2\right)=\mu(j_2),\\
 \int_{\Sigma_1\times\Sigma_2}Q\,d(\E_1\E_2)
 &=\nu_0\left(\int_{\Sigma_1\times\Sigma_2}Q\otimes d\big(\mu_{12}\circ(\E_1\otimes\E_2)\big)\right)=\mu(j).
\end{align*}
Applying $\mu$ to $j_1=j=j_2$ proves~\eqref{eq:Fubini}.
\end{proof}

\begin{remark}\label{rem:Fubini}
The tensor-valued lifts in the proof give the identity
\begin{equation}\label{eq:Fubiniproduct}
 \tau\left(\int_{\Sigma_1}\widetilde G_1\otimes d\E_1\right)
 =\int_{\Sigma_1\times\Sigma_2}Q\otimes d(\E_1\otimes\E_2)
 =\int_{\Sigma_2}\widetilde G_2\otimes d\E_2
\end{equation}
in $\MS\bar\otimes\M_1\bar\otimes\M_2$, with associativity understood and $\tau$ the flip specified above. This tensor-level identity requires no ambient multiplication map. Under the joint normality hypothesis, applying $\mu$ gives the literal equality of embedded integrals in $\MS\vee\M_1\vee\M_2$. Injectivity of the multiplication maps is not required.
\end{remark}

\section{Summary of integration results}\label{sec:intsummary}

We collect the results of Sections~\ref{sec:generalities}--\ref{sec:parametrized} into a single statement.

\begin{theorem}[Operator-valued integration]\label{thm:master}
Let $\MS$ be a $\mathrm{W}^*$-algebra with separable predual, $\MR$ a $\mathrm{W}^*$-algebra, $(\Sigma, \F)$ a measurable space, and $\E: \F \to \Eff(\MR)$ a normalized POVM.
\begin{enumerate}[label=(\roman*)]
    
  	\item \emph{Existence and uniqueness.} (Thm.~\ref{thm:intop}) For each uniformly bounded ultraweakly $\F$-measurable function $f: \Sigma \to \MS$ there exists a unique element $\int_\Sigma f \otimes d\E \in \MSR$ satisfying
    \begin{equation}\label{eq:masterchar}
        (\rho \otimes \omega)\Big(\int_\Sigma f \otimes d\E\Big) = \int_\Sigma f_\rho \, d\E_\omega
    \end{equation}
    for all $\rho \in \S(\MS)$ and $\omega \in \S(\MR)$, where $f_\rho := \rho \circ f$ and $\E_\omega := \omega \circ \E$.
    
    \item \emph{Functoriality.} (Thm.~\ref{thm:intop}) For any measurable $\alpha: (\Sigma,\F) \to (\Sigma',\F')$, uniformly bounded $\F'$-measurable $f: \Sigma' \to \MS$, and normal channels $\Psi: \MS \to \N_\S$, $\Phi: \MR \to \N_\R$ with $(\N_\S)_*$ separable, we have
    \[
        \int_\Sigma (\Psi \circ f \circ \alpha) \otimes d(\Phi \circ \E) = (\Psi \otimes \Phi)\Big(\int_{\Sigma'} f \otimes d(\alpha_*\E)\Big).
    \]
    
   \item \emph{Domains.} When $\MSp$ is separable, the space $\Bb$ of uniformly bounded ultraweakly measurable functions $f: \Sigma \to \MS$ is a $\mathrm{C}^*$-algebra under the pointwise operations and the uniform norm (Prop.~\ref{prop:BbisCstar}). It is the largest common domain for all integration maps: any ultraweakly measurable $f \notin \Bb$ fails to be integrable with respect to some POVM (Prop.~\ref{prop:integrability}).
   
   Once a POVM $\E$ is fixed, the natural domain is the quotient $\Linf := \Bb/\N_\E$ by the $*$-ideal $\N_\E = \{f \in \Bb : f = 0 \; \E\text{-a.e.}\}$. If additionally $\MRp$ is separable, $\Linf$ is a $\mathrm{W}^*$-algebra with $\Linf \cong \MS \, \bar\otimes \, L^\infty_\E(\Sigma)$ and predual $\MSp \,\hat\otimes_\pi\, L^1_\E(\Sigma)$, with $L^\infty_\E(\Sigma)=L^\infty_{\mu_\E}(\Sigma)$ for any dominating probability measure $\mu_\E$ on $(\Sigma,\F)$ (Thm. \ref{thm:Linfstructure}).
    
    \item \emph{Integration map.} (Thm.~\ref{thm:integrationmaps2}) Assuming separability of $\MSp$, the map
    \[
        \int_\Sigma d\E: \Bb \longrightarrow \MSR, \; f \longmapsto \int_\Sigma f \otimes d\E,
    \]
    is a pointwise-normal unital CP map. If $\E$ is localizable, it satisfies $\|\int_\Sigma f\otimes d\E\|=\|f\|_\E$ and has kernel $\N_\E$. It is a $*$-homomorphism if $\E$ is sharp $($PVM$)$. Together with pointwise-normality, the elementary simple-function formula
    \[
        \int_\Sigma \sum_j a_j \chi_{X_j} \otimes d\E = \sum_j a_j \otimes \E(X_j)
    \]
    
      determines the integral uniquely.
    \item \emph{Descent to $\Linf$} (Thm.~\ref{thm:descent}). The ideal $\N_\E \subset \Bb$ lies in the kernel of $\int_\Sigma d\E$, so the integration map factors uniquely as
    \[
        \Bb \twoheadrightarrow \Linf \xrightarrow{\;\int_\Sigma d\E\;} \MSR,
    \]
    and the descended map is a faithful unital CP map. It is an isometric (hence injective) $*$-homomorphism when $\E$ is sharp $($and isometric when $\E$ is localizable$)$. Under separability of $\MRp$, $\Linf \cong \MS \bar\otimes L^\infty_\E(\Sigma)$ is a $\mathrm{W}^*$-algebra and the descended map is a faithful normal unital CP map.
    
 	\item \emph{Monoidal characterization}
 		(Thm.~\ref{thm:int-categorical}). Assume that $\MRp$ is separable. Under the bijection $\E \leftrightarrow \Phi_\E$ between normalized POVMs and faithful normal unital CP maps on commutative domains (Prop.~\ref{prop:povm-bijection}), translating sharpness and localizability to multiplicativity and isometry, respectively, the descended integration map coincides with the $\mathrm{W}^*$-monoidal lift
 	\[
    	\int_\Sigma d\E \;\cong\; \id_{\MS} \,\bar\otimes\, \Phi_\E\;:\; \MS \,\bar\otimes\, L^\infty_\E(\Sigma) \;\longrightarrow\; \MS \,\bar\otimes\, \MR.
	\]

    \item \emph{Stinespring factorization} (Thm.~\ref{thm:stinespring}). In any faithful normal representation, the integration map factors through the Naimark dilation $(\hik,V,\hat{\E})$ of the represented POVM:
    \[
        (\pi_\S \otimes \pi_\R)\Big(\int_\Sigma f \otimes d\E\Big) = (\id_{\his} \otimes V^*)\Big(\int_\Sigma (\pi_\S \circ f) \otimes d\hat{\E}\Big)(\id_{\his} \otimes V),
    \]
    where the right-hand side is integration against a PVM $($a pointwise-normal unital $*$-homomorphism$)$.
        
    \item \emph{Commuting subalgebras} (Thm.~\ref{thm:embeddedprops}). If $\MS, \MR \subset \M$ are commuting $\mathrm{W}^*$-subalgebras of a $\mathrm{W}^*$-algebra $\M$ with $\MSp$ separable and admitting a normal multiplication map $\mu: \MSR \to \M$, then the \emph{embedded integration map}
    \[
        \int_\Sigma d\E: \Bb \longrightarrow \MS \vee \MR, \; f \longmapsto \int_\Sigma f \, d\E := \mu\left(\int_\Sigma f \otimes d\E \right),
    \]
    is a pointwise-normal unital CP map and a $*$-homomorphism for $\E$ sharp; if $\E$ is localizable and $\mu$ is injective, it satisfies $\|\int_\Sigma f\,d\E\|=\|f\|_\E$ and has kernel $\N_\E$; it retains the elementary formula for simple functions and, when $\MRp$ is also separable, descends to a normal unital CP map (Cor.~\ref{cor:embeddeddescent}):
\[
    \int_\Sigma d\E: \Linf \;\longrightarrow\; \MS \vee \MR, \qquad [f] \longmapsto \mu\Big(\int_\Sigma f \otimes d\E\Big).
\]
The descended map is faithful if $\mu$ is injective, isometric if $\E$ is localizable and $\mu$ is injective, and an isometric $*$-homomorphism if $\E$ is sharp and $\mu$ is injective.

\item \emph{Quantum Leibniz rule} (Prop.~\ref{prop:Leibniz}). Assume the jointly normal pair setting of item (viii). If $Q: I \times \Sigma \to \MS$ on an open interval $I \subset \Reals$ is such that $Q(t,\cdot)$ and its ultraweak derivative $\dot Q(t,\cdot)$ (Definition~\ref{def:uwdiff}) lie in $\Bb$ with locally uniform bounds in $t$, then $t \mapsto \int_\Sigma Q(t,x)\,d\E(x)$ is ultraweakly differentiable on $I$ in $\M$ and
    \[
        \frac{d}{dt}\int_\Sigma Q(t,x)\,d\E(x) = \int_\Sigma \dot Q(t,x)\,d\E(x).
    \]

\item \emph{Quantum Fubini theorem} (Thm.~\ref{thm:Fubini}). If $\MS, \M_1, \M_2 \subset \M$ form a jointly normal triple of $\mathrm{W}^*$-subalgebras with separable preduals, and $\E_j: \F_j \to \Eff(\M_j)$ are POVMs for $j = 1,2$, then for any uniformly bounded jointly measurable $Q: \Sigma_1 \times \Sigma_2 \to \MS$ we have
    \[
        \int_{\Sigma_1}\!\Big(\int_{\Sigma_2} Q \, d\E_2\Big) d\E_1 = \int_{\Sigma_1 \times \Sigma_2} Q \, d(\E_1\E_2) = \int_{\Sigma_2}\!\Big(\int_{\Sigma_1} Q \, d\E_1\Big) d\E_2
    \]
    in $\MS \vee \M_1 \vee \M_2 \subset \M$.
\end{enumerate}
\end{theorem}

\begin{proof}
Items (i) and (ii) are Theorem~\ref{thm:intop}; universal integrability of $\Bb$ is Theorem~\ref{thm:intop}. Item (iii): $\Bb$ is a $\mathrm{C}^*$-algebra by Proposition~\ref{prop:BbisCstar}; maximality is Proposition~\ref{prop:integrability}; the $\Linf$ description is Lemma~\ref{lem:canonical} and Theorem~\ref{thm:Linfstructure}. Item (iv) is Theorem~\ref{thm:integrationmaps2}, which establishes positivity, unitality, adjoint-preservation, pointwise-normality, norm recovery, and multiplicativity directly, and uses Corollary~\ref{cor:CP} for complete positivity. For item (v), the quotient, complete-positivity and faithfulness arguments in the proof of Theorem~\ref{thm:descent} require only separability of $\MSp$; separability of $\MRp$ is used for the von Neumann algebra structure and normality. Item (vi) is Theorem~\ref{thm:int-categorical}. Item (vii) is Theorem~\ref{thm:stinespring}. Item (viii) is Theorem~\ref{thm:embeddedprops} together with Corollary~\ref{cor:embeddeddescent}. Item (ix) is Proposition~\ref{prop:Leibniz}, and (x) is Theorem~\ref{thm:Fubini}.
\end{proof}

\section{Application to quantum reference frames}\label{sec:QRFapplication}

The main motivation for developing the integration theory for operator-valued functions and measures is to express observables on quantum systems relative to quantum reference frames by symmetry-constrained integration. This follows the operational approach of \cite{carette_operational_2025,glowacki_relativization_2024,loveridge_relativity_2017,loveridge_symmetry_2018,miyadera_approximating_2016} where the \emph{relativization map} constructs invariant observables on composite systems from equivariant data. 

To make this connection, in this section we enrich the minimal setup used so far by including symmetries: we will have a representation of a group on all the quantum systems considered, and on the sample space of the POVM. To be precise, throughout this section let $\MS,\MR$ be jointly normal $\mathrm{W}^*$-subalgebras of $\M$ with separable preduals, and let $G$ be a topological group acting on $(\Sigma,\F)$ by measurable bijections $L_g$ and on $\MS\vee\MR$ by normal $*$-automorphisms $\alpha_g$ preserving both subsystem algebras, with restrictions $\alpha_g^\S$ and $\alpha_g^\R$. We additionally assume throughout that the system action $\alpha^\S$ is point-ultraweakly continuous, i.e., $g\mapsto\alpha_g^\S(a)$ is ultraweakly continuous for every $a\in\MS$.\footnote{For second countable, locally compact Hausdorff groups, this is equivalent to joint continuity of $G\times(\MS)_1\to(\MS)_1$, $(g,a)\mapsto\alpha_g^\S(a)$, with the ultraweak topology on the unit ball, and hence on every norm-bounded ball. This follows from norm continuity of the predual orbits \cite[Thm.~III.3.2.2]{blackadar_operator_2006}. No continuity in the group parameter is assumed for the frame action $\alpha^\R$.} Additional assumptions on the topology of $G$ are stated when needed.

It turns out that the integration map is equivariant---with respect to the inherited $G$-action on the space of functions---whenever the POVM defining it is covariant, meaning that the effects transform coherently with the sample space action. We take $\E_\R$ to denote a normalized covariant POVM, so $\E_\R(g.X)=\alpha_g^\R(\E_\R(X))$, and we refer to it as a \emph{frame observable}; the triple $(\MR,\alpha^\R,\E_\R)$ is understood to form a \emph{quantum reference frame} (QRF). Integration with respect to a frame observable then outputs invariant operators in $\M$ whenever the domain is restricted to the invariant function algebra. We show the basic facts about this setup and how it provides a natural extension of the existing framework, leaving the full development of algebraic operational QRF theory on general $G$-spaces to forthcoming work.\footnote{Quantum reference frames form a rapidly moving research field; since the first version of this work was made available, several related contributions have appeared, extending quantum reference frame frameworks in complementary directions. We discuss them, and how they relate to this work, briefly in Sec. \ref{sec:conclusion}.}

\subsection{Relativization and restriction on measurable \texorpdfstring{$G$}{G}-spaces}\label{sec:QRFdomain}

Let us now investigate the symmetry constrained setup in some detail. Notice first that covariance of the POVM implies that $\E_\R(g.X)=0$ iff $\E_\R(X)=0$, for all $g\in G,\ X\in\F$. Thus the action of $G$ preserves $\E_\R$-null sets and the $G$-action
\[
    \beta_g[f]
    :=[x\mapsto\alpha_g^\S(f(g^{-1}.x))]
\]
is well-defined on $L^\infty_{\E_\R}(\Sigma,\MS)$. It is a $*$-automorphism with inverse $\beta_{g^{-1}}$, and hence is normal. The fixed-point algebra
\[
 L^\infty_{\E_\R}(\Sigma,\MS)^G
 :=\{[f]:\beta_g[f]=[f]\text{ for every }g\in G\}
\]
is a $\mathrm{W}^*$-algebra consisting precisely of classes satisfying
$f(g.x)=\alpha_g^\S(f(x))$ $\E_\R$-almost everywhere for each fixed $g\in G$. Write $\mu:\MSR\to\MS\vee\MR$ for the normal multiplication map supplied by joint normality. By functoriality and covariance, the embedded integration map is \emph{equivariant}. Indeed, for all $g\in G$ we have
\begin{align}\label{eq:QRFintertwining}
	\int_\Sigma \beta_g[f]\, d\E_\R =\int_\Sigma \alpha_g^\S \circ f \circ L_{g^{-1}}\,d\E_\R(x)
 	= \int_\Sigma \alpha_g^\S \circ f\,d\bigl(\E_\R \circ L_{g^{-1}}^{-1}\bigr)
 	= \int_\Sigma \alpha_g^\S \circ f\,d\bigl(\alpha^\R_g \circ \E_\R\bigr)
 	= \alpha_g\left(\int_\Sigma f\,d\E_\R\right),
\end{align}
where we have used that $g.X = \{y \in \Sigma \,|\, y=L_g(x), x \in X\}= L_{g^{-1}}^{-1}(X)$.

\begin{definition}[Relativization]\label{def:QRFrelativization}
The \emph{relativization map} associated with the frame observable $\E_\R$ is
\begin{equation}\label{def:yen}
 \yen^\R
 :=\left.\int_\Sigma d\E_\R\right|_{L^\infty_{\E_\R}(\Sigma,\MS)^G}
 :L^\infty_{\E_\R}(\Sigma,\MS)^G\longrightarrow(\MS\vee\MR)^G.
\end{equation}
The space of relative operators is defined by
\[
 \MS^\R:=\overline{\operatorname{Im}(\yen^\R)}^{\,\rm uw}
 \subseteq(\MS\vee\MR)^G.
\]
\end{definition}

The properties of the relativization maps follow directly from the results established in previous sections. Notice that these statements require no continuity in the group parameter of the frame representation $\alpha^\R$.

\begin{proposition}[Basic properties of relativization]\label{prop:QRFproperties}
The map $\yen^\R$ is normal, unital and completely positive, and
$\MS^\R$ is an ultraweakly closed operator system in $(\MS\vee\MR)^G$.
Moreover:
\begin{enumerate}[label=(\roman*)]
\item If $\mu$ is injective, then $\yen^\R$ is faithful.
If, in addition, $\E_\R$ is localizable, then $\yen^\R$ is isometric.
\item If $\E_\R$ is sharp, then $\yen^\R$ is a normal $*$-homomorphism
onto $\MS^\R=\operatorname{Im}(\yen^\R)$.
If, in addition, $\mu$ is injective, then
\[
 \yen^\R:L^\infty_{\E_\R}(\Sigma,\MS)^G\xrightarrow{\sim}\MS^\R
\]
is a $*$-isomorphism between $\mathrm{W}^*$-algebras.
\end{enumerate}
\end{proposition}

\begin{proof}
By Corollary~\ref{cor:embeddeddescent}, the embedded integration map
is normal, unital and completely positive, with the faithfulness,
isometry and multiplicativity properties stated above.
These properties pass to its restriction to the fixed-point algebra,
and \eqref{eq:QRFintertwining} ensures that its values are invariant.
Its range is a unital self-adjoint linear subspace, so its ultraweak
closure is an operator system in $(\MS\vee\MR)^G$.
When $\E_\R$ is sharp, the restricted map is a normal $*$-homomorphism
and hence has ultraweakly closed range. If $\mu$ is also injective,
faithfulness makes this $*$-homomorphism injective, proving the final assertion.
\end{proof}

There is another construction playing an important role in the operational QRF framework: for a normal frame state $\omega\in\S(\MR)$, let
\[
 \Gamma_\omega:=\id_{\MS}\bar\otimes\hat\omega:\MSR\longrightarrow\MS.
\]
These are called \emph{restriction maps} \cite{carette_operational_2025,loveridge_symmetry_2018}.
The characterizing identity for integration then gives, for every $\rho\in\S(\MS)$,
\begin{equation}\label{eq:QRFslice}
\hat\rho\left(\Gamma_\omega\left(\int_\Sigma f\otimes d\E_\R\right)\right)
 =\int_\Sigma f_\rho\,d(\E_\R)_\omega \in \Cn.
\end{equation}
Thus preparing the frame in $\omega$ yields the average with respect to the scalar measure $(\E_\R)_\omega$ of the system-valued integrand. The following extends the localization argument of \cite[Sec.~3.5]{carette_operational_2025} to the present algebraic setting and to arbitrary metrizable outcome spaces; no group action is needed for the proposition.

\begin{proposition}[Localization]\label{prop:QRFlocalization}
Suppose $\Sigma$ is metrizable, $\F$ contains its Borel sets, and $\E_\R$ is localizable. Fix $x_0\in\Sigma$ such that $\E_\R(U)\ne0$ for every open neighbourhood $U$ of $x_0$. There is a sequence of normal states $(\omega_n)$ on $\MR$, depending only on $x_0$ and $\E_\R$, such that for every $f\in\Bb$ ultraweakly continuous at $x_0$,
\begin{equation}\label{eq:QRFlocalization}
 \Gamma_{\omega_n}\left(\int_\Sigma f\otimes d\E_\R\right)
 \xrightarrow{\rm uw} f(x_0).
\end{equation}
The value $f(x_0)$ depends only on $[f]$ among such representatives.
\end{proposition}

\begin{proof}
Choose a compatible metric $d$ and put $U_n=\{x:d(x,x_0)<1/n\}$. Since $\|\E_\R(U_n)\|=1$ and normal states norm positive elements, choose $\omega_n$ with $\omega_n(\E_\R(U_n))>1-1/n$. Then
\[
 \left|\int_\Sigma f_\rho\,d(\E_\R)_{\omega_n}-f_\rho(x_0)\right|
 \leq\sup_{x\in U_n}|f_\rho(x)-f_\rho(x_0)|+\frac{2\|f_\rho\|_\infty}{n}
 \longrightarrow0.
\]
Equation~\eqref{eq:QRFslice} proves the assertion. If $f$ and $k$ represent the same class and are ultraweakly continuous at $x_0$, their integrals coincide, and applying \eqref{eq:QRFlocalization} to both gives $f(x_0)=k(x_0)$.
\end{proof}

In the context of operational QRFs, this result ensures that the non-relational description of system observables is a good approximation for relativized observables conditioned on well-localized reference states.

\subsection{QRFs on homogeneous spaces}\label{sec:QRFspaces}

We now discuss the case when the action of $G$ on $\Sigma$ is transitive and we have a homogeneous space $\Sigma\cong G/H$ for some stabilizer subgroup $H$. Assuming $G$ to be locally compact, second countable and Hausdorff and $H$ closed, we recover the previously developed construction \cite{fewster_quantum_2024} generalized to arbitrary von Neumann frame algebras and non-compact stabilizer subgroups. Let us emphasize that the equation \eqref{def:yen} defines the relativization map for arbitrary topological groups, representations by $*$-automorphisms with $\alpha^\S$ additionally point-ultraweakly continuous, and actions by measurable bijections; none of the additional assumptions imposed in this subsection are needed for that definition, even in the transitive case.

\begin{proposition}[Homogeneous-space domain]\label{prop:QRFhomogeneous}
Let $H\subseteq G$ be a subgroup, and take $\Sigma=G/H$ with its quotient topology and Borel structure and denote by $\MS^H$ the $H$-fixed point subalgebra of $\MS$. Then
\[
 j:\MS^H\longrightarrow L^\infty_{\E_\R}(\Sigma,\MS)^G,\qquad
 a\longmapsto[f_a],\quad f_a(gH)=\alpha_g^\S(a),
\]
is an injective normal unital $*$-homomorphism. Each class in the image of $j$ has a unique ultraweakly continuous equivariant representative, and $j^{-1}$ evaluates this representative at $H$. If, in addition, $G$ is locally compact, second countable and Hausdorff and $H$ is closed, then $j$ is surjective and hence a normal $*$-isomorphism.
\end{proposition}

\begin{proof}
For \(a\in\MS^H\), the orbit map \(g\mapsto\alpha_g^\S(a)\) is
ultraweakly continuous and constant on right \(H\)-cosets. Hence, by
the definition of the quotient topology, it descends to the
ultraweakly continuous equivariant function
\(f_a(gH)=\alpha_g^\S(a)\). Since the \(\alpha_g^\S\) are
\(*\)-automorphisms, the assignment \(a\mapsto f_a\) is linear,
multiplicative, unital and \(*\)-preserving. Thus \(j\) is a unital
\(*\)-homomorphism.

It is injective. Indeed, if \(j(a)=0\), then \(f_a=0\) outside an
\(\E_\R\)-null set. Since \(\E_\R(\Sigma)=\algunit\), its complement is
nonempty, and \(f_a(gH)=\alpha_g^\S(a)=0\) at any point of that
complement implies \(a=0\).

To prove normality, first let \(0\leq a_n\uparrow a\) in \(\MS^H\).
Normality of each \(\alpha_g^\S\) gives
\(f_{a_n}(x)\uparrow f_a(x)\) for every \(x\in G/H\), and hence
\(j(a_n)\uparrow j(a)\). Thus \(j\) is sequentially normal.

Since \((\MS)_*\) is separable, \(\MS\) admits a faithful normal state
\(\varphi\), whose restriction to \(\MS^H\) is faithful and normal.
Now let \(0\leq a_i\uparrow a\) be an arbitrary increasing net in
\(\MS^H\). We may choose an increasing sequence of indices \(i_n\)
such that \(\varphi(a-a_{i_n})\to0\). If
\(b=\sup_n a_{i_n}\), then normality and faithfulness of \(\varphi\)
give \(b=a\). Sequential normality therefore yields
\(j(a)=\sup_nj(a_{i_n})=\sup_i j(a_i)\), so \(j\) is normal.

Next let \(F:G/H\to\MS\) be ultraweakly continuous and equivariant,
and set \(b=F(H)\). Since \(hH=H\) for \(h\in H\), equivariance gives
\(\alpha_h^\S(b)=b\), so \(b\in\MS^H\); moreover,
\(F(gH)=\alpha_g^\S(b)=f_b(gH)\). Thus every ultraweakly continuous
equivariant function is an orbit function. If \(F\) represents
\(j(a)\), then \(j(b)=[F]=j(a)\), so injectivity gives \(b=a\) and
hence \(F=f_a\). This proves uniqueness and shows that \(j^{-1}\) is
evaluation at \(H\) on these representatives.

Assume now that \(G\) is locally compact, second countable and
Hausdorff and that \(H\) is closed. Let \(\mu\) be a left Haar measure
on \(G\), which is \(\sigma\)-finite by second countability. Choose a
faithful normal state \(\omega_0\) on \(\MR\) and set
\(\lambda=\omega_0\circ\E_\R\). Then \(\lambda\) is a probability
measure with exactly the same null sets as \(\E_\R\). Take \([f]\in L^\infty_{\E_\R}(\Sigma,\MS)^G\) and choose a bounded
ultraweakly measurable representative \(f\). For every fixed
\(k\in G\), invariance of \([f]\) gives
\begin{equation}
f(k.x)=\alpha_k^\S(f(x))
\qquad\text{for \(\lambda\)-a.e.\ }x.
\label{eq:QRF-ae-equivariance}
\end{equation}
The exceptional set may depend on \(k\), so we use Tonelli to reverse
the order of the almost-everywhere quantifiers. To this end, choose a countable family \((\rho_n)\subseteq(\MS)_*\) separating the
points of \(\MS\). For each \(n\), the map
\((k,x)\mapsto\rho_n(f(k.x))\) is jointly measurable, since the action
\(G\times G/H\to G/H\) is continuous. The map
\((k,x)\mapsto\rho_n(\alpha_k^\S(f(x)))\) is measurable in \(x\) for
fixed \(k\) and continuous in \(k\) for fixed \(x\); since \(G\) is
separable and metrizable, it is jointly measurable by the
Carathéodory joint-measurability theorem
\cite[Lem.~4.51]{aliprantis_border_2006}. Hence the set
\begin{equation}
N=
\bigcup_{n=1}^\infty
\left\{
(k,x)\in G\times G/H:
\rho_n(f(k.x))
\neq
\rho_n\bigl(\alpha_k^\S(f(x))\bigr)
\right\}
\label{eq:QRF-defect-set}
\end{equation}
is measurable. Since the \(\rho_n\) separate points, \(N\) is exactly
the set where equivariance fails. By
\eqref{eq:QRF-ae-equivariance}, every section
\(N_k=\{x:(k,x)\in N\}\) is \(\lambda\)-null. Tonelli's theorem
\cite[Thm.~2.37(a)]{folland_real_1999} therefore gives a
\(\lambda\)-conull set \(\Sigma_0\subseteq G/H\) such that, for every
\(x\in\Sigma_0\),
\begin{equation}
f(g.x)=\alpha_g^\S(f(x))
\qquad\text{for \(\mu\)-a.e.\ }g.
\label{eq:QRF-Haar-ae-equivariance}
\end{equation}

We claim that equivariance is exact between points of \(\Sigma_0\).
Indeed, let \(x,y\in\Sigma_0\) and \(y=k.x\). Applying
\eqref{eq:QRF-Haar-ae-equivariance} at \(x\) and \(y\), and using that
right translations preserve Haar-null sets, gives
\(\alpha_g^\S(f(y))=\alpha_{gk}^\S(f(x))\) for \(\mu\)-a.e.\ \(g\).
Applying \(\alpha_{g^{-1}}^\S\) therefore yields
\begin{equation}
f(y)=\alpha_k^\S(f(x)).
\label{eq:QRF-exact-equivariance}
\end{equation}

Choose \(x_0=tH\in\Sigma_0\) and set
\(a=\alpha_{t^{-1}}^\S(f(x_0))\). For \(h\in H\), the element
\(tht^{-1}\) fixes \(x_0\), so
\eqref{eq:QRF-exact-equivariance} gives
\(f(x_0)=\alpha_{tht^{-1}}^\S(f(x_0))\). Applying
\(\alpha_{t^{-1}}^\S\) shows that \(\alpha_h^\S(a)=a\), hence
\(a\in\MS^H\).

Finally, if \(x=gH\in\Sigma_0\), then
\(x=(gt^{-1}).x_0\), and \eqref{eq:QRF-exact-equivariance} gives
\(f(x)=\alpha_{gt^{-1}}^\S(f(x_0))=\alpha_g^\S(a)=f_a(x)\).
Thus \(f=f_a\) on a \(\lambda\)-conull, equivalently
\(\E_\R\)-conull, set. Hence \([f]=j(a)\), proving surjectivity.
Therefore \(j\) is a normal \(*\)-isomorphism.
\end{proof}

\noindent In the tensor-product setting we then have:
\[
 (\yen^\R\circ j): \MS^H \ni a \longmapsto \int_{G/H}\alpha^\S_g(a)\otimes d\E_\R(gH) \in (\MSR)^G.
\]
Consistency of our construction with the relativization maps defined previously is confirmed by uniqueness of integral operators, our characterizing formula and \cite[Prop.~4.45]{fewster_quantum_2024}.

\subsection{QRFs on trivial principal bundles}\label{ex:QRFbundle}

We now treat the general case of frame observables on trivial principal bundles. Generalization to non-trivial bundles and further investigations will be carried out in the forthcoming work.

Let $M$ be a Hausdorff topological space, to be thought of as spacetime.
Assume, for simplicity, that $G$ is second countable, locally compact and Hausdorff.
Let $P=M\times G\to M$ be a trivial
topological principal $G$-bundle,
with left action $h.(x,g)=(x,hg)$. We equip $P$ with its product topology and Borel
$\sigma$-algebra. Since $G$ is second countable,
\[
 \mathcal B(P)=\mathcal B(M)\otimes\mathcal B(G).
\]
We restrict to the composite systems setup $\M=\MSR$ with action
$\alpha^\S\bar\otimes\alpha^\R$. The frame observable induces the \emph{marginal} POVM
\[
 \mathsf F_\R:\mathcal B(M)\to\MR,
 \qquad
 \mathsf F_\R(B):=\E_\R(B\times G),
\]
to be thought of as the spacetime localization observable. The null sets on $M$ then are determined by the frame:
\[
 \N_{\mathsf F_\R}
 =
 \{B\in\mathcal B(M):\E_\R(B\times G)=0\},
\]
and functions on $M$ will be identified when they agree outside
an $\mathsf F_\R$-null set. We now show that the domain of $\yen$---the $\mathrm{W}^*$-algebra of classes of equivariant functions on the sample space of the frame observable---is again isomorphic to a different algebra, this time it is the $\mathrm{W}^*$-algebra of classes of functions on the base space $M$.

\begin{proposition}[Trivial-bundle correspondence]\label{prop:QRFbundle}
The map
\[
 \begin{aligned}
 j_P:L^\infty_{\mathsf F_\R}(M,\MS)
 &\xrightarrow{\sim} L^\infty_{\E_\R}(P,\MS)^G,\\
 [\hat\phi]&\longmapsto
 [(x,g)\mapsto\alpha_g^\S(\hat\phi(x))]
 \end{aligned}
\]
is a $*$-isomorphism between $\mathrm{W}^*$-algebras.
\end{proposition}

\begin{proof}
Both algebras in the statement are $\mathrm{W}^*$-algebras by
Theorem~\ref{thm:Linfstructure} and the ultraweak closedness of the fixed-point
algebra under the normal automorphisms $\beta_h$.
The chosen trivialization identifies bounded ultraweakly
measurable functions on the base with equivariant such functions on the bundle:
\[
 \hat\phi:M\to\MS
 \quad\longleftrightarrow\quad
 f_{\hat\phi}(x,g):=\alpha_g^\S(\hat\phi(x)),
 \qquad \hat\phi_f(x):=f(x,e).
\]
Indeed, equivariance gives
$f(x,g)=\alpha_g^\S(f(x,e))$.
Conversely, $f_{\hat\phi}$ is equivariant and
$\|f_{\hat\phi}(x,g)\|=\|\hat\phi(x)\|$ since $*$-automorphisms on $\mathrm{C}^*$-algebras are isometric.
For each $\rho\in\MSp$, the scalar function
$(x,g)\mapsto\rho(\alpha_g^\S(\hat\phi(x)))$ is measurable in $x$
and continuous in $g$, hence jointly measurable by
\cite[Lem.~4.51]{aliprantis_border_2006}.
The assignment is linear and unital
 and since each $\alpha_g^\S$ is a $*$-automorphism we have
$f_{\hat\phi^*}=f_{\hat\phi}^*$ and $f_{\hat\phi\hat\psi}=f_{\hat\phi}f_{\hat\psi}$.
The correspondence also holds modulo null sets:
every invariant class has a representative $f_{\hat\phi}$. To prove this, choose a faithful normal state $\omega_0$
on $\MR$ and introduce the auxiliary probability measures
\[
 \lambda=\omega_0\circ\E_\R,
 \qquad
 \nu=\omega_0\circ\mathsf F_\R,
 \qquad
 \nu(B)=\lambda(B\times G).
\]
Faithfulness ensures that $\lambda$ and $\nu$ have exactly
the $\E_\R$-null and $\mathsf F_\R$-null sets, respectively.
Let $m$ be a left Haar measure on $G$. It is $\sigma$-finite since $G$ is second countable.
Covariance makes $\lambda$ quasi-invariant, and for every
Borel $N\subseteq P$ we have
\[
 \E_\R(N)=0
 \quad\Longleftrightarrow\quad
 \lambda(N)=0
 \quad\Longleftrightarrow\quad
 (\nu\otimes m)(N)=0.
\]
Indeed, write $N_x=\{g\in G:(x,g)\in N\}$.
Tonelli's theorem \cite[Thm.~2.37(a)]{folland_real_1999} gives
\[
 \int_G\lambda(h^{-1}.N)\,dm(h)
 =
 \int_{M\times G}m(N_xg^{-1})\,d\lambda(x,g).
\]
By quasi-invariance, the left-hand side vanishes exactly
when $\lambda(N)=0$.
Since right translations preserve Haar-null sets,
the right-hand side vanishes exactly when $m(N_x)=0$
for $\nu$-almost every $x$.
By Tonelli's theorem again, this is equivalent to
$(\nu\otimes m)(N)=0$. For a bounded ultraweakly measurable representative
$f$ of an invariant class, put
\[
 F(x,g)=\alpha_{g^{-1}}^\S(f(x,g)).
\]
This is ultraweakly measurable.\footnote{
For each $\rho\in\MSp$, apply the same joint-measurability
criterion to
$(k,x,g)\mapsto\rho(\alpha_{k^{-1}}^\S(f(x,g)))$,
which is continuous in $k$ and measurable in $(x,g)$,
and then set $k=g$.}
Invariance gives $F(x,hg)=F(x,g)$ for
$(\nu\otimes m)$-almost every $(x,g)$, for each fixed~$h$.
Tonelli's theorem \cite[Thm.~2.37(a)]{folland_real_1999},
applied to the indicators of the sets where
$\rho(F(x,hg))\neq\rho(F(x,g))$, for a countable separating
family of $\rho\in\MSp$,
shows that for $\nu$-almost every $x$ and Haar-almost every $g$,
the equality holds for Haar-almost every $h$. For each such $x$,
there is $g_0\in G$ such that
\[
 F(x,hg_0)=F(x,g_0)
 \qquad\text{for Haar-a.e.\ }h.
\]
Right translations preserve Haar-null sets, so
for each such $x$, $F(x,g)$
is constant for Haar-almost every $g$.
Choose a compact neighborhood $K$ of $e$ and define
\[
 \hat\phi(x)=\frac{1}{m(K)}\int_K F(x,g)\,dm(g),
\]
with the integral taken ultraweakly.
Explicitly, for $\rho\in\MSp$,
\[
 \rho(\hat\phi(x))=\frac{1}{m(K)}\int_K\rho(F(x,g))\,dm(g).
\]
The right-hand side is a bounded linear functional of $\rho$, of norm at most
$\|f\|_\infty$, and hence defines an element of $(\MSp)^*=\MS$.
Its measurability in $x$ follows from scalar integration with a parameter.
Since $0<m(K)<\infty$, this defines a bounded ultraweakly
measurable function $\hat\phi$, and
$f=f_{\hat\phi}$ almost everywhere.
Moreover, $[f_{\hat\phi}]=[f_{\hat\phi'}]$ exactly when
$\hat\phi=\hat\phi'$
outside an $\mathsf F_\R$-null set, since
\[
 \|f_{\hat\phi}(x,g)-f_{\hat\phi'}(x,g)\|
 =
 \|\hat\phi(x)-\hat\phi'(x)\|.
\]
Indeed, $D:=\{x:\hat\phi(x)\ne\hat\phi'(x)\}$ is Borel,
by separability of $\MSp$, and
$\{(x,g):f_{\hat\phi}(x,g)\ne f_{\hat\phi'}(x,g)\}=D\times G$.
Since $\E_\R(D\times G)=\mathsf F_\R(D)$, the two lifted functions agree
$\E_\R$-almost everywhere exactly when the base functions agree
$\mathsf F_\R$-almost everywhere. Thus the map on equivalence classes is
well-defined and injective; the preceding construction proves surjectivity. Since the algebraic operations are pointwise, the
correspondence therefore induces a $*$-isomorphism
\[
 L^\infty_{\mathsf F_\R}(M,\MS)
 \cong
 L^\infty_{\E_\R}(P,\MS)^G,
 \qquad
 [\hat\phi]\longmapsto[f_{\hat\phi}].
\]
This is the map $j_P$.
As a $*$-isomorphism of $\mathrm{W}^*$-algebras, it is normal.

\end{proof}

Under this identification, we write $\yen^\R$ also for
$\yen^\R\circ j_P$. Thus relativization takes functions
on spacetime, modulo $\mathsf F_\R$-null sets, as inputs:

\[
 \yen^\R(\hat\phi)
 =
 \int_{M\times G}
 \alpha_g^\S(\hat\phi(x))\otimes d\E_\R(x,g).
\]
This is precisely the setup envisioned in \cite{glowacki_towards_2024}, now made rigorous by the use of integration theory developed in this work.

As a canonical example, consider the sharp frame obtained by choosing
a nonzero $\sigma$-finite Borel measure $\nu$ on $M$
such that $L^2(M,\nu)$ is separable\footnote{
This holds, in particular, when $M$ is second countable.
Since $L^2(G,m)$ is separable, $\hir$ is then separable and $\MR$
has separable predual.}, and taking

\[
 \hir=L^2(P,\nu\otimes m),\qquad
 \MR=B(\hir),\qquad
 \E_\R(Y)=M_{\chi_Y},
\]
where $Y\subseteq P$ is Borel and $M_b$ denotes multiplication
by $b$. The strongly continuous unitary representation
\[
 (U_\R(h)\psi)(x,g)=\psi(x,h^{-1}g)
\]
implements $\alpha_h^\R(A)=U_\R(h)AU_\R(h)^*$.
Thus $\E_\R$ is a normalized covariant PVM:
\[
 \E_\R(P)=\id,\qquad
 U_\R(h)\E_\R(Y)U_\R(h)^*=\E_\R(h.Y).
\]
Its null sets are exactly the $(\nu\otimes m)$-null sets.
Here $\mathsf F_\R(U)=M_{\chi_{U\times G}}$, so the
$\mathsf F_\R$-null sets are exactly the $\nu$-null sets.

Thus $\yen^\R$ represents the algebra of classes of functions on the base
in the invariant algebra $(\MSR)^G$. Since $\E_\R$ is sharp,
Theorem~\ref{thm:descent} makes this realization a faithful normal
$*$-homomorphism with ultraweakly closed range. Hence
\[
 \yen^\R:L^\infty_\nu(M,\MS)
 \xrightarrow{\sim}\MS^\R.
\]
For a unit vector $\psi\in\hir$, let
$\omega_\psi(B)=\langle\psi,B\psi\rangle$ be the corresponding pure frame state.
Equation~\eqref{eq:QRFslice} gives
\[
 \Gamma_{\omega_\psi}\bigl(\yen^\R([\hat\phi])\bigr)
 =\int_{M\times G}|\psi(x,g)|^2\,
 \alpha_g^\S(\hat\phi(x))\,d\nu(x)\,dm(g).
\]
Thus restriction averages the system field over the frame's distribution
of base points and orientations. This formula connects this construction with the frame-bundle approach of Vanzella and Butterfield \cite[Sec.~II, Eqs.~(7)--(8)]{vanzella_frame_2024}, who average observables depending on an observer's tetrad with weights given by squared frame amplitudes. On a trivialized orthonormal frame bundle, globally square-integrable amplitudes\footnote{Their definition requires square-integrability along each fibre, without generally requiring integrability over spacetime. The comparison here concerns amplitudes satisfying this additional condition. This condition holds if $\nu(M)<\infty$ and the fibre $L^2$-norms are essentially bounded, in particular for a compact base with a Radon measure and continuous fibre $L^2$-norms.} define, after normalization, pure states of the sharp frame above. For bounded equivariant operator fields, this form of averaging is obtained by restricting the joint relative observable to the chosen frame state. Relativization therefore constructs the relative observable independently of the amplitude, which enters through the frame preparation.

On a base of finite measure, amplitudes with essentially bounded fibre $L^2$-norms are globally square-integrable. For such amplitudes and bounded equivariant operator fields, the normalized averaging prescription is recovered through pure-state restriction of a sharp frame. The present construction extends this mechanism to general covariant POVMs, arbitrary normal frame states, and system and frame von Neumann algebras with separable preduals, on general measurable spaces carrying a group action.

For a concrete relative observable, take also
$\his=L^2(P,\nu\otimes m)$ and $\MS=B(\his)$, with the same left regular action.
Explicitly, for $k\in G$ and $\xi\in\his$,
\[
 (U_\S(k)\xi)(y,h)=\xi(y,k^{-1}h),
 \qquad \alpha_k^\S(A)=U_\S(k)AU_\S(k)^*.
\]
Let $a(x,y,h)$ be a bounded scalar function measurable in the product
$\sigma$-algebra, and let $\hat\phi_a(x)$ multiply system wavefunctions by
$a(x,y,h)$. Explicitly,
\[
 (\hat\phi_a(x)\xi)(y,h)=a(x,y,h)\xi(y,h),
 \qquad \xi\in\his.
\]
Here $x$ labels the field, while $(y,h)$ are the system configuration
variables. Boundedness and measurability of $a$ ensure that
$x\mapsto\hat\phi_a(x)$ is a bounded ultraweakly measurable $\MS$-valued field.
Conjugation by $U_\S(g)$ replaces $h$ with $g^{-1}h$ in the multiplier.
On $\his\otimes\hir$ we then have
\[
 \bigl(\yen^\R([\hat\phi_a])\Psi\bigr)(y,h;x,g)
 =a(x,y,g^{-1}h)\Psi(y,h;x,g)
\]
almost everywhere. The group argument $g^{-1}h$ describes the system's
configuration relative to the frame: it is unchanged by simultaneous
left translation of $g$ and $h$. For an additive translation group it
becomes the familiar difference $h-g$. Here the group acts only on the
fibre coordinates; the comparison uses the chosen trivialization,
while the base points $x,y$ are left unchanged.

\section{Concluding remarks}\label{sec:conclusion}

This paper has developed a theory of operator-valued integration in the $\mathrm{W}^*$-algebraic setting. For a jointly normal pair with separable preduals, the descended embedded integration map is a normal unital CP map
\[
    \int_\Sigma d\E: \Linf \;\longrightarrow\; \MS \vee \MR, \qquad [f] \longmapsto \mu\Big(\int_\Sigma f \otimes d\E\Big),
\]
which is faithful when $\mu$ is injective. The theory lifts the change-of-variables formula, simple-function approximation, monotone convergence, the Leibniz rule, and the Fubini theorem to the $\mathrm{W}^*$-algebraic setting.

The direct construction characterizes integral operators by scalar integration against product normal states; equivalently, the integration map is determined by its simple-function rule and continuity under uniformly bounded pointwise ultraweak convergence of sequences. The tensor-product description identifies its descent with the amplification of scalar POVM integration. The former construction assumes separability only of $\MSp$, while the von Neumann algebra and tensor-product descriptions developed here assume separability of both preduals.

A structural question is whether the construction extends beyond $\mathrm{W}^*$-algebras to general dual order-unit spaces equipped with appropriate positive tensor products \cite{kavruk_tensor_2011}. The $\mathrm{W}^*$-algebraic framework supplies not only the dual order structure needed for positivity of the integration map but also the Hilbert-space realization used to establish the key norm bound (Lemma~\ref{lem:semiip}) and the Naimark/Stinespring factorization. Since abstract operator systems admit Hilbert-space theoretic representations, we believe the $\mathrm{W}^*$-integration theory can be lifted to this setting. Such a generalization would extend the reach of the theory to a wider class of convex-operational physical models.

From the categorical perspective, Section~\ref{sec:monoidal} suggests studying integration as amplification $\id_\S\otimes\Phi$ in other ordered monoidal settings. A realization of the tensor-product domain by measurable functions, and hypotheses ensuring the required continuity, would have to be established separately.

The principal application is the extension of relativization to covariant POVMs on arbitrary measurable $G$-sets. The realization of the relativization maps as symmetry-constrained integration gives a mathematical basis for operational QRFs beyond transitive action sample spaces. In particular, treating principal bundles and general von Neumann algebras provides integration machinery for applications to relativistic gauge theories. The present paper establishes this machinery; the associated operational framework, relative states and frame transformations will be developed in forthcoming work.

Recent contributions extend relativization in complementary directions. Ludescher, Mekonnen, Galley and M\"uller \cite[Sec.~II.2]{ludescher_beyond_2026} reconstruct and generalize the perspective-neutral framework using covariant quantum instruments as quantum reference frames. Their main mathematical development assumes compact Lie groups and finite-dimensional Hilbert spaces, but extends relativization to more general types of frames that also need not correspond to subalgebras. Our extension retains commuting, jointly normal system and frame algebras while broadening the algebraic and geometric setting to general von Neumann algebras and measurable $G$-spaces.

Fedida, Ibort and Mas-Dorca extend the operational approach from groups to continuous groupoids, explicitly using the present integration theory for their fibrewise relativization of operator fields against covariant POVMs \cite[Prop.~4.24]{fedida_groupoidal_2026}. Their constructions involving Atiyah groupoids of frame and gauge bundles \cite[Secs.~2.2 and~7]{fedida_groupoidal_2026} connect this more general symmetry structure to the principal-bundle applications that motivate our theory (Section~\ref{ex:QRFbundle}); establishing the precise relation between the two formulations remains for future work.

Brady \cite{brady} develops a further complementary approach using Hilbert modules and equivariant $\mathrm{C}^*$-correspondences. In this framework, relativization is formulated as a completely positive extension of the canonical map on invariant observables, and the standard operational Hilbert-space construction provides examples. In the present work, a covariant POVM determines relativization explicitly through integration on the algebra of measurable equivariant fields, with normality established under the stated hypotheses. Comparing the resulting relative observables with those obtained from Hilbert-module couplings is a natural direction for further investigation.

While the mathematical foundations underpinning various approaches to quantum reference frames are being extended and clarified, much remains to be done in the context of putting all of the frameworks on equal footing and understanding them as subsectors of a general relational quantum theory. We hope this contribution will be helpful in achieving this goal.

\section*{Statements and Declarations}

\paragraph*{Data availability.} No data were generated or analysed.

\paragraph*{Conflict of interest.} The authors have no competing interests to declare.

\paragraph*{AI assistance.} During preparation of this manuscript, generative AI tools were used at different stages to assist with editorial revision, literature and citation checks, and the drafting and checking of mathematical arguments.

\end{document}